\documentclass[11pt,a4paper]{article}
\usepackage[utf8]{inputenc}
\usepackage[english]{babel}
\usepackage{amsmath}
\usepackage{amsfonts}
\usepackage{amssymb}
\usepackage{graphicx}
\usepackage[round]{natbib}
\usepackage{booktabs}  % better tables
\usepackage{eurosym}

\usepackage{bm}% permette di usare grassetto nelle formule matematiche
\usepackage{bbm}% è simile ma permette di fare 1 in blackboard bold

\usepackage{amsthm} %theorems

\usepackage{mathtools} %allows cases definitions for equations
\usepackage{multirow}

\usepackage[footnotesize,bf]{caption} %pacchetto per controllare lo stile delle capition
\usepackage[section,above,below]{placeins} %allows to prevent floats being positioned after a \FloatBarrier command
\usepackage[inner=3cm,outer=3cm,top=1in,bottom=1in]{geometry}

\usepackage{hyperref}

\usepackage[affil-it]{authblk}

\newcommand{\diff}{\mathrm{d}}
\newcommand{\E}[1]{\mathbb{E}\left[ {#1}\right]} %expected value

\newcommand{\fil}{{\cal{F}}} %filtration symbol
\newcommand{\edur}{\psi}

\usepackage[inner=3cm,outer=3cm,top=1in,bottom=1in]{geometry}
\graphicspath{{./Images/}}

\newtheoremstyle{no_par}% name of the style to be used
  {\topsep}% measure of space to leave above the theorem. E.g.: 3pt topsep is the same as in plain style
  {\topsep}% measure of space to leave below the theorem. E.g.: 3pt
  {\itshape} % name of font to use in the body of the theorem
  {} % measure of space to indent
  {} % name of head font
  {} % punctuation between head and body
  {5pt plus 1pt minus 1pt} % space after theorem head; " " = normal interword space. \newline for a line break !DOESN'T work with enumerate or itemize! to get a newline after the theorem head when using "eumerate", use \mbox{} after \begin{theorem}. 
  {\thmname{\textbf{#1}}\thmnumber{ \textbf{#2.}}\thmnote{ \textit{#3}}} % Manually specify head

\theoremstyle{no_par}
\newtheorem{algo}{Algorithm}[section]

\theoremstyle{plain} %standard style
\newtheorem{prop}{Proposition}[section]

\theoremstyle{plain} %standard style

\newtheorem{teo}{Theorem}

\author[1]{Marcello Rambaldi}
\author[2]{Emmanuel Bacry}
\author[1]{Fabrizio Lillo}
\affil[1]{\footnotesize{Scuola Normale Superiore, Piazza dei Cavalieri 7, Pisa 56126, Italy}}
\affil[2]{Centre de Mathématiques Appliquées, CNRS, École Polytechnique, UMR 7641, 91128 Palaiseau, France}

\date{}

\title{The role of volume in order book dynamics: \\ a multivariate Hawkes process analysis}

\begin{document}
\maketitle

\begin{abstract}
We show that multivariate Hawkes processes coupled with the nonparametric estimation procedure first proposed in \cite{bacry2014} can be successfully used to study complex interactions between the time of arrival of orders and their size, observed in a limit order book market. We apply this methodology to high-frequency order book data of futures traded at EUREX. 
Specifically, we demonstrate how this approach is amenable not only to analyze interplay between different order types (market orders, limit orders, cancellations) but also to include other relevant quantities, such as the order size, into the analysis, showing also that simple models assuming the independence between volume and time are not suitable to describe the data. 
%We first explore the interplay between different trade sizes and we then extend our model to the whole first level of the order book. We are able to identify several interesting features of this market at very short time scales. 
\end{abstract}

\section{Introduction}

Modeling the order book is a complicated task. Indeed, even in the simplest setting, multiple types of orders (limit orders, market orders, cancellations) arrive on the market at random times, and each one has a label that specifies the quantity to be negotiated (called the volume) and the price. Several are the challenges for a mathematical description of this system. First, the arrival times of the orders are not well described by a Poisson process (see for example \citep{Chakraborti:2011em} and references therein). In fact, the time durations between events are not independent but display strong correlations. Moreover, also the sequence of volumes presents nontrivial correlations \citep{Gould:2013hg}. Furthermore, the interplay between them is relevant and need to be considered for a deep understanding of the system. Finally, the existence of different types of events that could influence each other increases further the complexity of the problem. 

In this paper we are interested in exploring the effects of order size on the order book dynamics. To frame the problem, let us first consider just one type of event, say the occurrence of a transaction. From a statistical perspective, the observation of transactions and the corresponding volumes represents the realization of a \emph{marked} point process ($t_i, v_i$). Since we are interested in the effects of order size on the dynamics and vice-versa, we cannot model the volume as an independent mark process, but instead we need to consider their interdependence. This two-way feedback mechanism between marks and the underlying point process represents a challenge for modelers.

Some models for this system have been proposed in the Autoregressive Conditional Duration framework \citep{engle&russell1998, Bauwens:2004gs, Pacurar:2008df}. These models describe the point process using a duration representation.
The $i$-th duration $d_i$ of a point process is defined as $d_i = t_i -t_{i-1}$. ACD models can be summarized as follows
\begin{equation}
\begin{split}
d_i &= \edur_i \epsilon_i, \;\;\;\;\; \epsilon_i \sim \mbox{i.i.d.}(1, \sigma^2_{\epsilon})\\
\edur_i &\equiv \E{d_i|\fil_i; \theta_d}
\end{split}
\end{equation}
where $\theta_d$ is a set of parameters to be estimated. A parametric form has to be specified for $\edur$.
When marks enter the picture, indicating with $d_i$ the $i$-th duration and with $x_i$ the relative mark, the data generating process can be written
\begin{equation}
(d_t,x_i) \sim f(d_i,x_i|\fil_i)
\end{equation}  
where $f$ is the joint distribution of the $i$-th duration and mark conditional on all the available information $\fil_i$. As was pointed out in \cite{engle2000}, this joint density can be rewritten as 
\begin{equation}
(d_i,x_i) \sim f(d_i,x_i|\fil_i) = g(d_i|\fil_i) \cdot h(x_i|d_i,\fil_i)
\end{equation}

A possible way to construct a model is thus to specify $f$ or $g$ and $h$ as well as the distribution of the error terms. This approach has been followed for example in \cite{manganelli2005duration}. In this stream of literature, the event process is described using Engle's and Russels' ACD models. Then, a model has to be specified for each mark's conditional distribution.  

These models allow for a lot of flexibility, however, they also presents some significant limitations. In fact, parameters do not have a clear interpretation so that the excitement structure is not easy to recover. More important, these models are strongly parametric and the choice of the error distribution can have a large influence on the model performance \citep{AllenComparison}. More generally, a strong structure is imposed \textit{a priori} with the selection of the functional form for $g$ and $h$.
Finally, a severe limitation of ACD models for our purposes is that duration based approaches are not easily generalized to multiple dimensions. Therefore, it is hard to extend these models to describe the whole order book dynamics. 

Intensity-based approaches are on the contrary much more amenable to extension to multiple dimensions. Among intensity based models, Hawkes self-exciting processes \citep{hawkes1971} have been applied successfully in finance to model the irregular arrival in time of a number of event types (trades, quotes, etc.), see for instance \citet{Bowsher2007,bacry2011, Filimonov:2015fm,bormetti_cojumps, Hardiman:2013kj, bauwens2009_review,Rambaldi2015} and \citet{2015arXiv150204592B} for a recent review.

Hawkes processes are a family of doubly stochastic point processes in which the events rate of arrival, $\lambda_t$, is a random function. In particular, at time $t$, $\lambda_t$ is given by
\begin{equation}
\lambda_t  = \mu +\int_{-\infty}^t \phi(t-s) \diff N_s = \mu + \sum_{t_i<t} \phi(t-t_i)
\end{equation} 
where 
%$\fil_t$ is the information available at time $t$, 
$\mu$ is a constant baseline intensity and the function $\phi$, called the kernel of the process, is a non-negative function that is causal (in the sense that its support lies in $\mathbb{R}^+$) which determines the influence of past event on the present value of the event rate.

The process can easily be generalized to a multivariate setting. Denoting with $D$ the number of components, the intensity of component $i$ reads
\begin{equation}
\label{eq:hawkes_multi}
\lambda^i_t = \mu^i +\int_{-\infty}^t \sum_{j=1}^D \phi^{ij}(t-s) \diff N^j_s
\end{equation}
where $\phi^{ij}(t)$ determines the influence of past event of type $j$ on the intensity of type $i$ events. Let us note that, in the following, we will indifferently use the notation $\phi^{ij}$ or $\phi( j \rightarrow i)$ (though heavier, this last notation has the advantage to clearly indicate the direction of the causal relation).

For applications, a very useful particular case of model \eqref{eq:hawkes_multi} is when the time-arrival process of the jumps is stationary.
It has been shown \citep{hawkes1971}  that it is the case if
the spectral radius of the matrix
$$ ||\phi||_1 =\left( \Vert \phi^{ij} \Vert_1 \right) _{1\leq i,j \leq D} $$
is strictly smaller than 1 (we used the notation $||\phi^{ij}||_1 = \int \phi^{ij}(t) dt$). The intensity process $\lambda_t$ is then shown to be a stationary process and 
\begin{equation}
\label{eq:stat_rel}
\Lambda = \E{\lambda_t} = (\mathbb{I}-\Vert \phi \Vert_1)^{-1} \mu.
\end{equation}
So far, Hawkes processes have been used in finance mainly to model the sequence of event times (e.g. change of midprice), disregarding the corresponding sequence of marks, i.e. how much the corresponding quantity has changed (volume traded, price, etc.). In this paper we demonstrate how multivariate Hawkes process can be used to investigate and model the mutual relationship between different orders taking the order size into account.

%Let us point out that a multivariate point process like \eqref{eq:hawkes_multi} is nothing but a marked point process with marks in a finite set. However, for clarity, throughout this paper we will use the term multivariate when we refer to the simultaneous modeling of different \emph{types} of event (e.g. trades and limit orders), while we reserve the term \emph{marked} for quantities associated to each event (such as price, volume,..) that typically assume values in an infinite (or very large) set. In this work we are thus interested in the multivariate and marked case, i.e. we observe a sequence $\lbrace \dots (t_s^i, V_s^i), (t_{s+1}^i, V_{s+1}^i) \dots \rbrace$ 

A simple way to incorporate volumes $v$ in the dynamics is to assume that they are independent from the underlying point process and that the dependence on the volume factorizes,
\begin{equation}\label{eq:naive}
\lambda_t=\mu +\int_{-\infty}^t f(v_s) \phi (t-s) \diff N_s, 
\end{equation} 
As shown by \citep{bacry2014estimation}, the estimation of this model does not pose particular difficulties.

However the predictions of this model are at odds with empirical data. To illustrate this fact and anticipating some of the analyses of this paper, Figure \ref{fig:motivation} shows some of the estimated kernels for the DAX future. We binned the volumes in six groups, $1$ ($6$) labelling the group containing trades with the smallest (largest) volume. The figure plots four kernels $\phi(j\to i)$, measuring how much trades with volume in bin $j$ triggers trades with volume in bin $i$. %Note that under the naive model of Eq. \ref{eq:naive} all the kernels  $\phi(j\to i)$ coincides, while empirical data shows large differences among them. Thus the naive model is unable to capture the important role of volume in the trade dynamics.
Note that under the naive model of Eq. \ref{eq:naive} all the kernels $\phi(j\to i)$ would have the same shape, while empirical data shows large differences among them. Thus the naive model is unable to fully capture the role of volume in the trade dynamics.
\begin{center}
\begin{figure}[tb]
\centering
\includegraphics[width=0.57\textwidth]{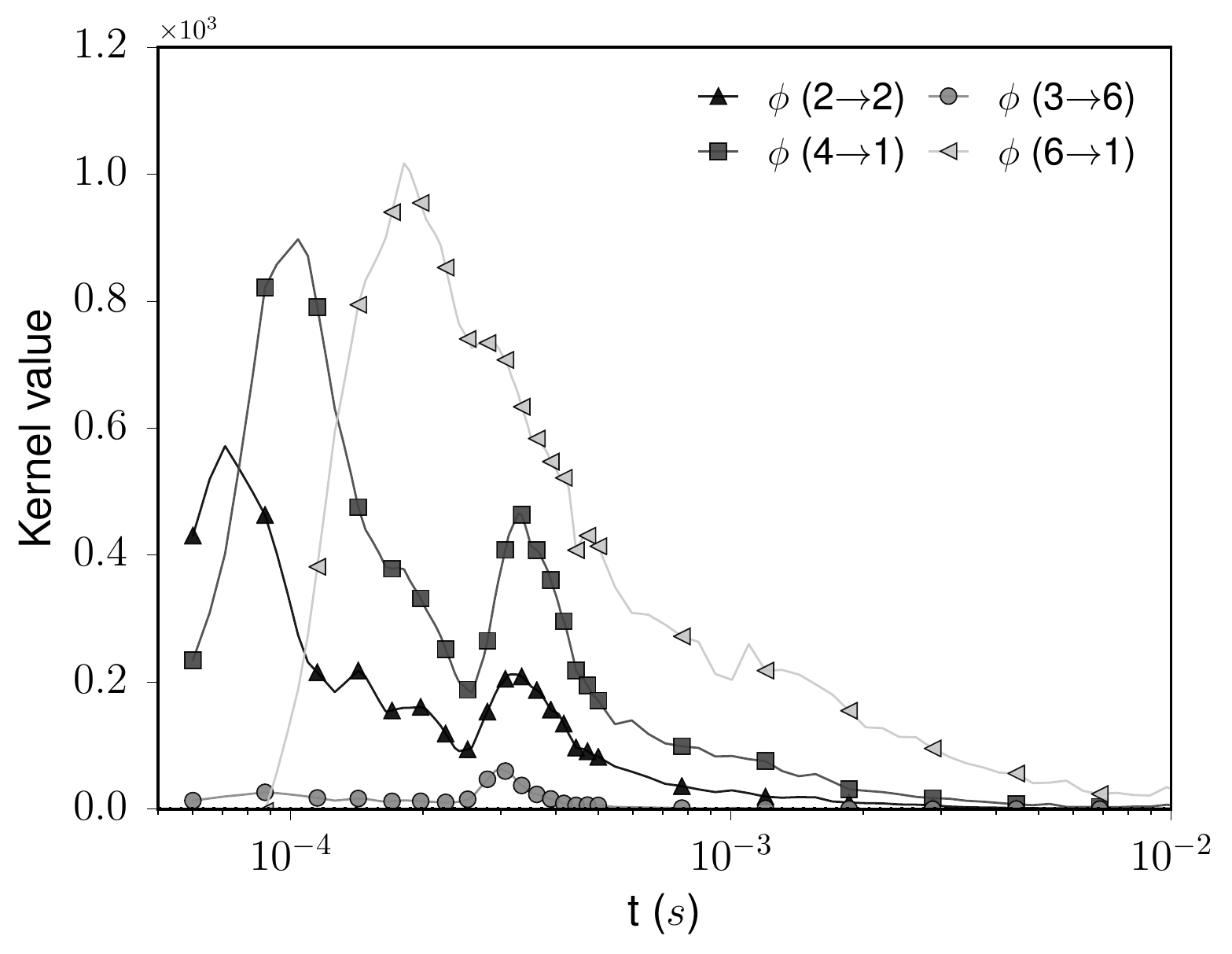}
\caption{Kernel estimates for the DAX future. Trade volumes are binned in six groups, $1$ ($6$) labelling the group containing trades with the smallest (largest) volume. The figure plots the four kernels $\phi^{ij} = \phi(j\to i)$, measuring how much trades with volume in bin $j$ triggers trades with volume in bin $i$.}
\label{fig:motivation}
\end{figure}
\end{center}

This evidence motivates our work and the development of a model that takes into account the complex dependence between time and volume. In the univariate case, the most general Hawkes model specification is 
\begin{equation}
\label{eq:1d_marked}
\lambda_t=\mu +\int_{-\infty}^t \phi (t-s; v_s) \diff N_s
\end{equation}
where the kernel $\phi$ depends on the time distance to previous events as well as on their marks\footnote{Note that in the case of volume we furthermore require that
\begin{equation}
\phi( t, v_1 + v_2) = \lim_{\Delta t \to 0} \phi(t, v_1) + \phi(t+\Delta t, v_2).
\end{equation}
That is in the limit where two events become indistinguishable in time, their effect must be the same of the combined event.}. 
%In turn, the probability of attaching the mark $v$ to an event at time $t$ conditional on the filtration $\fil_t$ of the process can depend both on the time distances to previous realizations and on the relative marks:
%\begin{equation}
%\label{eq:1d_marked_2}
%P(v;t|\fil_t)=f(v, \{v_s\}_{s=t-1, t-2,...}, \{t-s\}_{s=t-1, t-2,...})
%\end{equation}
The interdependence structure of the time $t$ and the mark process $v_t$ of the kernel is very difficult to estimate directly. A simple workaround consists in 
considering the volume process $v_t$ as the superposition of $D$ unmarked point processes, each of which corresponds to one of the possible $D$ values $\{V^i\}_{1\le i \le D}$ that $v_t$ can take. 
One can then consider the $D$ multivariate Hawkes process defined by \eqref{eq:hawkes_multi}, where $\lambda^i$ represents the intensity process for the events associated to the value $V^i$. The so-obtained model is an equivalent representation of the original process. 
%Of course, in doing so, we assumed that the process $v_t$ take a finite number of values, in practice in can be very large (e.g. all the possible traded volumes). In order to reduce the dimension $D$ of the so-obtained multivariate model, one generally
%replaces the values $V^i$ by some appropriate buckets, i.e., intervals of possible mark values.
%we consider the volume process $v_t$ as the superposition of $D$ unmarked point processes, each of which corresponds to one of the possible $D$ values that $v$ can take. Denoting with $N^V_t$ the counting process associated with $V_t$, and with $N_t^{V^{(i)}}$ the counting process associated with the $i-$th component, we have $N_t^V=\sum_i N^{V^{(i)}}_t$. Hence we can model the observed process by correctly imposing the cross- and self- excitation between the $n$ components.

This approach is convenient as it allows to treat the case of dependent mark with the already available tools developed for multivariate point process. For a mark variable that can take an infinite number of values, some sort of binning is necessary in order to map the marked process to a multivariate one\footnote{Let us point out that the two approaches are strictly equivalent in the limit of infinitely many bins.}.

In this paper we follow this approach and we use the non-parametric estimation method of \citet{bacry2014} to estimate the kernels and the baseline intensities from empirical data. We start by applying this model to unsigned trades, then by increasing the total dimension of the model we extend the study to signed trades (i.e. differentiating between buyer initiated transactions and seller initiated ones), and finally to the whole first level of the order book. 

In the next section we present the estimation method we use and we discuss how to handle inhibition effects. In Section \ref{sec:data} we describe our dataset and some of its empirical properties. Section \ref{sec:unsigned_vol} contains the empirical results on trades when volume is taken into account. Section \ref{sec:fullob} extends our Hawkes based analytical tool to the first level of the order book by including also limit orders and cancellations, and accounting for order size. Finally, in Section \ref{sec:fin_rem} we present our conclusions.

\section{Kernel estimation procedure and inhibition effects}
\label{sec:neg_np}

All along this paper, we shall estimate the kernels of the multivariate Hawkes process by using the non parametric method proposed in \citet{bacry2014} and its variant presented in \citet{bacry2014estimation}. It has been shown that the method works well when large amounts of data are available and the kernel is non-localized, which is typically the case in high-frequency finance applications. This methodology leverages the fact that a stationary multivariate Hawkes process is completely specified by its first- and second-order properties. In fact, the knowledge of the average intensity vector $\Lambda =\E{\lambda(t)}$ and of the conditional laws $g^{ij}(t)$ 
\begin{equation}
\label{eq:claw_def}
g^{ij}(t) \diff t = \E{\diff N_t^i |\diff N_0^j =1} -\epsilon_{ij} \delta(t) -\Lambda^i \diff t
\end{equation}
where $\epsilon_{ij} = 1$ for $i=j$ and 0 otherwise, is sufficient to recover the kernel matrix $\phi$ and the vector of baseline intensities $\mu$ thanks to the following result (proved in \citet{bacry2014}) : 

\begin{prop} 
Let $g(t)$ be the matrix of conditional laws of a Hawkes process, then the kernel matrix $\phi$ is the only causal solution of the Wiener-Hopf system
\begin{equation}
\label{eq:wiener-hopf}
g(t) = \phi(t) + g \ast \phi(t), \;\; \forall \; t>0.
\end{equation}
where $\ast$ denotes the usual matrix product where multiplication is substituted by convolution.
\end{prop}

We remark that \eqref{eq:wiener-hopf} has at most one solution even if $g(t)$ is not generated by a Hawkes process. More precisely the following result holds (cf. \cite{jaisson:tel-01212087} p. 142)
\begin{teo}
If $g$ $\in\;L^1$ is the conditional law of a point process, then Equation \eqref{eq:wiener-hopf} has one and only one solution in $L^1$.
\end{teo} 
Hence, it is reasonable to look for a solution of \eqref{eq:wiener-hopf} when $g(t) $ is estimated on empirical data, which is the case of interest here.

The linear Hawkes process \eqref{eq:hawkes_multi} allows for self- and cross-excitation but not for inhibition. In fact, the non-negativity of the kernels implies that the occurrence of an event never causes a decrease in the intensity\footnote{This also implies that in a linear Hawkes process framework there cannot be a minimum time separation between events as there is always a positive probability of a new event arrival.}. However, in real systems such as the order book we can expect the presence of inhibitory effects alongside exciting ones. 

A proper account for inhibition in a Hawkes processes framework requires to abandon the linear specification \eqref{eq:hawkes_multi} in favor of nonlinear formulations that ensure the positiveness of the intensity, while at the same time allowing the presence of negative valued kernels and thus inhibitory effects (\cite{bremaud1996stability,Sornette:2005js, Bowsher2007, Zheng2014}). This comes however at the  expenses of mathematical tractability. 

An interesting case for practical application, discussed in \cite{bremaud1996stability, ReynaudBouret:2010kw, hansen2015lasso} is
\begin{equation}
\label{eq:pos_part_hawkes}
\lambda(t) = \left( \mu + \int_{-\infty}^t \phi(t-s) \diff N_s \right)^+
\end{equation}
where $(x)^+ = x$ if $x\geq0$ and $(x)^+=0$ otherwise. This case is relevant because if $\mathbb{P} \left[  \mu + \int_{-\infty}^t \phi(t-s) \diff N_s < 0\right]$ is negligible, then it is almost equivalent to the linear case. If this requirement is met, then the non-parametric estimation procedures developed in \citet{bacry2014} and \citet{ReynaudBouret:2010kw} have been shown to lead to reliable results even in presence of moderate inhibition. On the contrary, the non-parametric approach developed in \cite{lewis2011nonparametric} leverages a probabilistic interpretation of the kernel and therefore leads always to non-negative estimates that rule out inhibition.

When we consider the nonparametric estimation method of Equation \ref{eq:wiener-hopf}, it is important to stress that the solution $\phi$ of the equation above is not guaranteed to be positive. Moreover it is possible to show that the negativity of the conditional law $g$ implies the negativity of the kernel $\phi$.  More precisely,

\begin{prop}
\label{prop:2}
Let $g(t)$ and $\phi(t)$ be two matrices of functions in $L_1$ that satisfies 
\begin{equation*}
g(t)=\phi(t)+\phi \ast g(t)\;\;\;\forall \; t>0
\end{equation*}
elementwise. Moreover, let the following assumptions hold:
\begin{enumerate}
\item $g(t)$ is bounded;
\item the spectral radius of the matrix $||\phi||_1=\lbrace \Vert \phi^{ij}(t)\Vert_1\rbrace$ is less than one;
\item $\phi^{ij}(t)$ is a causal function, i.e. $\phi^{ij}(t)=0$ $\forall \; t<0,\; \forall i,j$;
\end{enumerate}

Then, if at least one entry of every column or row of $g$ take on negative values for some $t>0$, then also some elements of $\phi$ take on negative values.
\end{prop}
\begin{proof}
See Appendix \ref{sec:proof}.
\end{proof}

Assumptions 1-3 basically require that $\phi(t)$ possesses the characteristics for being the kernel of a stationary Hawkes process and $g(t)$ those for being a conditional law. Note in particular that condition 2 was shown by \cite{bremaud1996stability} to hold as a stability condition even when $\lambda(t)$ is a nonlinear positive Lipschitz function of $\phi \ast \diff N$ and where the kernels can take negative values.

Whenever some inhibition effect is present, the conditional laws \eqref{eq:claw_def} will take negative values. In fact, we would have that the expected arrival rate of type $i$ event given that event $j$ has occurred is lower than the unconditional rate $\Lambda_i$. Therefore, the above proposition suggests that if some inhibition effect is present, then with our procedure we should see some (at least partially) negative kernel. We will actually encounter some of them, particularly in section \ref{sec:fullob}.

We conclude this section by giving an outline of the estimation procedure we employ. We refer to the original papers \cite{bacry2014} and \cite{bacry2014estimation} for further details.

\begin{algo}{Non-parametric estimation of Hawkes kernels}
\begin{enumerate}
\item For each day of trading, the times and volumes of orders are extracted.
\item Each event is assigned to a component of the multivariate process according to its volume.
\item The conditional law \eqref{eq:claw_def} is estimated using empirical means. For this estimation a linear-log binning is used, with edges:
$$ [0,\delta_{\mbox{\tiny{lin}}},2\delta_{\mbox{\tiny{lin}}},3\delta_{\mbox{\tiny{lin}}},h_{\mbox{\tiny{min}}}, e^{\delta_{\mbox{\tiny{log}}}},e^{2\delta_{\mbox{\tiny{log}}}},e^{3\delta_{\mbox{\tiny{log}}}},\dots,h_{\mbox{\tiny{max}}} ]$$
where $\delta_{\mbox{\tiny{lin}}}$, $h_{\mbox{\tiny{min}}}$ and $\delta_{\mbox{\tiny{log}}}$ are user-defined parameters. In this work we fixed $h_{\mbox{\tiny{min}}}=10^{-3}s$ and $h_{\mbox{\tiny{max}}} = 2\cdot 10^4s$.  $\delta_{\mbox{\tiny{lin}}}$ and $\delta_{\mbox{\tiny{log}}}$ are chosen so to have 50 bins in the linear part and 1500 bins in the log-spaced part.

\item The kernels are then estimated by solving the integral equation \eqref{eq:wiener-hopf}. Again a lin-log quadrature scheme similar to the one above is used. All kernels are estimated on the domain $[0, x_{\mbox{\tiny{max}}}]$ and the quadrature points are 
$$ [0,\epsilon_{\mbox{\tiny{lin}}},2\epsilon_{\mbox{\tiny{lin}}},3\epsilon_{\mbox{\tiny{lin}}},x_{\mbox{\tiny{min}}}, e^{\epsilon_{\mbox{\tiny{log}}}},e^{2\epsilon_{\mbox{\tiny{log}}}},e^{3\epsilon_{\mbox{\tiny{log}}}},\dots,x_{\mbox{\tiny{max}}} ]$$
again, $\epsilon_{\mbox{\tiny{lin}}}$, $x_{\mbox{\tiny{min}}}$, $\epsilon_{\mbox{\tiny{log}}}$ and $x_{\mbox{\tiny{max}}}$ are user-defined parameters.
We choose $x_{\mbox{\tiny{min}}}=0.5\cdot 10^{-3}s$ and $x_{\mbox{\tiny{max}}} = 0.5 s$.  $\epsilon_{\mbox{\tiny{lin}}}$ and $\epsilon_{\mbox{\tiny{log}}}$ are chosen so to have 80 bins in the linear part and 80 bins in the log-spaced part.

\item Finally the baseline intensities $\mu_i$ are recovered by solving for $\mu$ the stationarity condition \eqref{eq:stat_rel}.
\end{enumerate}

\end{algo} 
Let us point out that some care is needed in the choice of the grid used in point 3. of the algorithm above for the estimation of the conditional law (as shown in \cite{bacry2014estimation}) since, in typical application to financial data, its mass is found to span many orders of magnitude in $t$.

\section{Data and main statistics}
\label{sec:data}
In this Section we describe our dataset and we provide an empirical analysis of its main characteristics.

\subsection{Data}
In this paper we use level-I order book data provided by QuantHouse EUROPE/ASIA (http://www.quanthouse.com) for two future contracts, namely the German Bund Future and the DAX Future\footnote{This is the same dataset used in \cite{bacry2014estimation}}. The data span the period from July 2013 to November 2014.

In the dataset, a timestamp is added by the market every time a change in the first level of the order book is registered and outstanding quantities at the best quotes as well as the corresponding prices are specified. A specific timestamp is also added for every trade together with the size, price and side of the order book at which it was executed. Timestamps are provided by the market itself with microsecond precision.%, however, this does not mean they are accurate to that level of precision.   

We use the quantity and price information available in the dataset to reconstruct the sequence and type of orders. We distinguish between limit orders, cancellations and trades. We consider only orders at the best quotes (first level of the order book). It is important to remark that we treat multiple orders that happen at the same time and on the same side of the order book as a single event (for instance a market order that hits two limit orders present in the book at the same price is regarded as a single trade). It is still possible to have two simultaneous events on opposite side of the book, i.e. one at the ask and one at the bid. This occurrence is however very rare thanks to the fine time resolution of the dataset, resulting in a ratio of less than 0.2\% of simultaneous events to the total number of events.

\subsection{Empirical properties of durations and volumes}

Figure \ref{fig:dur_all_bund} shows the histograms of all inter-event times (right panel) and of inter-trades times alone (left panel). We never observe events closer than about $10\mu $s, and this may be a minimum technical delay due to data processing. For trades alone, the minimum observed time distance is larger, at about $50 \mu$s, possibly due to longer processing time needed in case of transactions. 

In the duration distributions, we note a major peak around $30 \mu$s when all events are considered and around $70-100\mu$s when only trades are taken into account. There is a second noticeable peak between $200$ and $300\mu s$ for all events and around $300-400\mu$s for trades. Finally, even quite far in the tail, small peaks are visible. Those might be the results of order splitting and automatic trading.
\begin{center}
\begin{figure}[tb]
\centering
\includegraphics[width=0.47\textwidth]{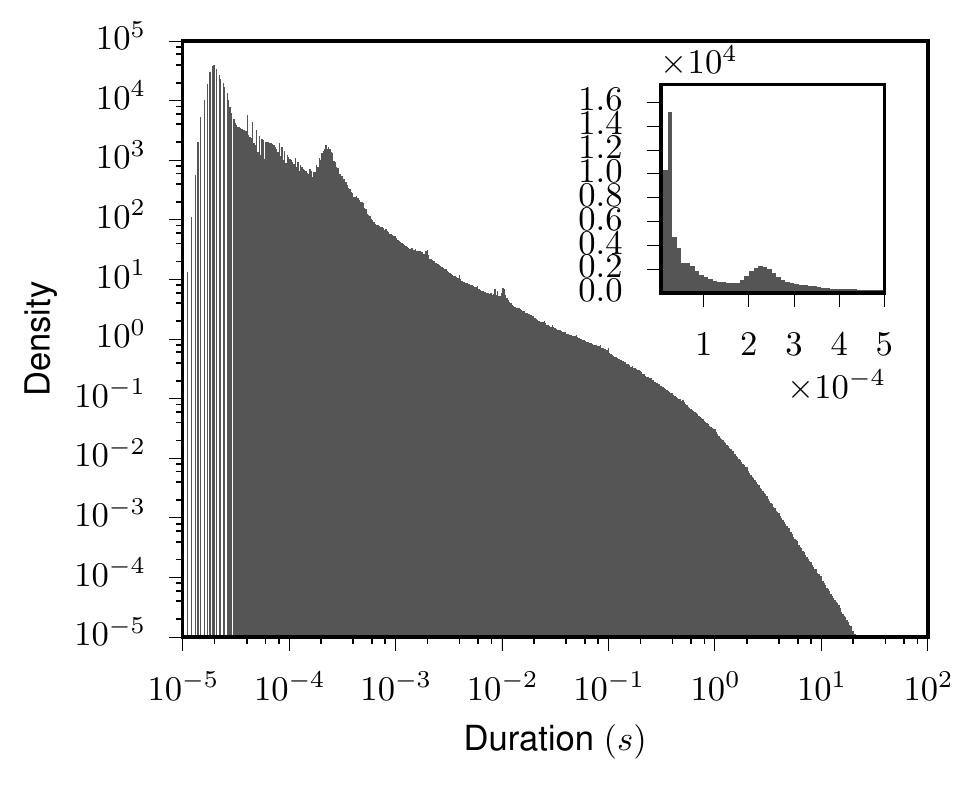}
\includegraphics[width=0.47\textwidth]{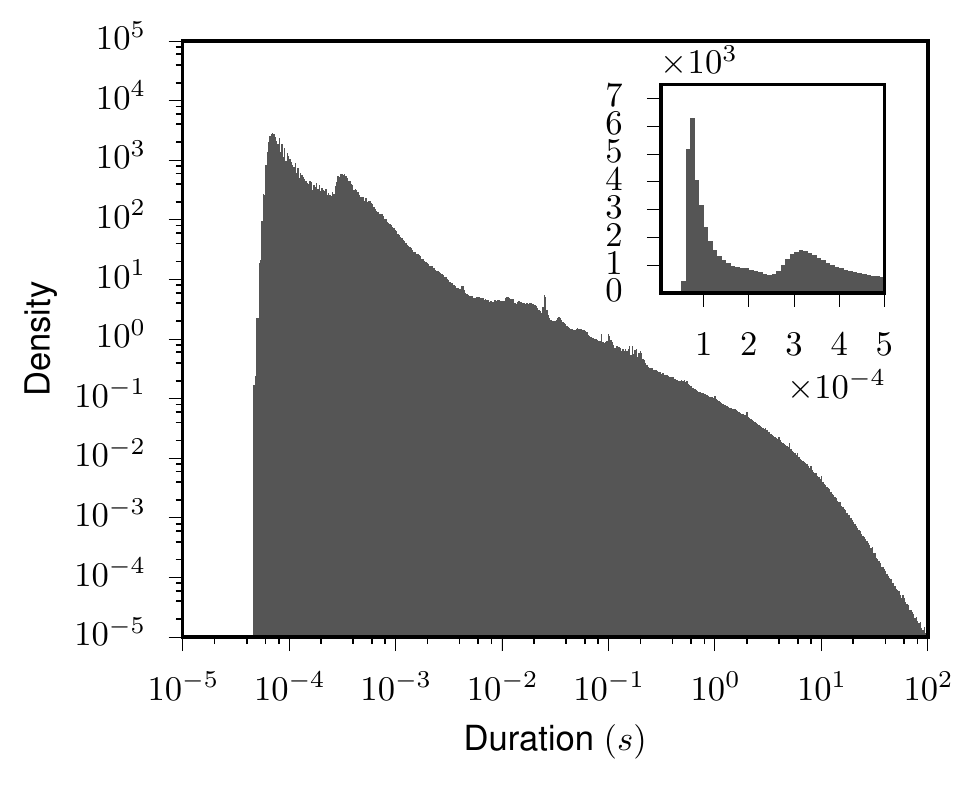}
\caption{Histograms of the inter-event times. All events (left) and only trades (right). Main plots use log-spaced bins, while the inset show equally spaced histograms for short durations. Data refer to the Bund future, no significant difference is observed for DAX future data.}
\label{fig:dur_all_bund}
\end{figure}
\end{center}
Figure \ref{fig:vol_dist} shows the empirical distribution of the trade sizes, measured in number of contracts, for the DAX and Bund futures. A positive sign indicates buyer initiated transaction while a negative sign stands for seller initiated ones. For both assets we note sharp peaks in correspondence of "round" order sizes such as 1, 50, 100. In the DAX future market, transactions of size above 100 contracts are very infrequent, while in the Bund Future market volumes of 1000 contracts are not uncommon. This difference is probably due to the fact that the Bund belongs to the class of the so called "large-tick" assets, while the DAX future is considered a "small tick" one. That is, in the former case, the spread is almost always equal to one tick, thus resulting in fewer price changes and larger available volumes at the best quotes \citep{eisler2012price}. 
\begin{center}
\begin{figure}[tbh]
\centering
\includegraphics[width=0.47\textwidth]{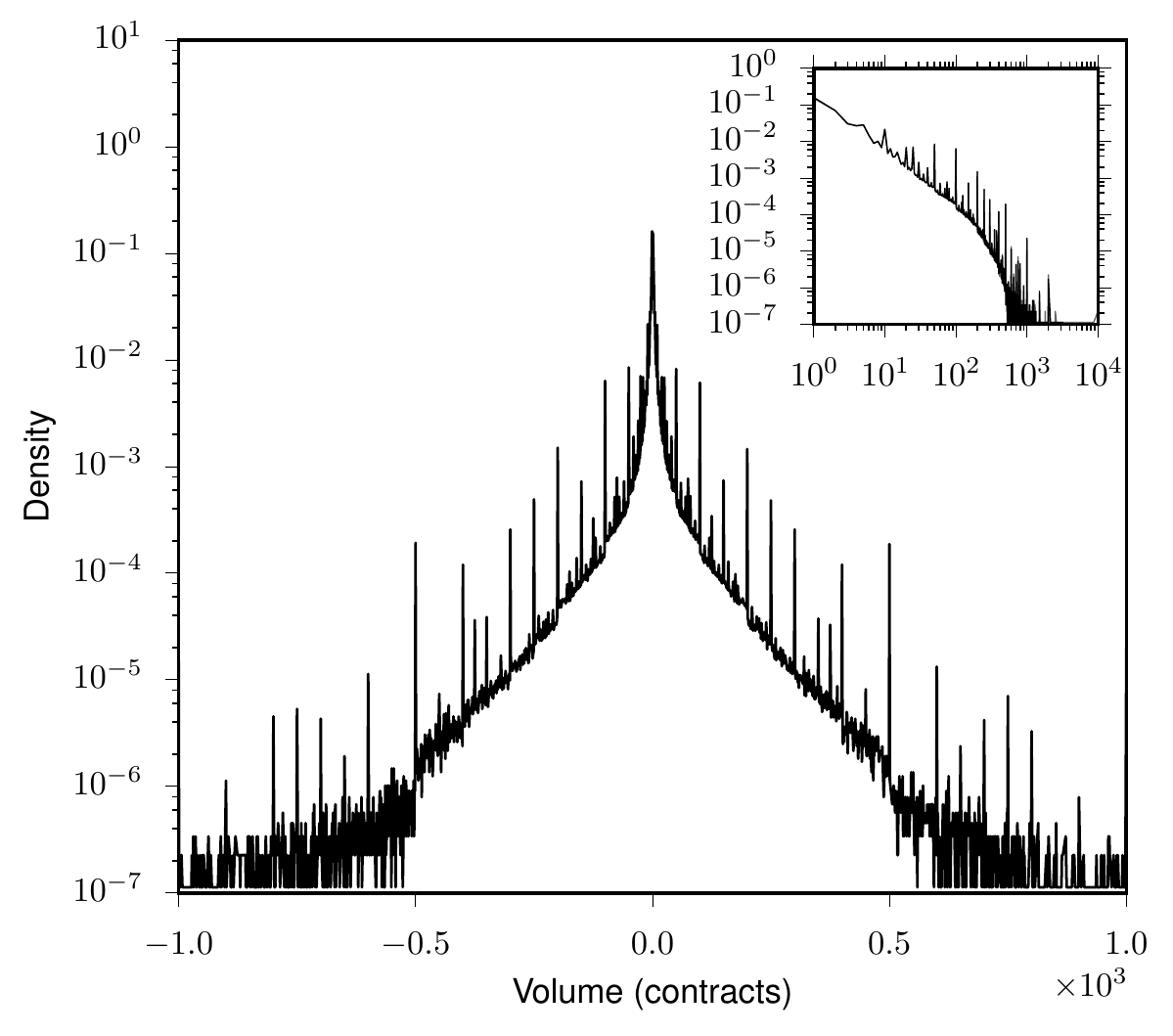}
\includegraphics[width=0.47\textwidth]{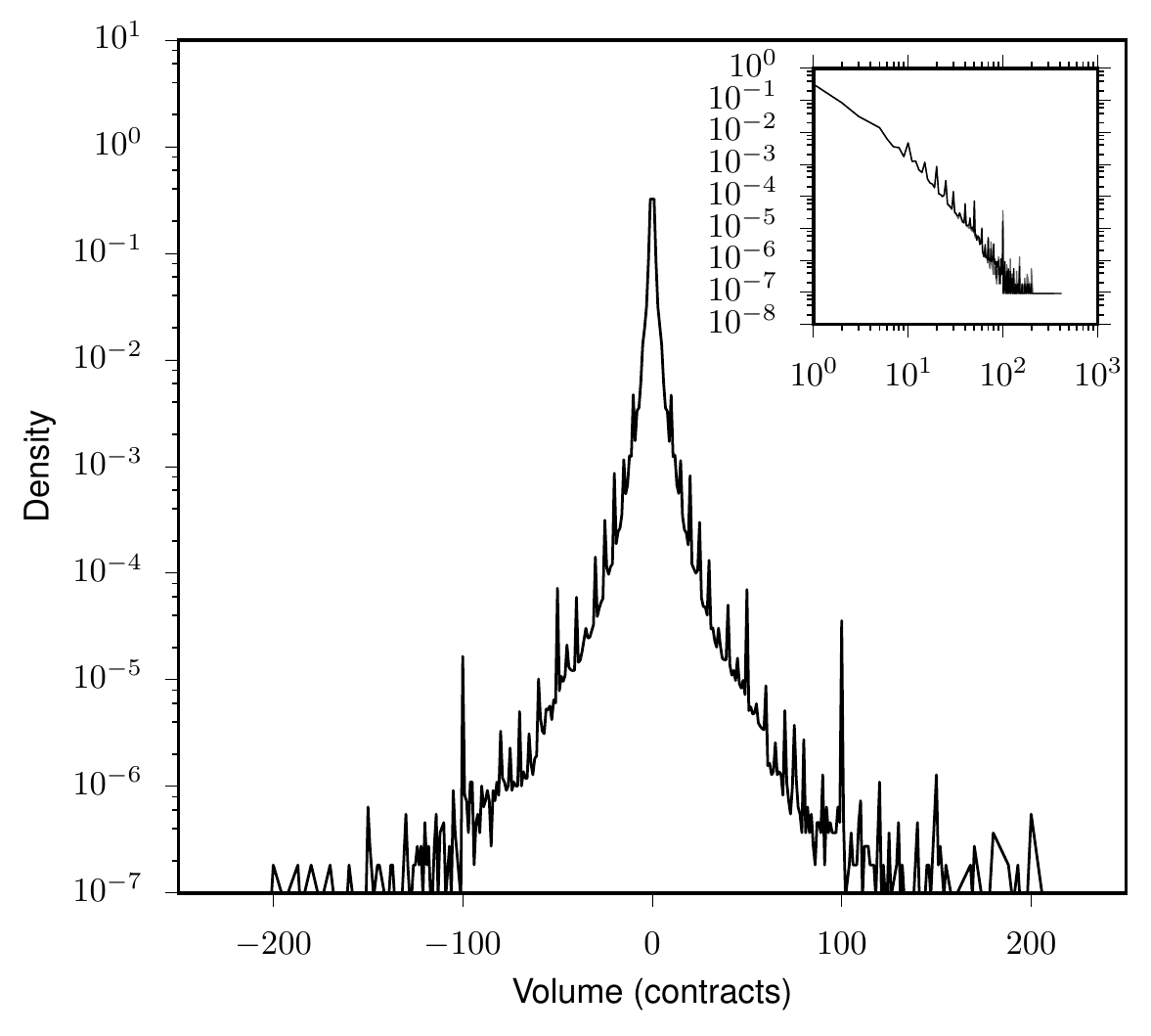}
\caption{Empirical distribution of the trade volumes for the Bund (left) and DAX (right) futures. The inset shows the distribution in log-log scale where buy and sell are aggregated.}
\label{fig:vol_dist}
\end{figure}
\end{center}

We also examined serial correlations in the trade sizes both in trade time and in real (continuous time). We find that autocorrelation in trade time is very small, and that in real time can be almost entirely explained by correlation in the number of transactions. 
\subsection{High frequency characteristics of Eurex market}
\label{subs:300mus}
 The $300 \mu s$ peak observed in the duration distribution is likely related to technical characteristics of the Eurex market. In particular, on the exchange website is reported that

\textit{"For futures orders Eurex Exchange currently offers customers (daily) average roundtrip times as low as 0.2-0.35 milliseconds in co-location"}\footnote{From \url{http://www.eurexchange.com/exchange-en/technology/co-location-services} visited on 7 May 2015. }.

This means that the minimum time it takes for an agent to react to an event on the market is around $300\mu s$, and explains the bump we see in the duration distribution. As a consequence, no direct causal relationship can hold between events that are closer than this timescale, except if they are the result of the same agent's actions. We need to keep this fact in mind when analyzing the result of our study in the next sections.

 %Rather, it can be due to traders following the same signals or by actions conducted by the same agent. Furthermore, it is also an indication that a very significant fraction of the activity in this market is conduced by high-speed traders. Finally, the ability of our methodology to detect this and other peculiar market details is a good evidence of its power and reliability. 

\section{Mapping Traded volume to multivariate Hawkes processes}
\label{sec:unsigned_vol}
We first apply the multivariate Hawkes framework to the series of trades, thus disregarding other types of order book events. We map trades of different sizes to separate the components of a multivariate Hawkes point process. To keep the total dimension of the resulting multivariate process manageable, we divide the range of order sizes into a small number of bins (six). Once each transaction has been assigned to a component of the multivariate Hawkes process, we use the non-parametric estimation method described before to determine the kernels and baseline intensities of the Hawkes process.
The whole trading day was used (from 8am to 10pm) without taking into account intraday seasonality (the results do not change dramatically by restricting the analysis to the most active hours see Appendix \ref{app:b}). 

\subsection{Unsigned Trades}

We begin by considering unsigned trades, i.e. we do not distinguish between buyer initiated trades and seller initiated ones. We fix as many bins for the unsigned volume as component of our multivariate Hawkes model. Then we assign each transaction to a component based on the bin in which its volume falls.
In Table \ref{tab:bin6D} we report the bin choice as well as the average number of events per day in each bin for the Bund and DAX futures respectively. We note the overwhelming prevalence of size one trades in the case of the DAX future, where they represent more than 60\% of the total. 
%
%WHY WE CHOSE THESE BINS?
%
%% tabella bin 4D
%\begin{table}[ht]
%\begin{center}
%\begin{tabular}{ccc}
%  \toprule[1pt]
% & Volumes (contracts) & Avg. N events  \\ 
%  \midrule
% $B_1$ &  1 & 11049\\
% $B_2$ & $(1,3]$ & 7149\\
% $B_3$ & $(3,10]$ & 8284\\
% $B_4$ & $(10,\infty)$& 8782\\
%\bottomrule[1pt]
%\end{tabular}
%\caption{Bund Future: Binning scheme in the four dimensional case.}
%\label{tab:bin4D}
%\end{center}
%\end{table}
%%
%
\begin{table}[htb]
\centering
\begin{tabular}{c|ccc|ccc}
\toprule[1pt]
&  \multicolumn{1}{p{2cm}}{\centering Volume \\ \small{(contracts)}}
& Avg. N  & Fraction \small{(\%)} &  \multicolumn{1}{p{2cm}}{\centering Volume \\ \small{(contracts)}}
& Avg. N & Fraction \small{(\%)} \\ 
\midrule
$B_1$ &1 & 11049 & 31.3&1 & 28140 & 64.5\\ 
$B_2$ & 2& 4979 & 14.1 & 2&7308 & 16.8\\ 
$B_3$ &3& 2169 & 6.2&3& 2753 & 6.3\\ 
$B_4$ & $(3,7]$ & 5571 &15.8 & $(3,5]$ &2959 & 6.8\\  
$B_5$ & $(7,20]$& 5619 & 15.9& $(5,10]$&1700 & 3.9 \\
$B_6$ & $(20,\infty)$ & 5876 & 16.7& $(10,\infty)$&735 & 1.7\\ 
\bottomrule[1pt]
\end{tabular} 
\caption{Bund (left) and DAX (right) Future: Binning scheme and average number $N$ of events per day in each bin in the six dimensional case}
\label{tab:bin6D}
\end{table}
%

% CLAW 
\begin{center}
\begin{figure}[tb]
\centering
\includegraphics[width=0.49\textwidth]{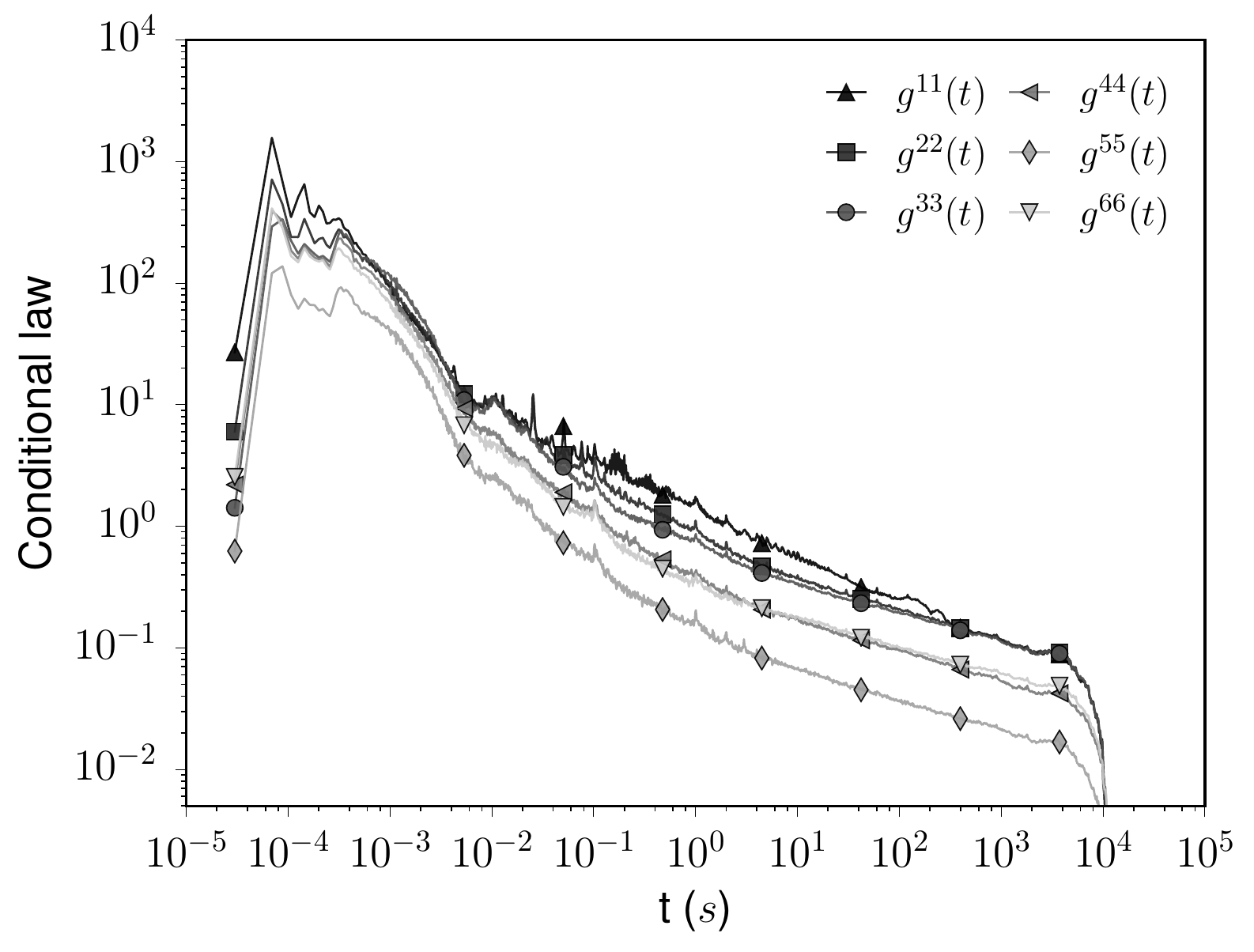}
\includegraphics[width=0.49\textwidth]{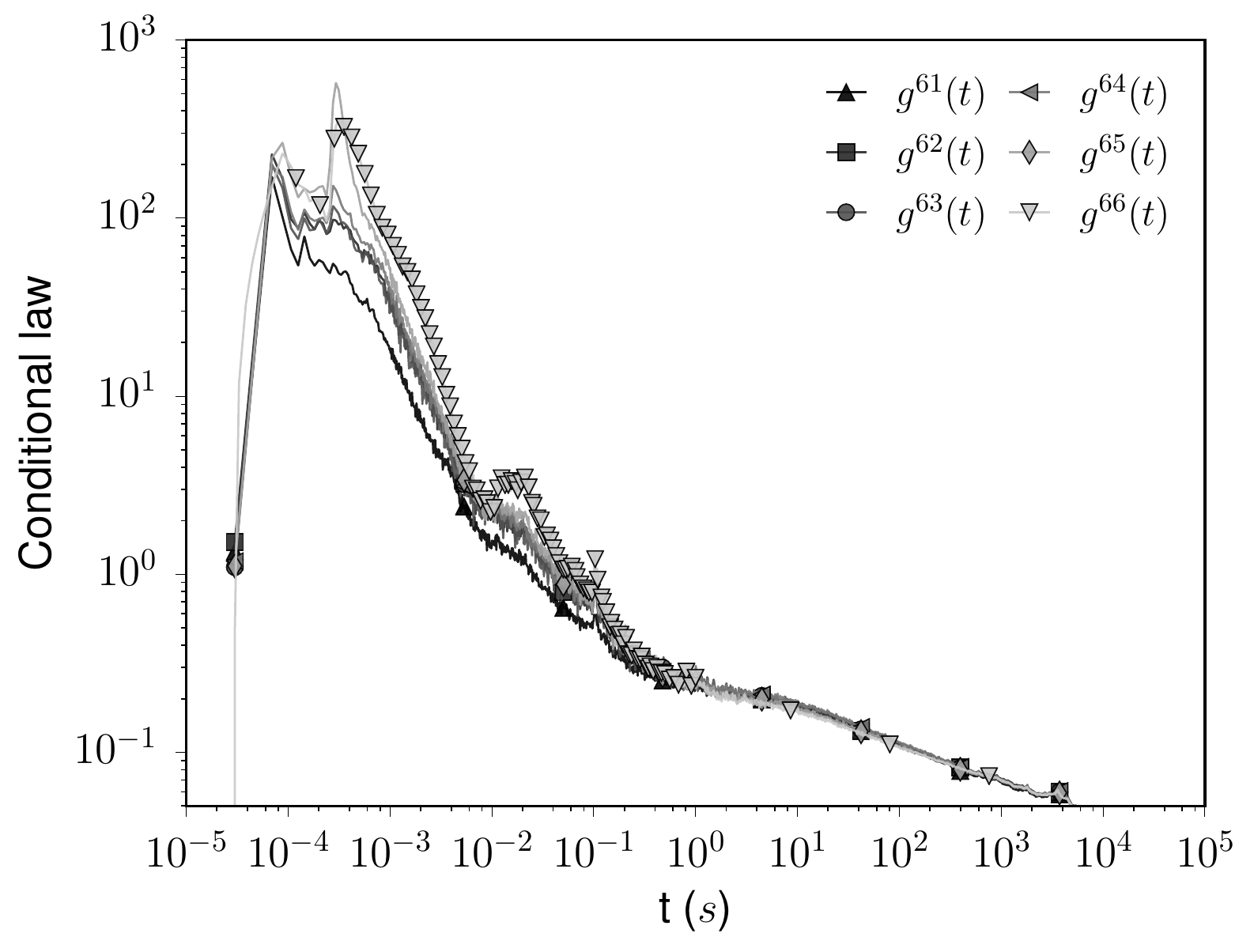}
\includegraphics[width=0.49\textwidth]{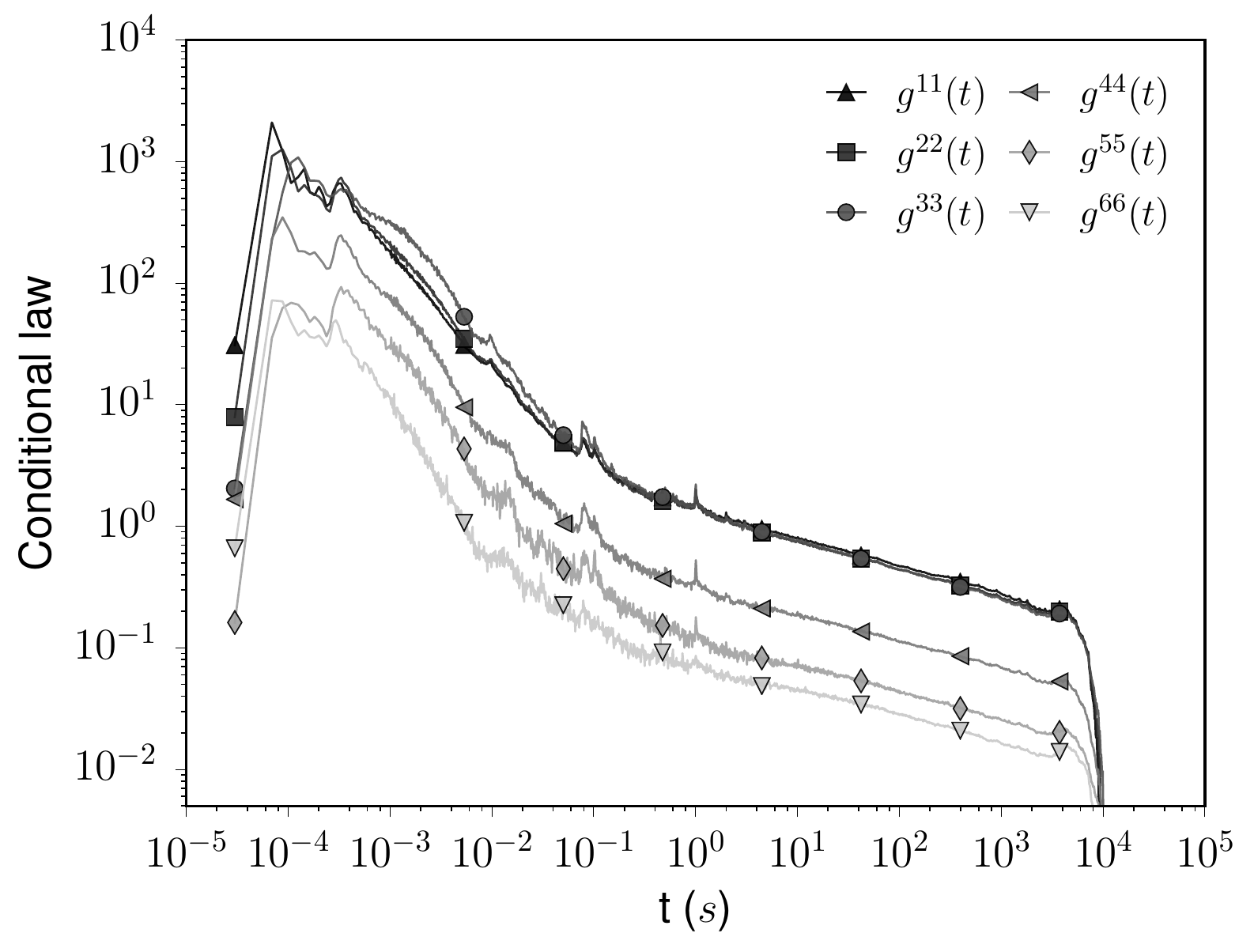}
\includegraphics[width=0.49\textwidth]{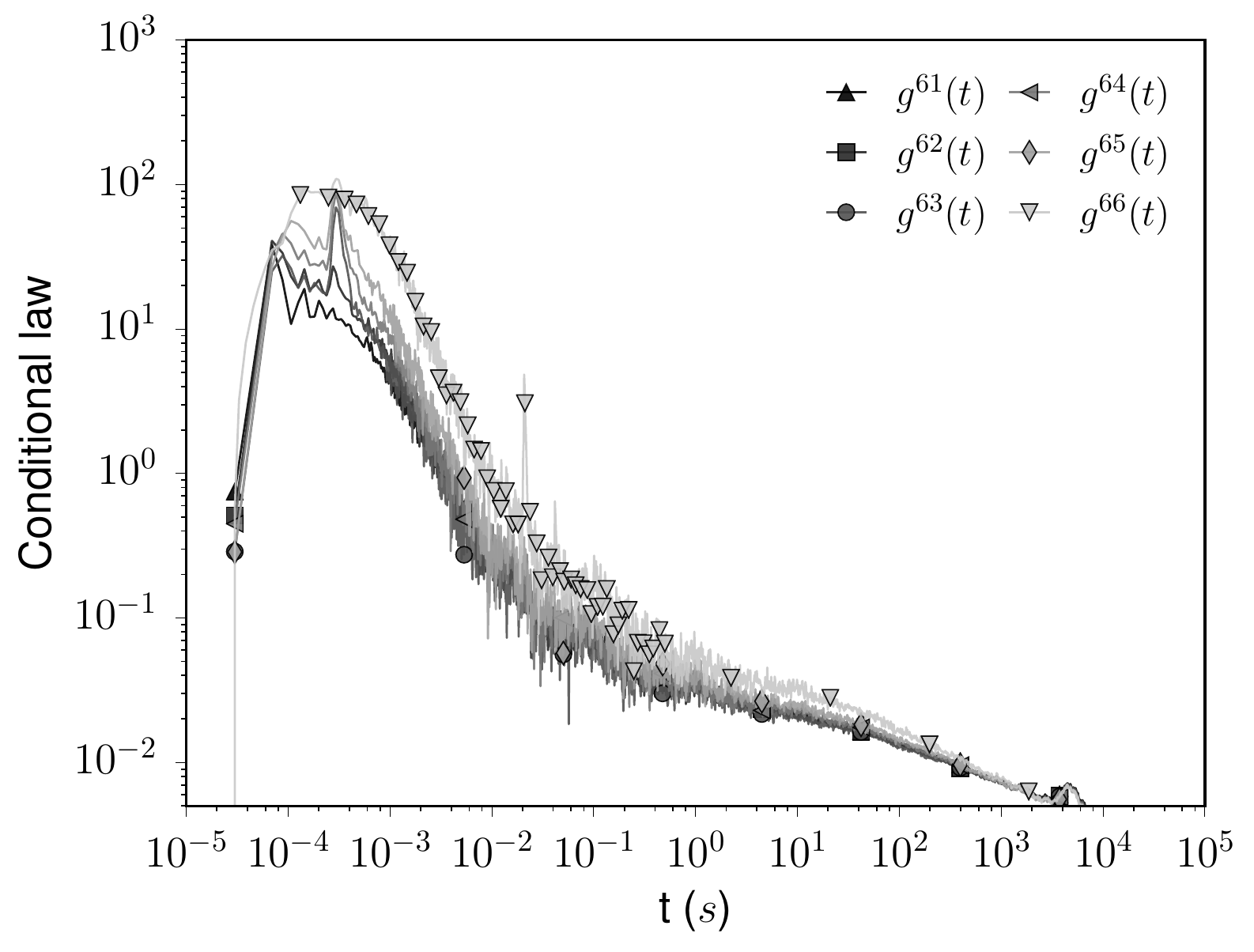}
\caption{Conditional law estimates. Left column: diagonal terms. Right: the row corresponding to largest volumes. Figures on top refers to the Bund future, those at the bottom to the DAX future.}
\label{fig:6D_claw}
\end{figure}
\end{center}

\paragraph{Conditional laws}
The first step of the procedure is the estimation of the conditional laws $g^{ij}(t)$. In the left panels of Figure \ref{fig:6D_claw} we plot the diagonal conditional laws estimates $g^{ii}(t)$ for the Bund and DAX futures. 
These functions attain their largest values at very short lags. In particular, we note two main peaks, the first around 100 $\mu$s and the second around 300 $\mu$s. We also note that this second peak tends to become more relevant for larger volumes. 
Another interesting feature that emerges is the presence of several sharp peaks at "round" time values (0.1 $s$, 1 $s$, ...). These peaks are more evident for the smallest volumes and they are likely the result of automatic trading and order splitting. For lags larger than about one second the diagonal conditional laws show a power law decay with exponents smaller than one.
The right panels of the same figures show instead the terms $g^{6j}(t)$, i.e. those that measure the effect of the different volumes on large trades. At short lags we again note the two major peaks described before. Here we note how the second peak is more pronounced when large-large trades are involved. We will comment more on this when examining the kernels.

% KERNEL
\begin{center}
\begin{figure}[tb]
\centering
\includegraphics[width=0.49\textwidth]{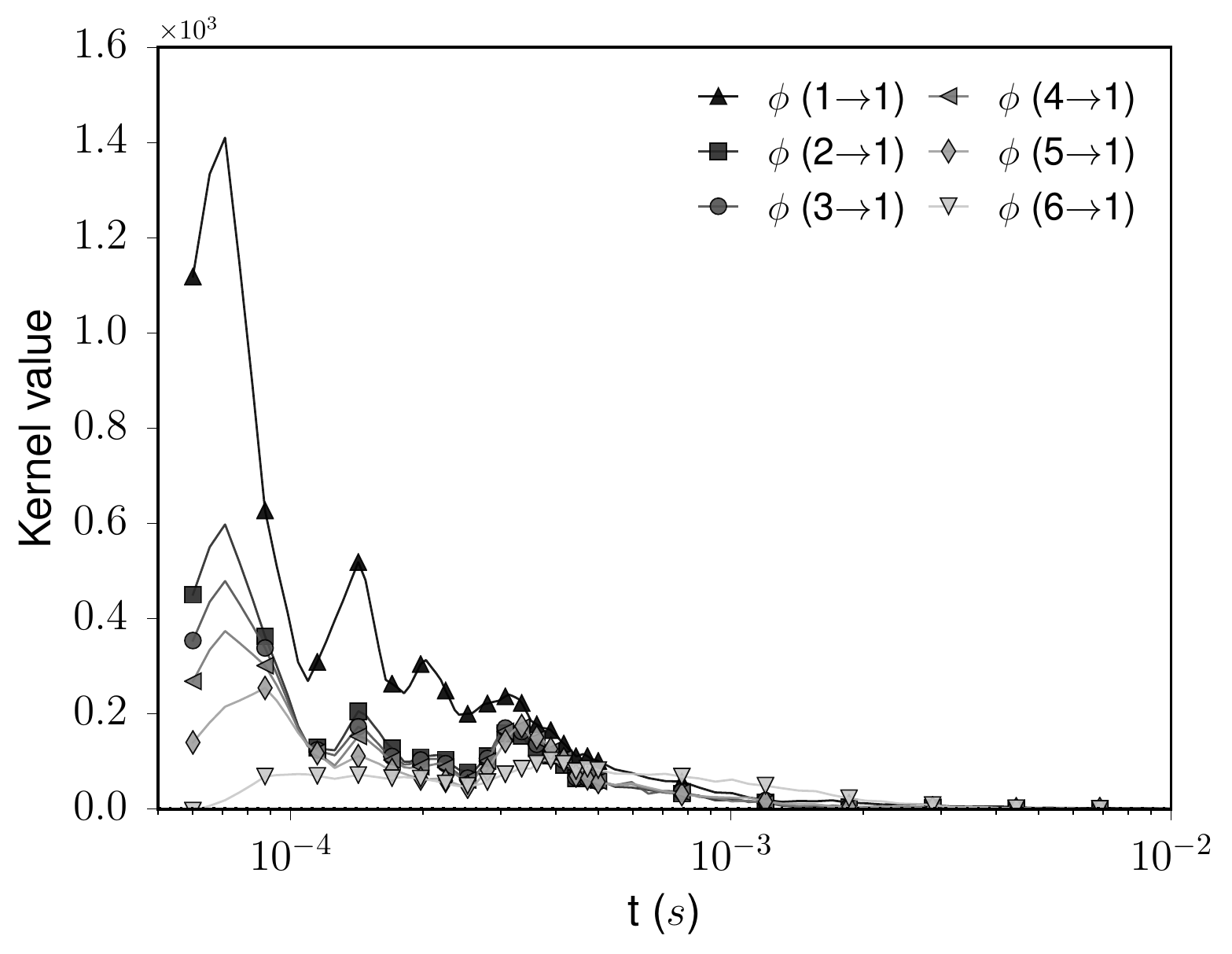}
\includegraphics[width=0.49\textwidth]{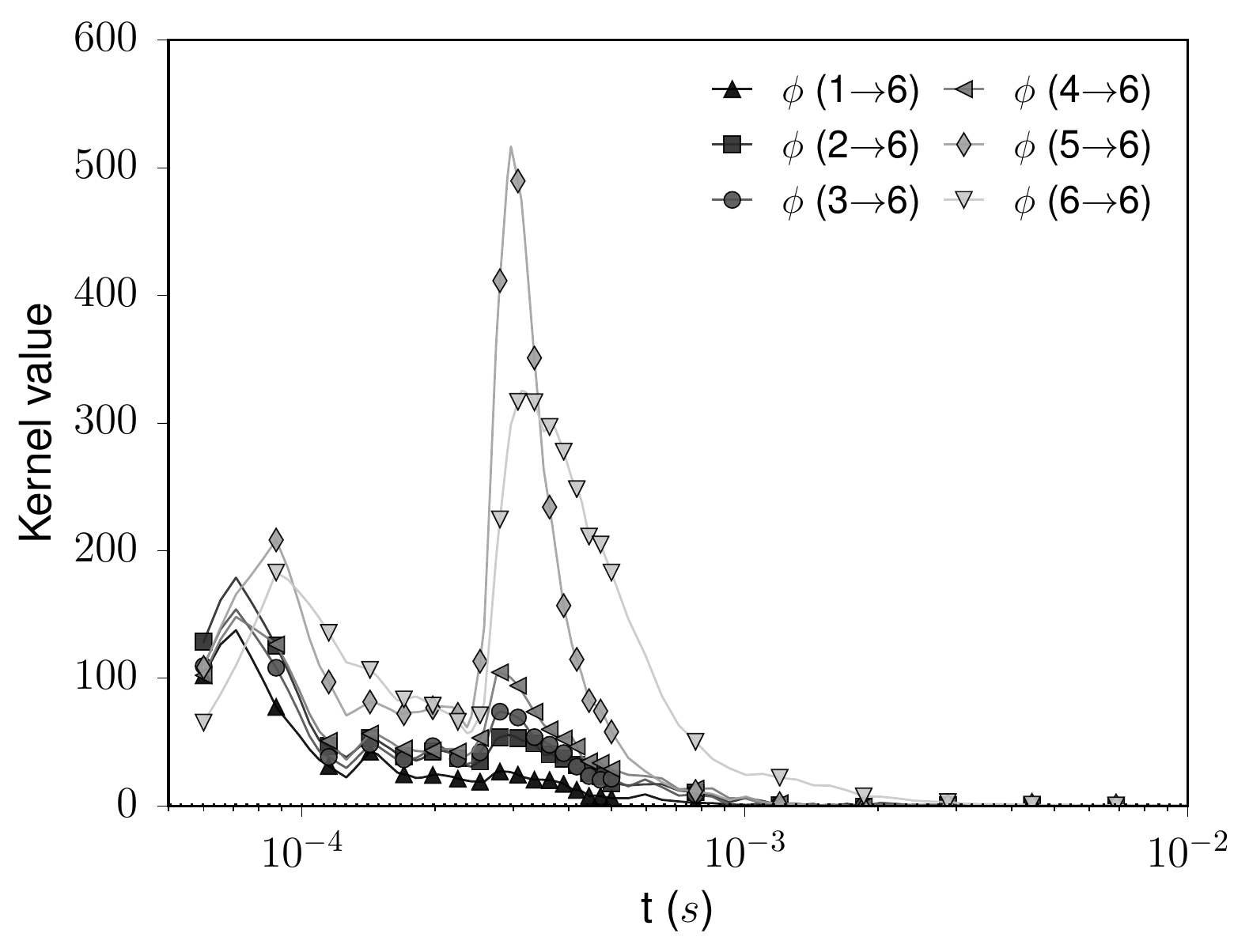}
\includegraphics[width=0.49\textwidth]{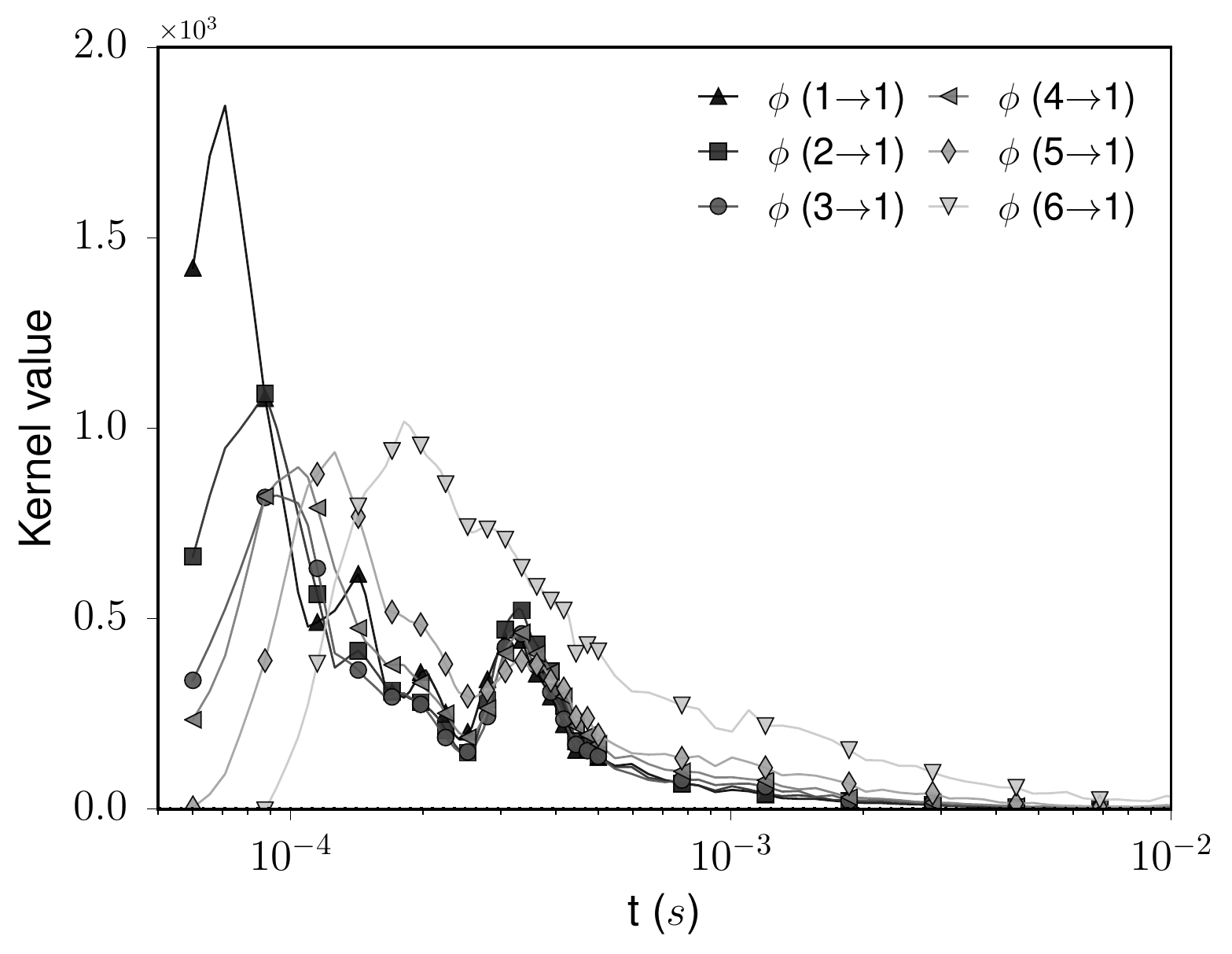}
\includegraphics[width=0.49\textwidth]{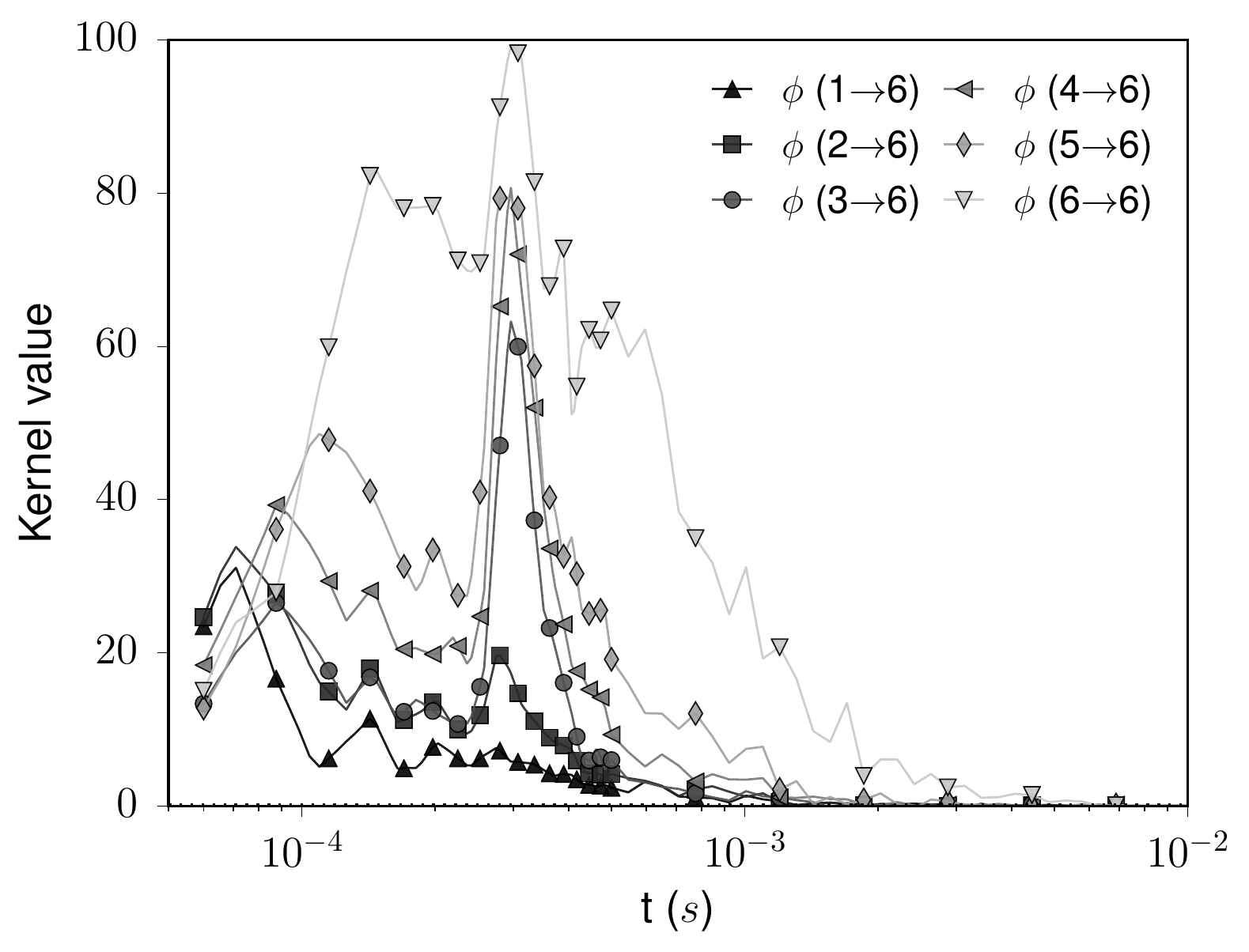}
\caption{Kernel estimates. Figures at the top refers to the Bund, while those at the bottom to the DAX. Left hand side figures: influence on smallest volume. Right hand side: the row corresponding to largest volumes.}
\label{fig:ker_lin}
\end{figure}
\end{center}

\paragraph{Kernels shape} Once the conditional laws have been estimated, we solve the Wiener-Hopf system \eqref{eq:wiener-hopf} for the kernels. A total of 36 kernels are estimated non parametrically and each function $\phi^{ij}(t)$ is estimated for values of $t$ up to $0.5$s. Beyond this value the estimation becomes very noisy. 
Given their number, in examining them, we focus on those contributing to the intensity of the smallest and largest bins. In Figure \ref{fig:ker_lin} we plot the functions $\phi^{1j}$ and $\phi^{6j}$ for the Bund and the DAX futures. To begin with, we remark that taking the volume of the trades into account does add new information. Indeed, the shapes of the kernels involving different sizes are markedly distinct. This has the important consequence that models where the time dependence is separated from size dependence appear to be inadequate.

We now examine the most relevant facts that emerge from the analysis of Figure \ref{fig:ker_lin}. First, we note that the influence of large trades is more intense and more persistent at longer time scales. It thus takes more time for the market to forget about a large trade than it takes for a small one. This can be linked to the fact that the execution of a large trade is more informative for the rest of the market than that of a small one. This is even more relevant in the DAX case where large trades are particularly rare. 

Second, we observe sharp peaks located around $300\mu s$ in the right figures, i.e. those depicting the influence over large trades. These peaks are visible also in the left plots, albeit they are not the main ones there. In light of Section \ref{subs:300mus}, we interpret these peaks as corresponding to the market reaction to some trader's action. This is consistent with the fact that the $300\mu s$ peaks are more significant when large trades are involved. Indeed, a single agent does not typically execute several large order one after the other, and we know that at least the reaction time is needed for an interaction.

The large peaks found on the left figures at $100 \mu s$ are instead likely the result of order splitting, i.e. the same trader executing several consecutive orders, or of traders following the same signals. Actually, we also found the peak at $100\mu s$ to be stronger in small size diagonal term $\phi^{11}$, $\phi^{22}$, which further corroborate the role played by order splitting as traders do not change the size too much when splitting.

\paragraph{Kernels norms} To complete the analysis of the kernel matrix, we plot in Figure \ref{fig:dax_norm} the norms $n_{ij}=\int_0^{\infty} \phi^{ij}(t) \diff t$ of the kernels together with the rescaled norms $\tilde{n}^{ij}=\frac{\Lambda_j}{\Lambda_i} n^{ij} = 1 -\frac{\mu_i}{\Lambda_i}$. It is important to underline the different meaning of these two quantities.
The norm $n^{ij}$ is the average number of events of type $i$ triggered by \emph{one} event of type $j$.Thus, the larger this number is, the bigger is the influence of events of type $j$ on events of type $i$. The rescaled norms represent instead the fraction of the average intensity $\Lambda_i$ attributable to excitation from component $j$. The former is useful to understand what component $i$ is more affected by the arrival of an event $j$ (reading by column) and also how many events of type $i$ are expected after an event of type $j$ (reading by row). The latter instead tells us how much of the activity of component $i$ is explained by component $j$ (reading by row). The comparison of the two matrices is particularly relevant when the components have very heterogeneous average intensities.

In the case of the Bund, the two pictures appear quite similar since, by construction, the average intensities are almost equal across all the components. We observe that self-excitation is preponderant, followed in importance by the excitation from large volumes. The components of the DAX model have instead very different average intensities, with the first component largely dominating. It appears that for this small-tick asset the unusual occurrence of a large trade has a dramatic impact on the activity. However, the scarcity of such large orders is such that when we consider the fraction of each intensity attributable to excitation from other components, size one trades clearly make up the largest amount.
%Thus, while a single large trade triggers on average ore events than smaller trades, when taking relative intensities into account, the largest fraction becomes unsurprisingly that attributed to the first component, i.e. the volume one component. 

\paragraph{Baseline intensities} Lastly, we examine the contribution of the baseline intensities $\mu_i$. In Table \ref{tab:binR6D} we report the ratio between the exogenous intensity $\mu_i$ and the average intensity $\Lambda_i$ for each bin. Since the kernels are estimated only up to $t=0.5 s$, the norms can be underestimated, because the contribution from the tail is disregarded. With this in mind, for both assets it appears nevertheless that large trades are more exogenously driven than the others. The exogenous fractions $\mu_i/\Lambda_i$ for the DAX are higher than those obtained for the Bund, with all but size one trades having exogenous components accounting for 40-50\% of the average intensity. 

% NORM
\begin{center}
\begin{figure}[tb]
\centering
\includegraphics[width=0.49\textwidth]{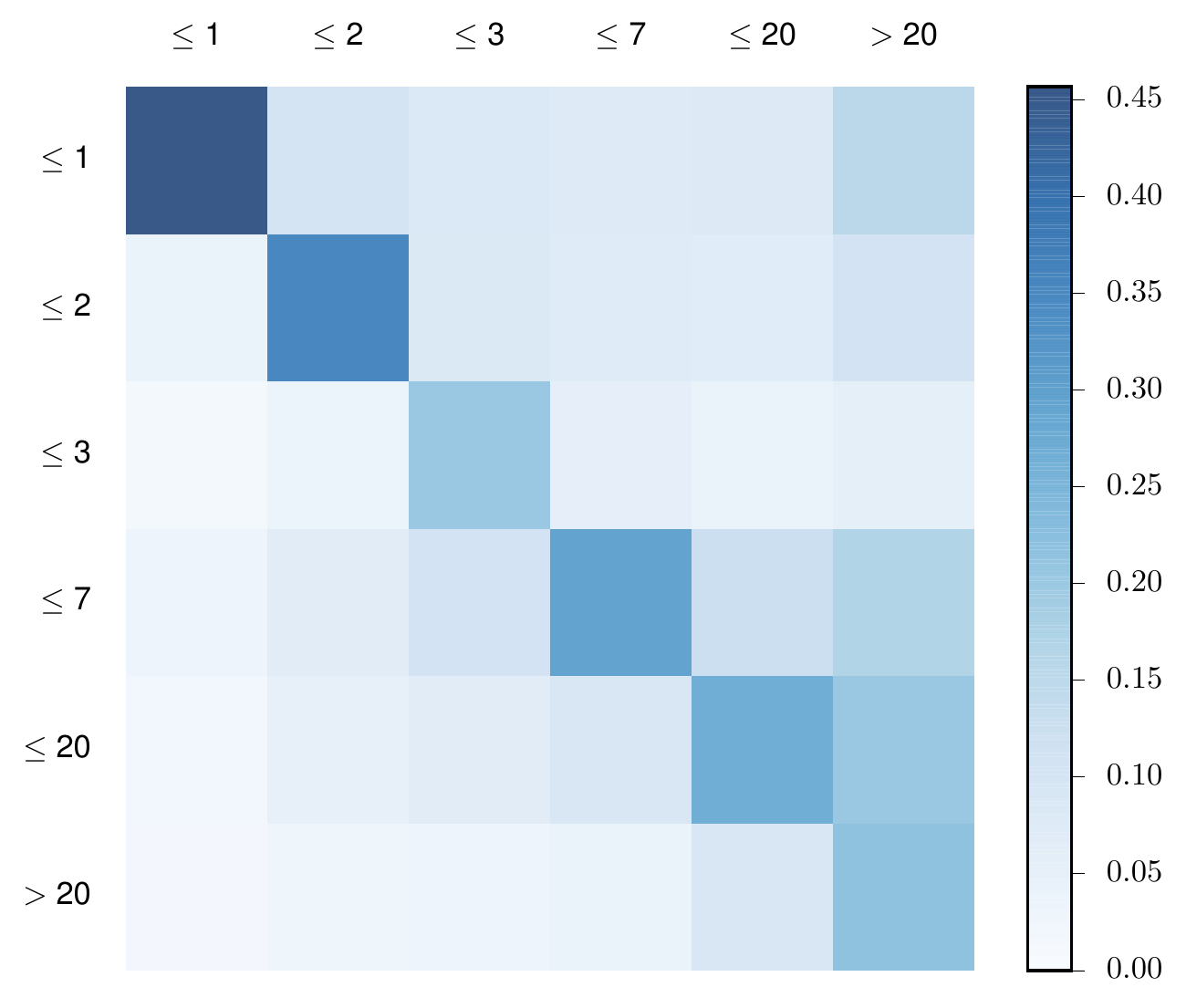}
\includegraphics[width=0.49\textwidth]{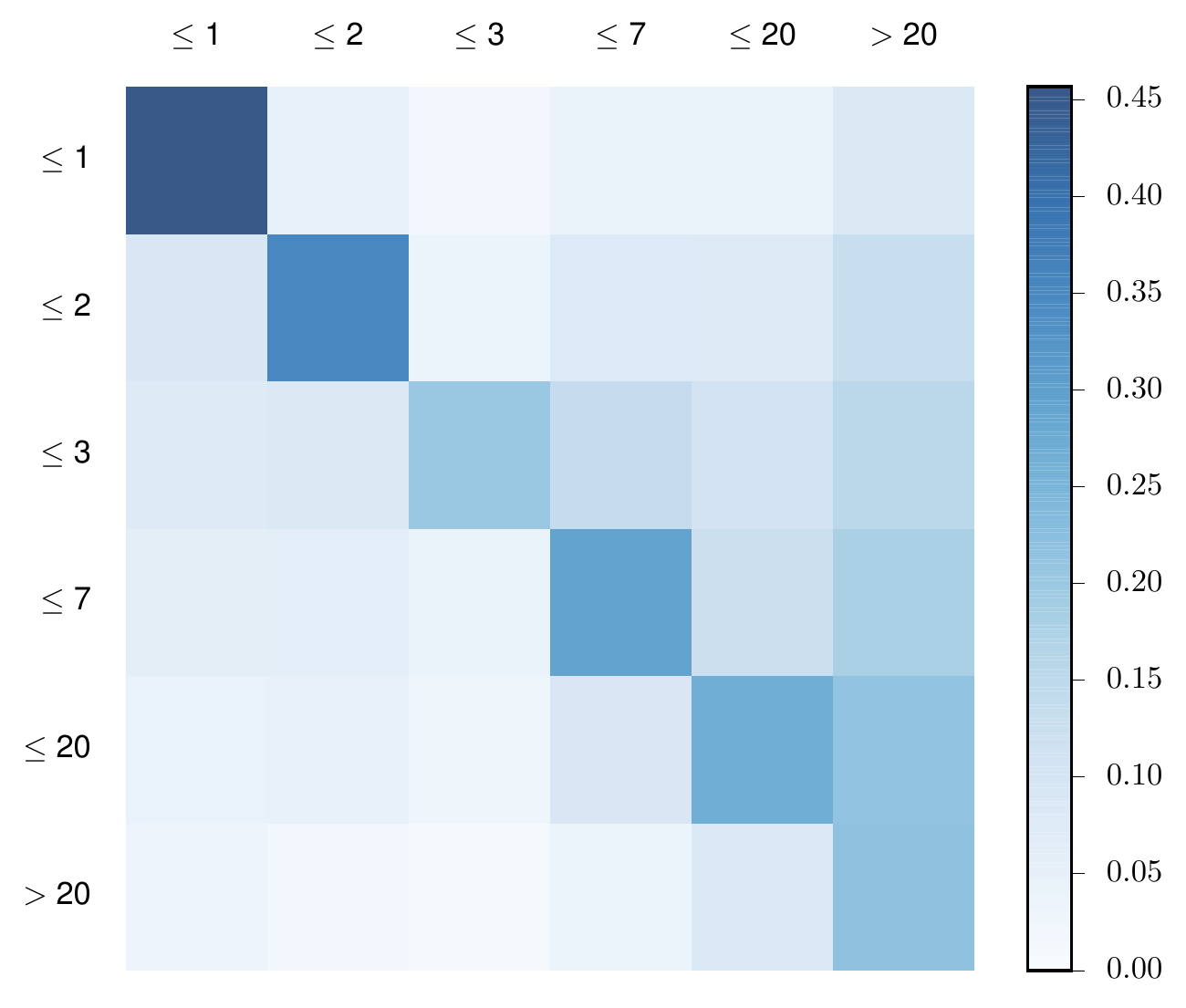}
\includegraphics[width=0.49\textwidth]{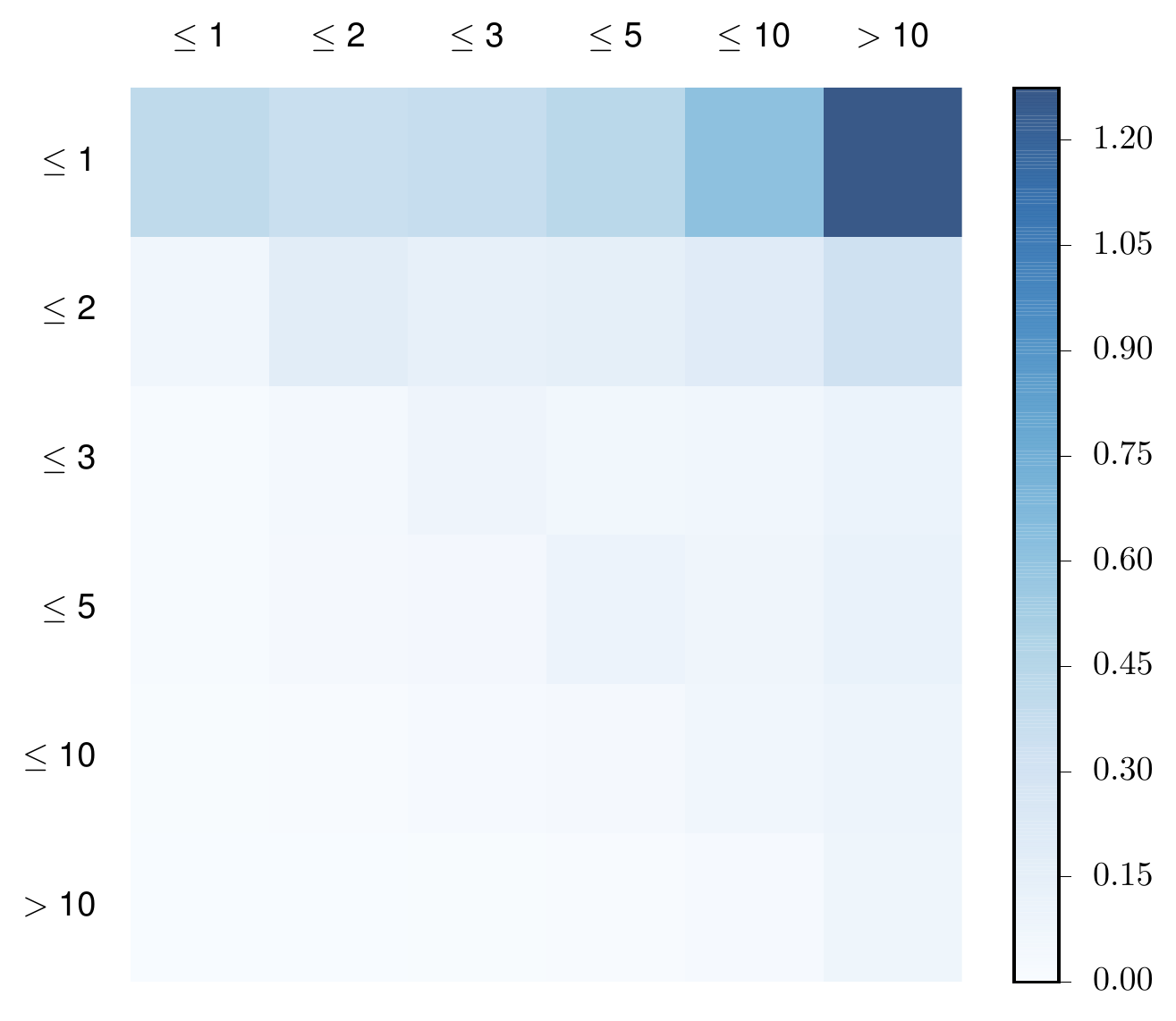}
\includegraphics[width=0.49\textwidth]{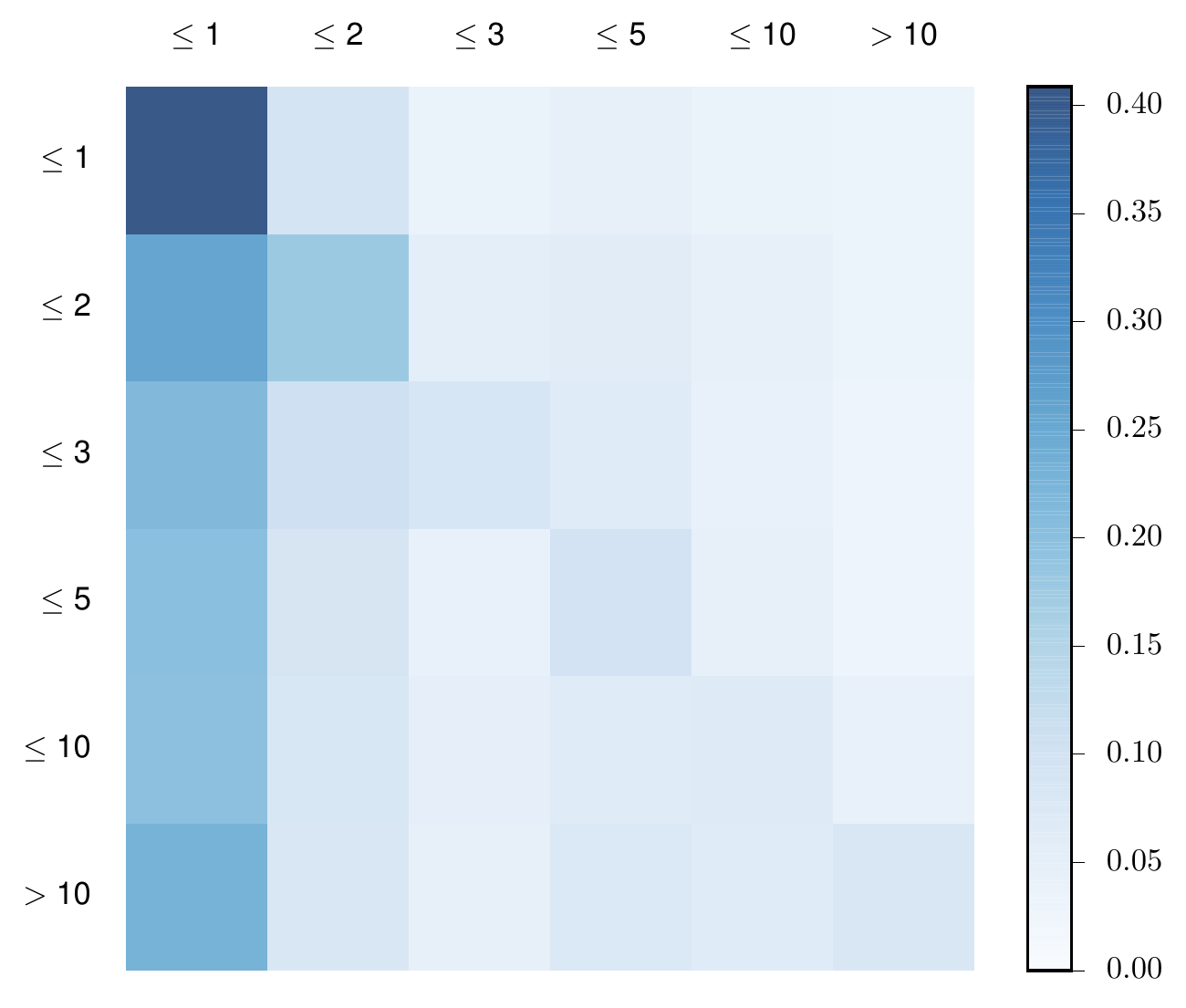}
\caption{Kernel norms. Top: Bund future. Bottom: DAX future. The left figures represent the kernel norms, while in the right side figures the kernel norms normalized by $\frac{\Lambda_j}{\Lambda_i}$ are shown.}
\label{fig:dax_norm}
\end{figure}
\end{center}
\begin{table}[tb]
\begin{center}
\begin{tabular}{ccccccc}
  \toprule[1pt]
  &\multicolumn{6}{c}{Bund}\\
 Bin & 1 &2 &3& $(3,7]$ & $(7,20]$ &  $>20$  \\ 
  \midrule
$R_i=\frac{\mu_i}{\Lambda_i}$ & 31.62 \% & 22.58 \% &23.76 \% & 23.63 \% &31.06 \% & 58.93 \% \\
 \bottomrule[1pt]
 & \multicolumn{6}{c}{DAX}\\
 \midrule
 Bin & 1 &2 &3& $(3,5]$ & $(5,10]$ &  $>10$  \\ 
  \midrule
$R_i=\frac{\mu_i}{\Lambda_i}$ & 34.44 \% &
 36.98 \% &
 45.84 \% &
 49.81 \% &
 49.79 \% &
 42.14 \% \\
\bottomrule[1pt]
\end{tabular}
\caption{Ratios $\mu_i/\Lambda_i$ for the two assets examined. All values in percent.}
\label{tab:binR6D}
\end{center}
\end{table}

\subsection{Signed Trades}
We now take the ``sign" of the transaction into account. If a market order hits the ask side of the book we label it as a buy trade, while if it hits the bid side we call it a sell trade. To analyze the interaction between different order sizes, we proceed as before by binning the trades according to their volume.
To keep the overall dimension of the model manageable we use four bin of volume for each side (buy/sell). The binning schemes for the two assets examined are detailed in Table \ref{tab:bin4+4D}. 

\begin{table}
\centering
\begin{tabular}{cc|cc|cc}
\toprule[1pt]
& Volumes (contracts) & Avg. N events & Fraction (\small{\%}) & Avg. N events& Fraction (\small{\%})\\ 
\midrule
$S_1$ & 1 & 5678 & 16.10& 14066 & 32.24\\ 
$S_2$ &  $(1,3]$&3558 & 10.09&5047& 11.57\\ 
$S_3$ &  $(3,10]$& 4132 & 11.72& 2340& 5.37\\ 
$S_4$ & $(10,\infty]$ &4397 & 12.47&376&0.86 \\  
$B_1$ & $1$&5371 & 15.23&14082& 32.28\\
$B_2$ & $(1,3]$&3591 & 10.18 &5019& 11.50\\ 
$B_3$ & $(3,10]$&4152 & 11.77 &2322 &5.32\\ 
$B_4$ & $(10,\infty)$&4384 & 12.43&370& 0.85\\ 
\bottomrule[1pt]
\end{tabular} 
\caption{Bund (left) and DAX (right) Future: Binning scheme and average number of events per day in each bin in the signed case.}
\label{tab:bin4+4D}
\end{table}

In Figure \ref{fig:dax_sig_ker_lin} we plot the kernels for the sell components obtained on Bund (top 4 panels) and DAX (bottom 4 panels) data, the results on the buy components are almost exactly symmetrical and thus not shown. We again observe two main peaks at 100$\mu s$ and $300\mu s$ respectively, with the latter being more relevant when large trades are involved. Now that trade sign is taken into account we immediately notice that short term influence is almost entirely due to same-sign trades, in line with the hypothesis that is the result of order splitting. Excitation from opposite side trades is negligible, except for the largest ones that show some influence around the millisecond timescale. In the DAX figures, the excitation from large trades of opposite sign over larger trades is even more evident and the cross excitement starts to become relevant around time lags of 0.5 milliseconds. A possible explanation is that at this timescale, enough time has passed for the sign of the order to become less relevant thus the contribution is towards a general increase in activity rather than targeting specifically more buying or selling. 

In Figure \ref{fig:dax_sig_norm} we plot the kernel norms for the Bund an the DAX in a color map. We can divide the matrix into the four sectors sell-sell, sell-buy, buy-sell and buy-buy. The two same-sign blocks and the two opposite-sign blocks appear identical to each other as we would expect, i.e the buy-sell symmetry is respected. Moreover, opposite-sign trades provide almost no contribution, with the notable exception of the column of largest size trades. In the Bund case, we find the same-sign blocks to be almost perfectly symmetric. Indeed, the asymmetry observed in the unsigned case, that is that large trades influence small one more than the other way round, is now found to be due to excitation from opposite-sign large trades.
In DAX future instead, each of the four main blocks is markedly asymmetric, with influence of large trades being more pronounced. As already noted this is related to the much lower frequency of large orders.
\begin{center}
\begin{figure}[]
\centering
\includegraphics[width=0.48\textwidth]{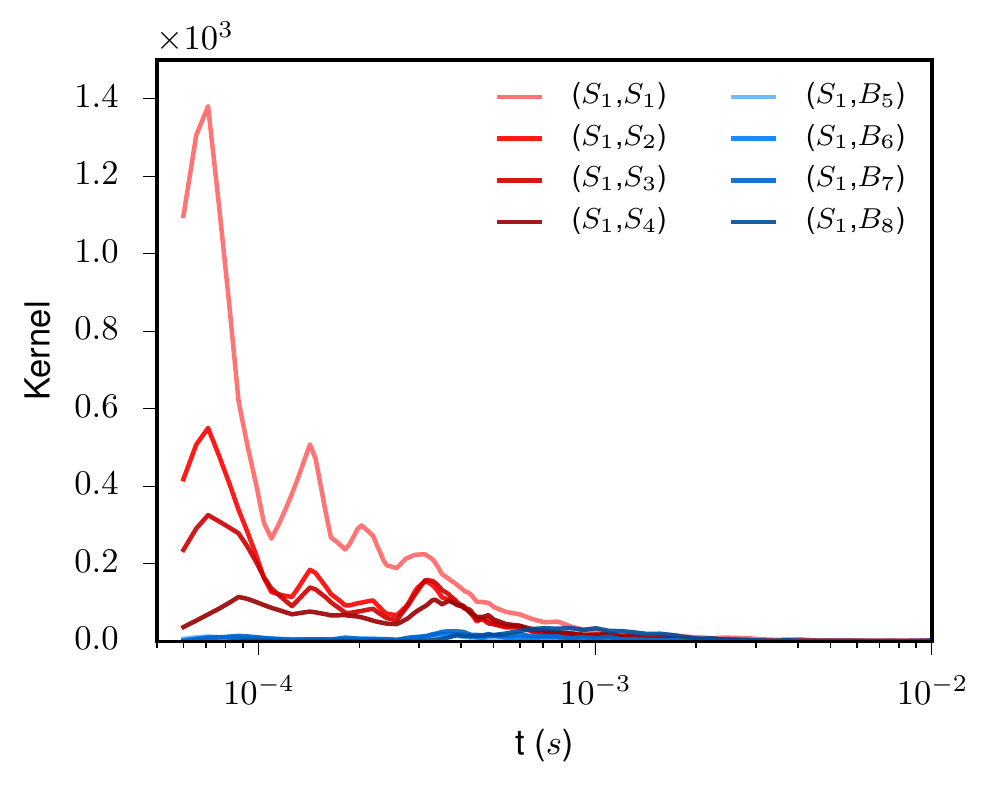}
\includegraphics[width=0.48\textwidth]{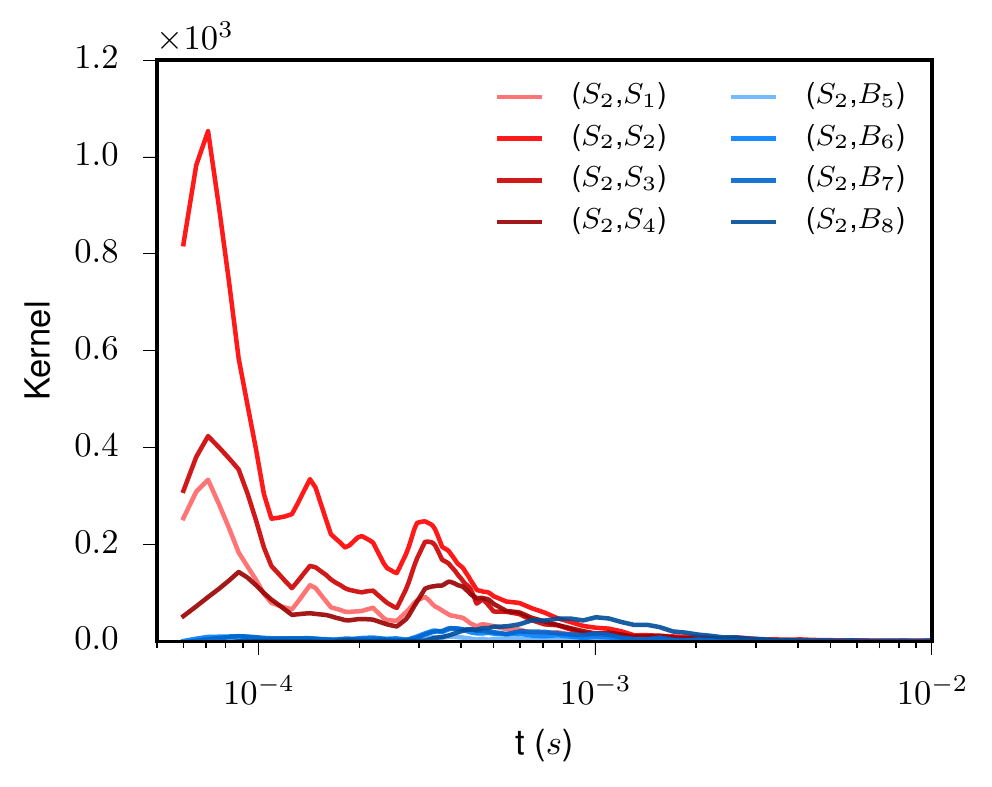}
\includegraphics[width=0.48\textwidth]{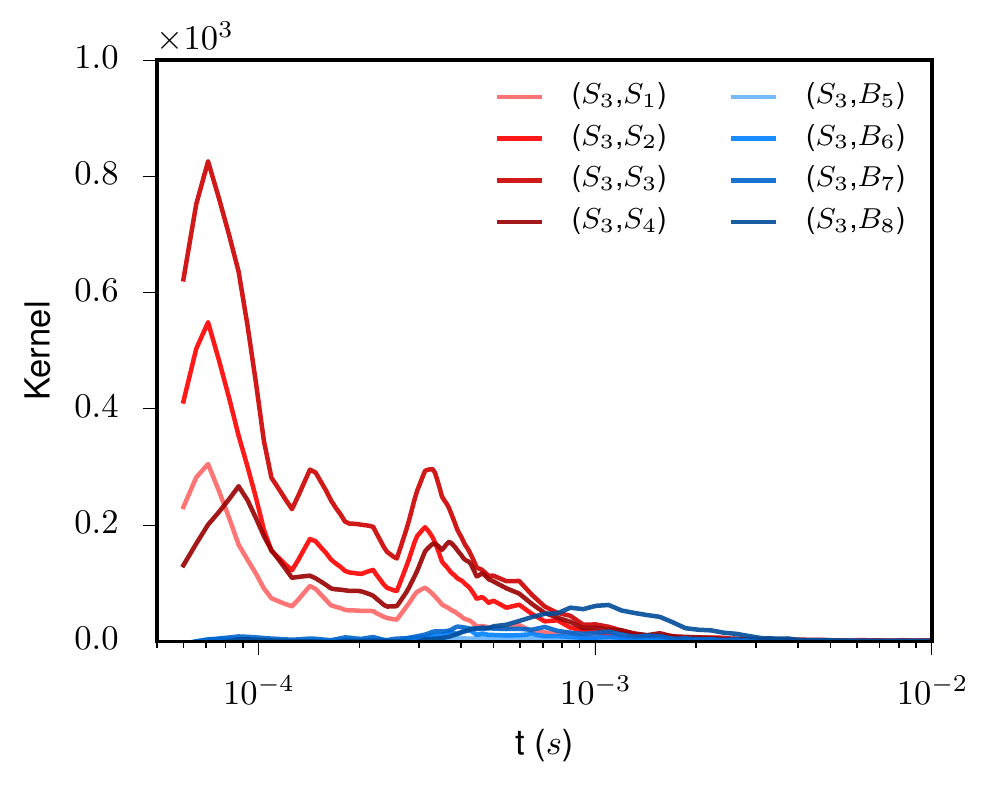}
\includegraphics[width=0.48\textwidth]{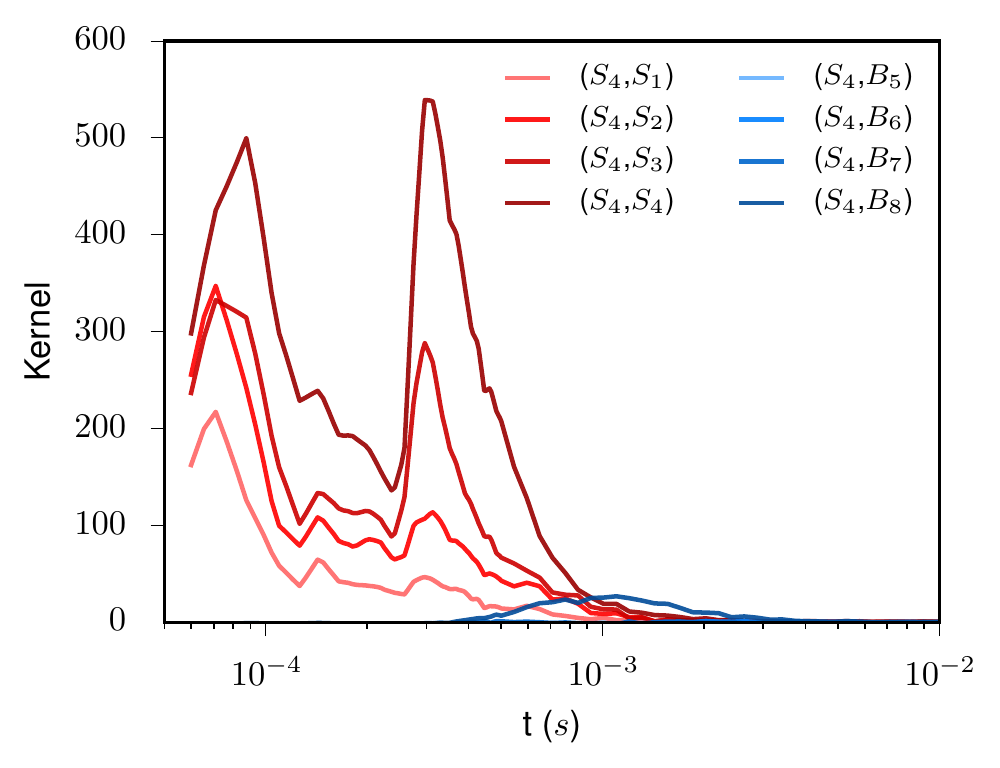}
\includegraphics[width=0.48\textwidth]{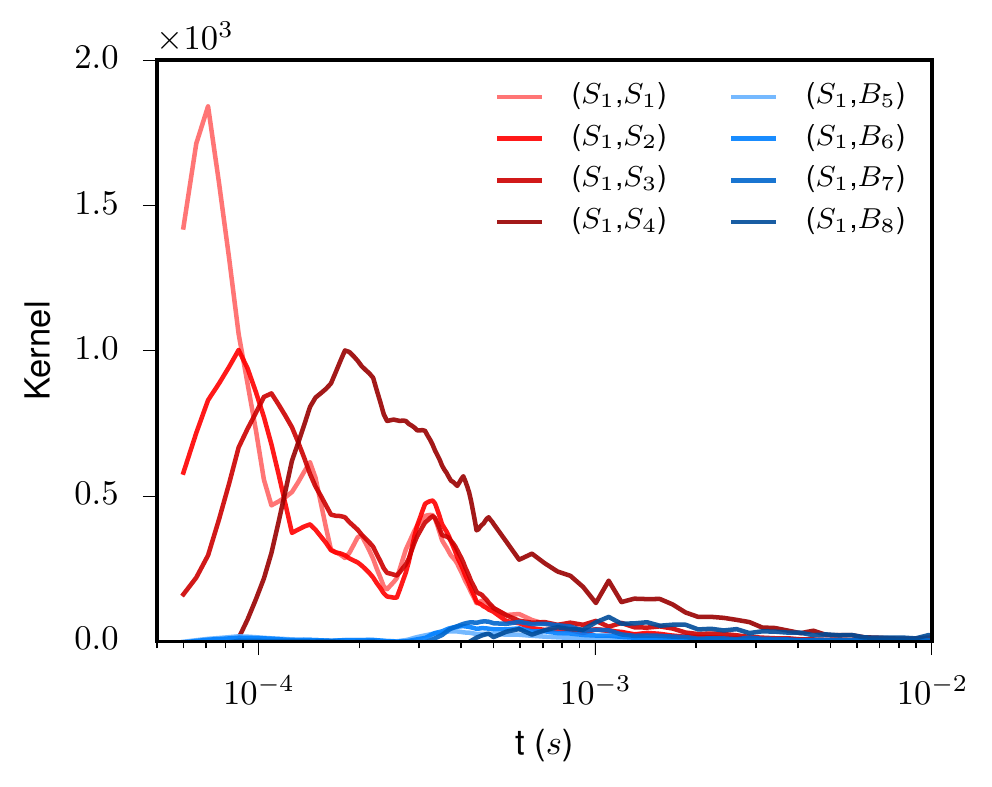}
\includegraphics[width=0.48\textwidth]{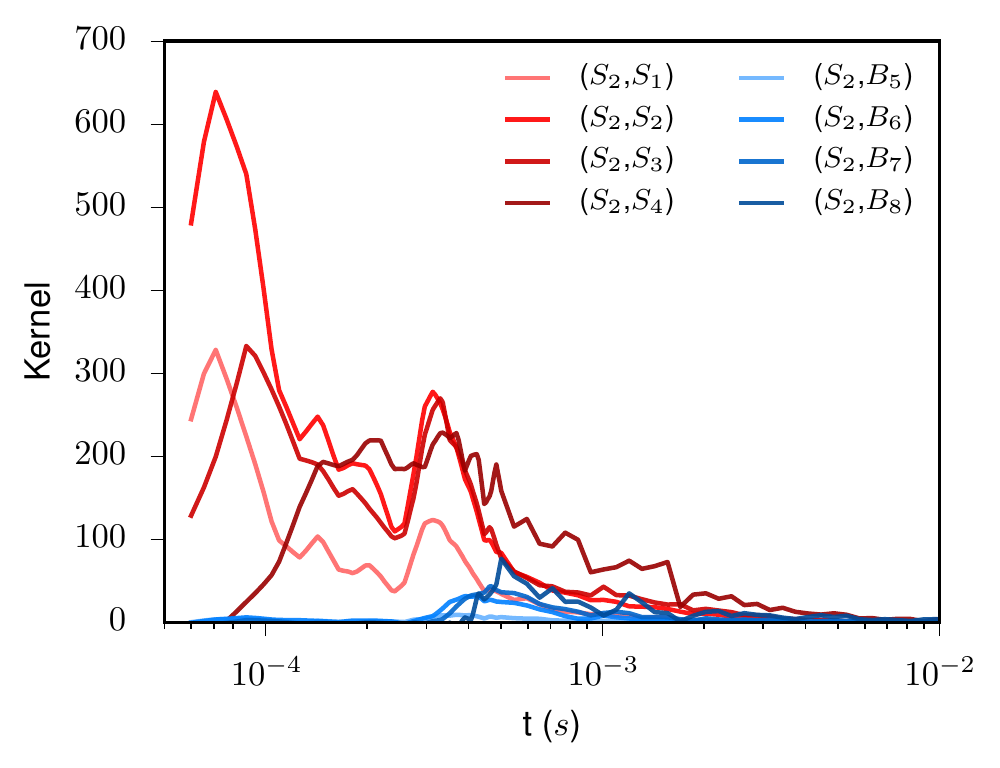}
\includegraphics[width=0.48\textwidth]{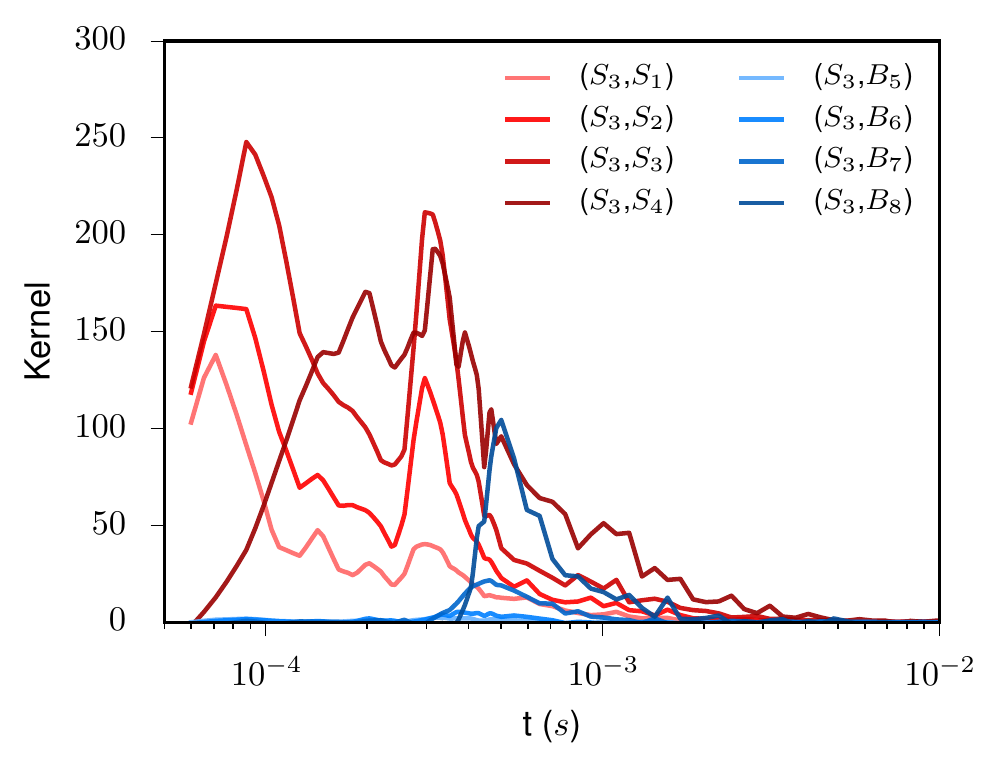}
\includegraphics[width=0.48\textwidth]{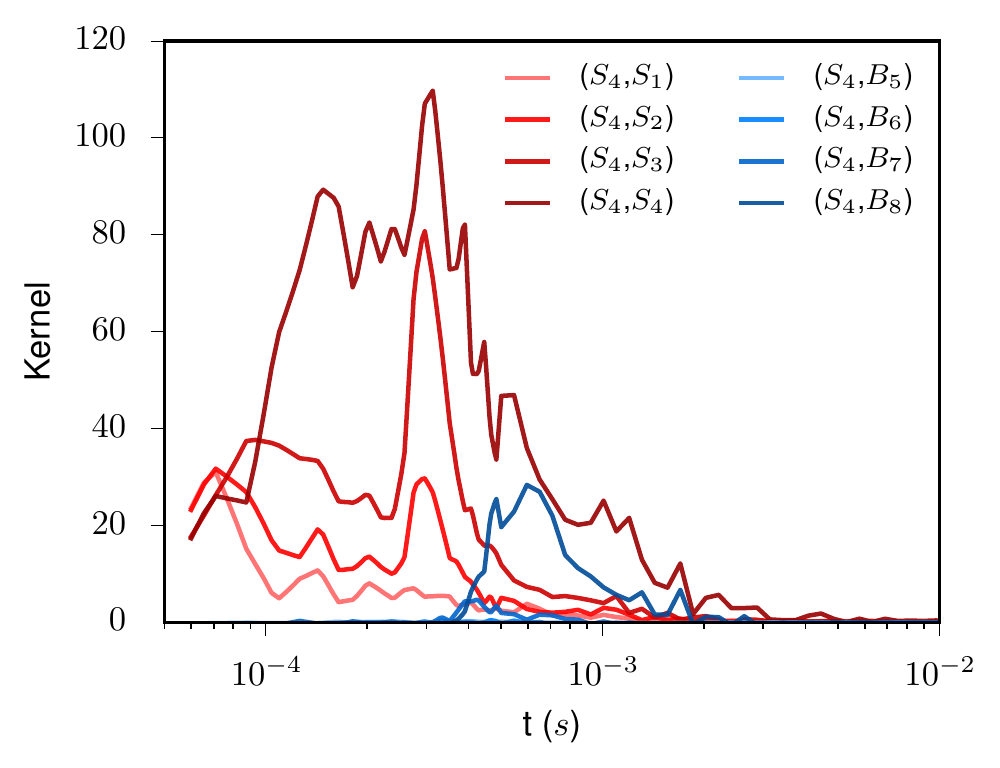}
\caption{Estimated kernels relative to the four sell components in the signed trades model. Top four panels refer to the Bund, bottom four to the DAX. }
\label{fig:dax_sig_ker_lin}
\end{figure}
\end{center}
\begin{center}
\begin{figure}[htb]
\centering
\includegraphics[width=0.49\textwidth]{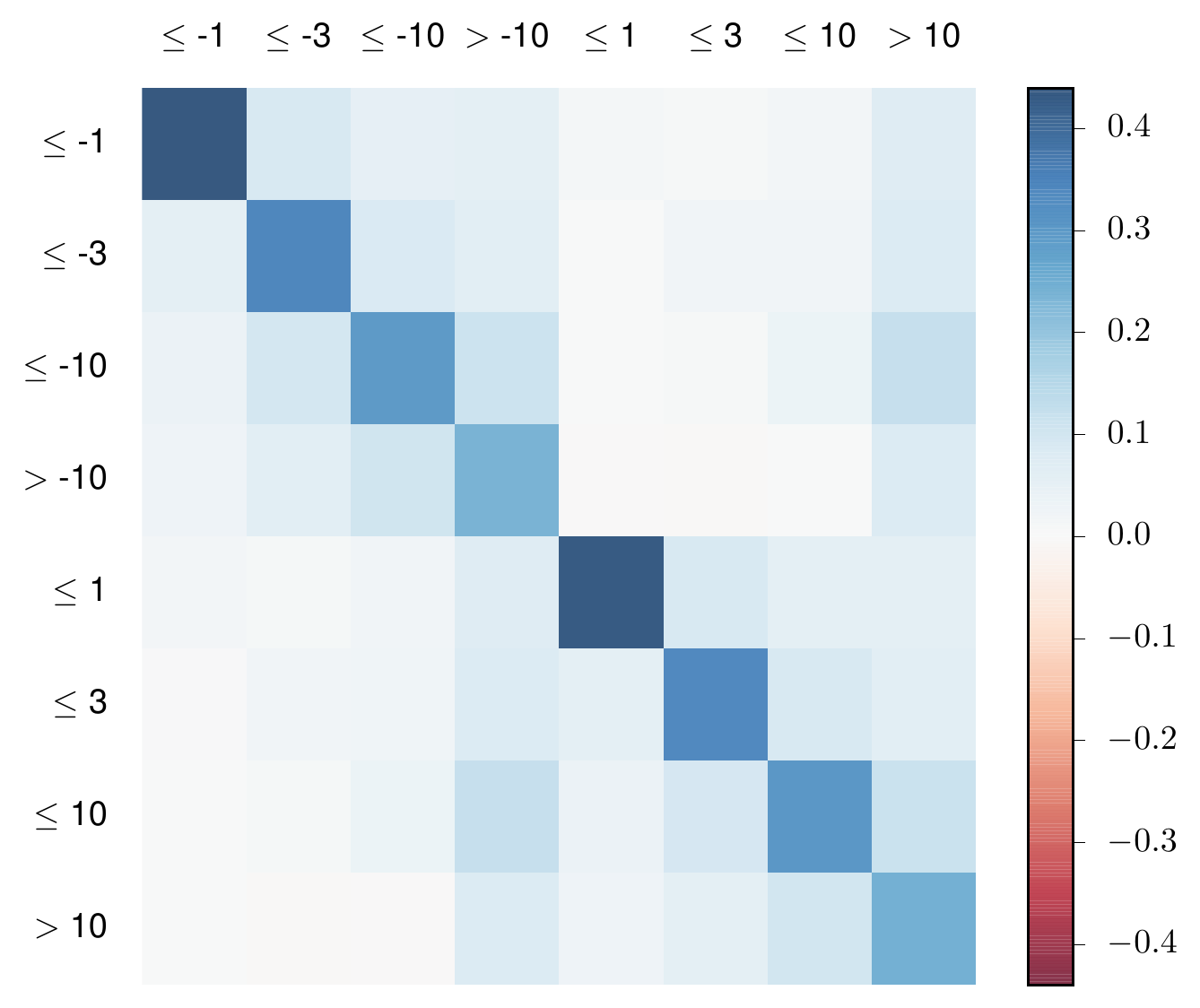}
\includegraphics[width=0.49\textwidth]{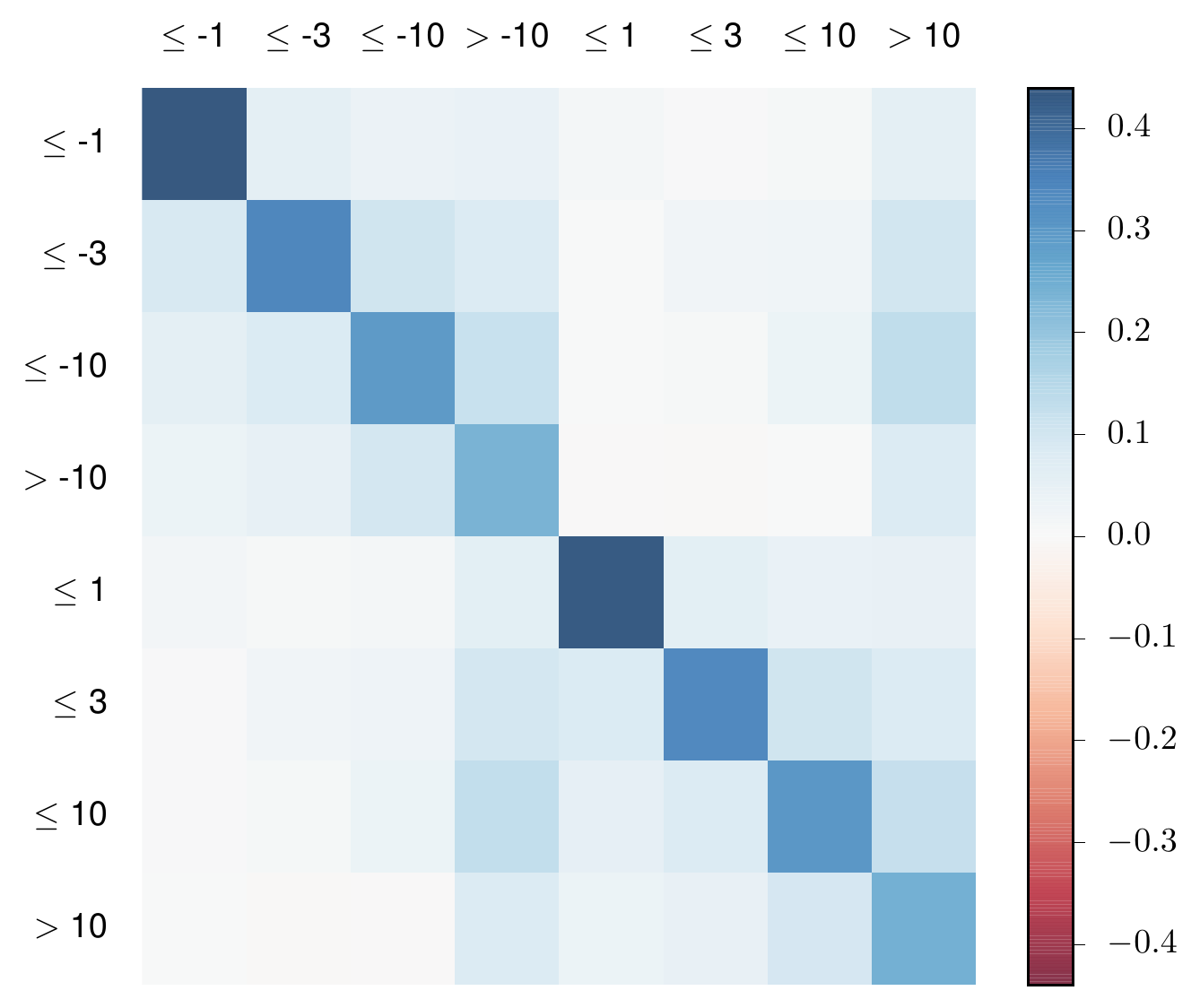}
\includegraphics[width=0.49\textwidth]{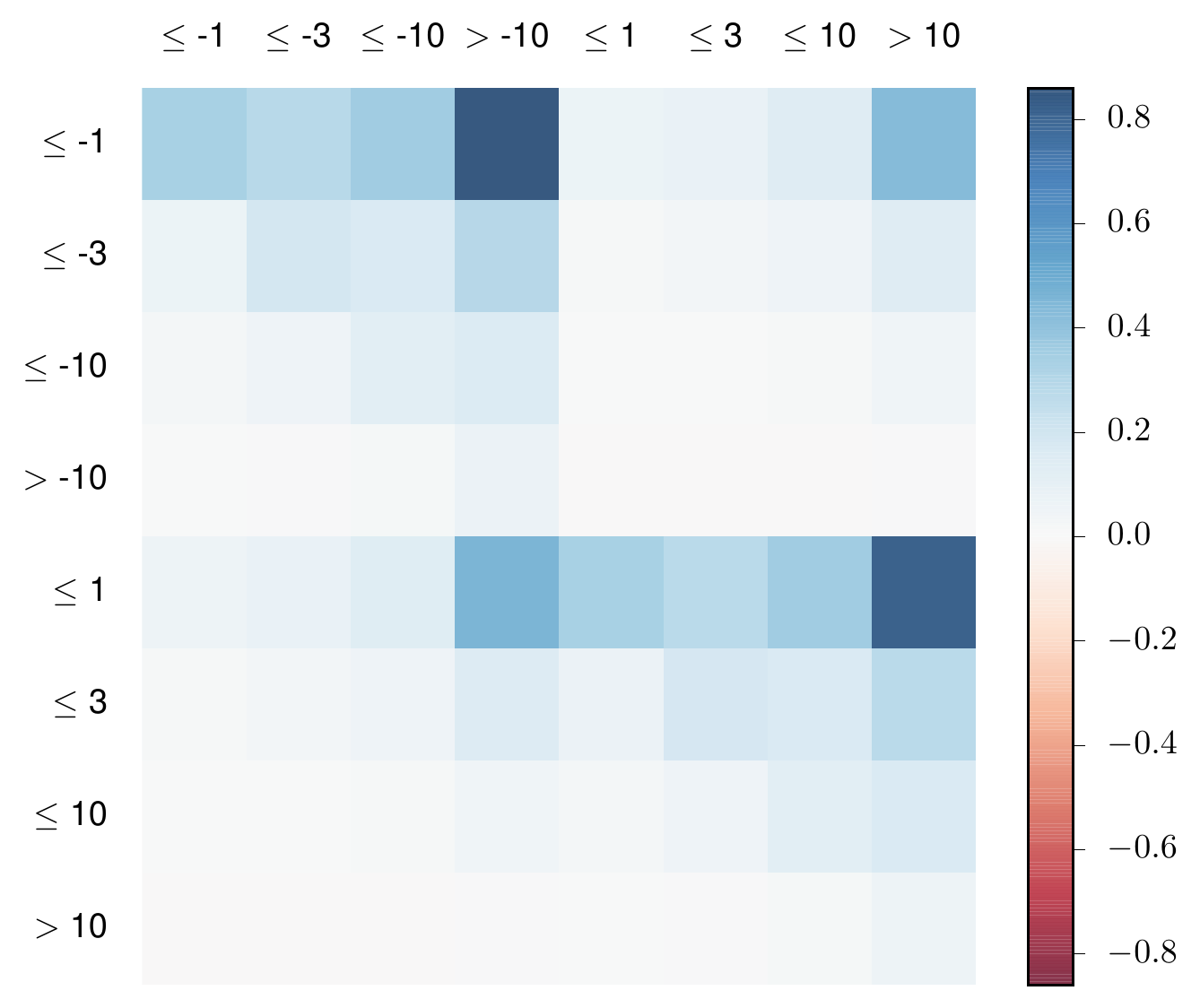}
\includegraphics[width=0.49\textwidth]{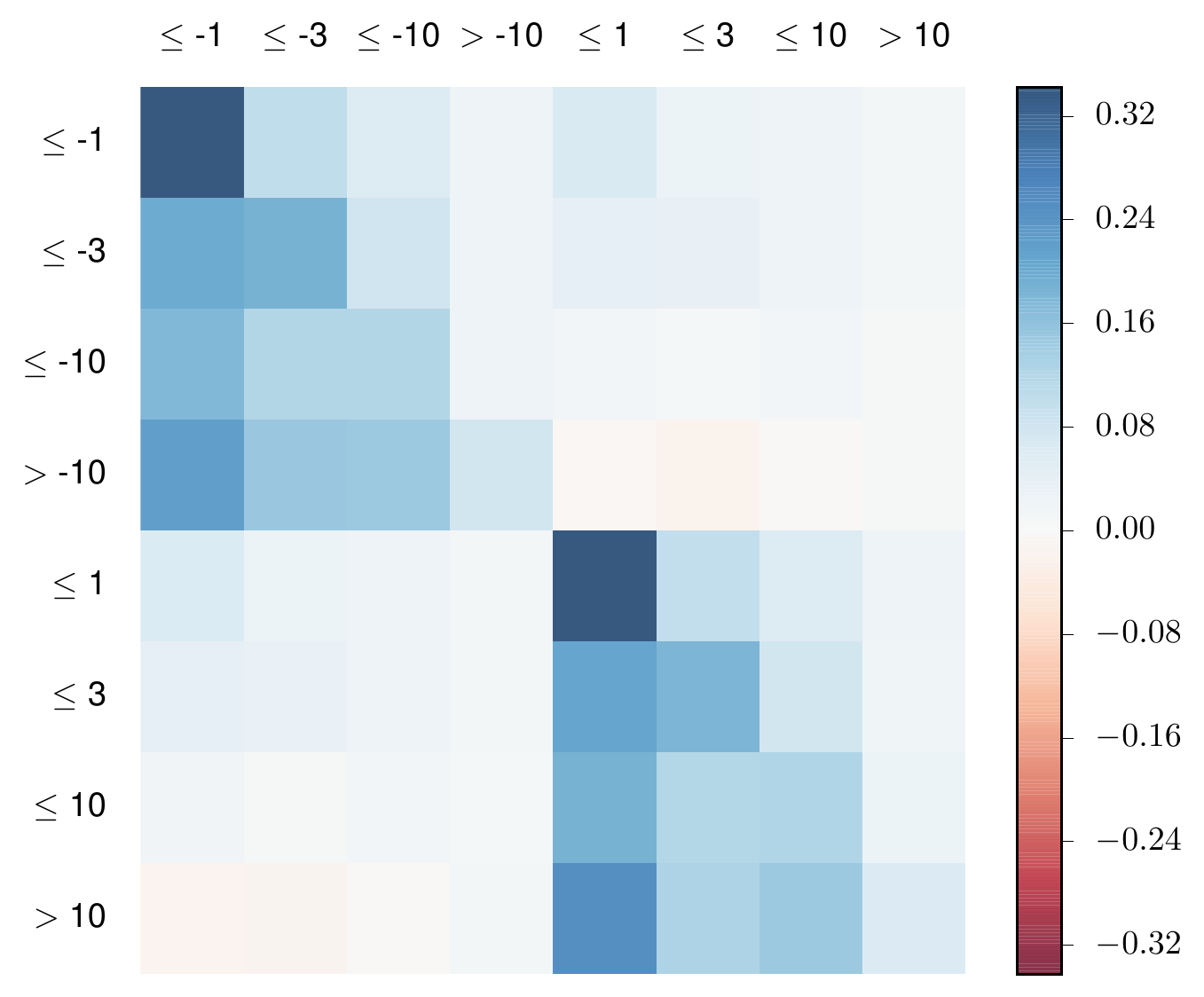}
\caption{Kernel norms (left) and rescaled kernel norms (right) for the signed trades model. Top figures refer to the Bund, bottom ones to the DAX.}
\label{fig:dax_sig_norm}
\end{figure}
\end{center}

Finally, we mention that the baseline intensities confirm what emerged for the unsigned case, namely that a significant fraction of the trades is of exogenous origin and this fraction is larger for big trades.

\section{Level one order book with volume}
\label{sec:fullob}

\begin{table}
\centering
\begin{tabular}{ccccccccccccc}
\toprule[1pt]
&\multicolumn{12}{c}{Bund}\\
\midrule
&$L^a_1$ & $L^a_2$ & $L^a_3$ & $L^a_4$ & $C^a_1$ & $C^a_2$ & $C^a_3$ & $C^a_4$ & $T^a_1$ & $T^a_2$ & $T^a_3$ & $T^a_4$\\
mil & 11.46 &15.24 &12.22 &4.41 &10.61 &13.33 &10.05 &3.66 &1.80 &1.24 &1.44 &1.53 \\
\% & 6.60 &8.78 &7.04 &2.54 &6.11 &7.68 &5.79 &2.11 &1.04 &0.71 &0.83 &0.88 \\
\midrule
&$L^b_1$ & $L^b_2$ & $L^b_3$ & $L^b_4$ & $C^b_1$ & $C^b_2$ & $C^b_3$ & $C^b_4$ & $T^b_1$ & $T^b_2$ & $T^b_3$ & $T^b_4$ \\
mil &11.49 &15.24 &12.19 &4.28 &10.64 &13.30 &9.92 &3.50 &1.88 &1.22 &1.43 &1.53 \\
\% & 6.62 &8.78 &7.02 &2.47 &6.13 &7.66 &5.71 &2.02 &1.08 &0.70 &0.82 &0.88 \\
\toprule[1pt]
&\multicolumn{12}{c}{DAX}\\
\midrule
&$L^a_1$ & $L^a_2$ & $L^a_3$ & $L^a_4$ & $C^a_1$ & $C^a_2$ & $C^a_3$ & $C^a_4$ & $T^a_1$ & $T^a_2$ & $T^a_3$ & $T^a_4$ \\
mil &29.07 &3.84 &0.71 &0.05 &30.29 &2.45 &2.53 &0.76 &5.14 &1.82 &0.82 &0.13 \\
\% & 18.80 &2.48 &0.46 &0.03 &19.59 &1.58 &1.64 &0.49 &3.32 &1.18 &0.53 &0.08 \\
\midrule
& $L^b_1$ & $L^b_2$ & $L^b_3$ & $L^b_4$ & $C^b_1$ & $C^b_2$ & $C^b_3$ & $C^b_4$ & $T^b_1$ & $T^b_2$ & $T^b_3$ & $T^b_4$ \\
mil & 28.88 &3.81 &0.71 &0.06 &29.95 &2.41 &2.51 &0.76 &5.12 &1.82 &0.83 &0.13 \\
\% & 18.68 &2.46 &0.46 &0.04 &19.37 &1.56 &1.62 &0.49 &3.31 &1.18 &0.54 &0.08 \\
\bottomrule[1pt]
\end{tabular}
\caption{Total number of events in our database, divided by type and size.}
\label{tab:full_ob_counts}
\end{table}

In this section we extend the model beyond trades including a total of three type of events, namely limit orders, cancellations, and trades. For each one, we differentiate between ask and bid side and we consider four bins of volume with edges at 1, 3, and 10 contracts. The resulting Hawkes model has thus 24 components, and 576 kernels. The total number of events recorded in each component is detailed in Table \ref{tab:full_ob_counts} for both assets. We use the same estimation procedure as before to recover the kernels and the baseline intensities. Finally, we indicate with $L^x_i$ ($x\in \{a,b\}$, $i=1,..,4$) the limit order events at the ask ($x=a$) or at the bid ($x=b$) with volume in bin $i$. Similar notation is used for cancellations ($C^x_i$) and trades ($T^x_i$).

\subsection{Kernel norms}
\begin{figure}[tbh]
\centering
\includegraphics[width=.48\textwidth]{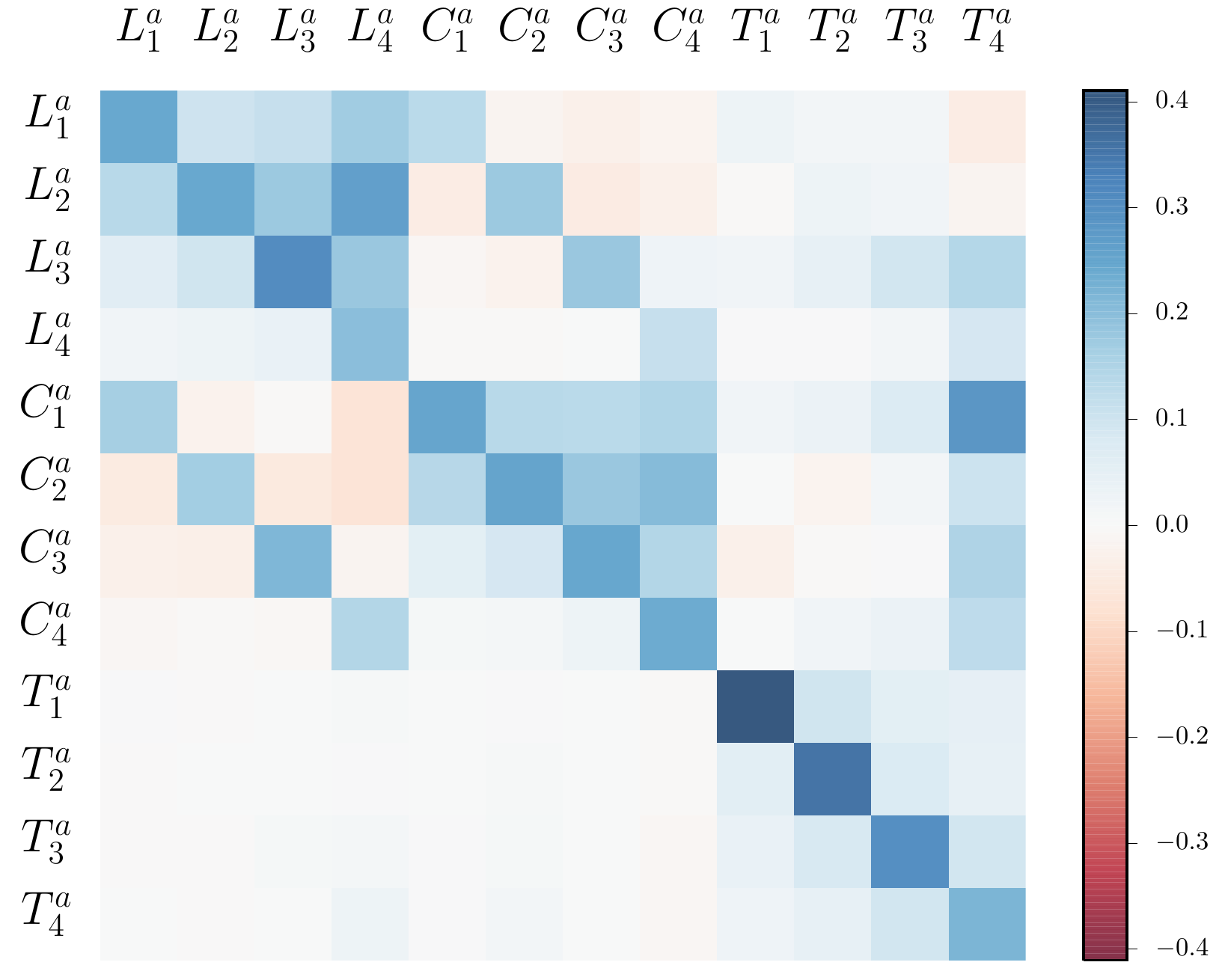}
\includegraphics[width=.48\textwidth]{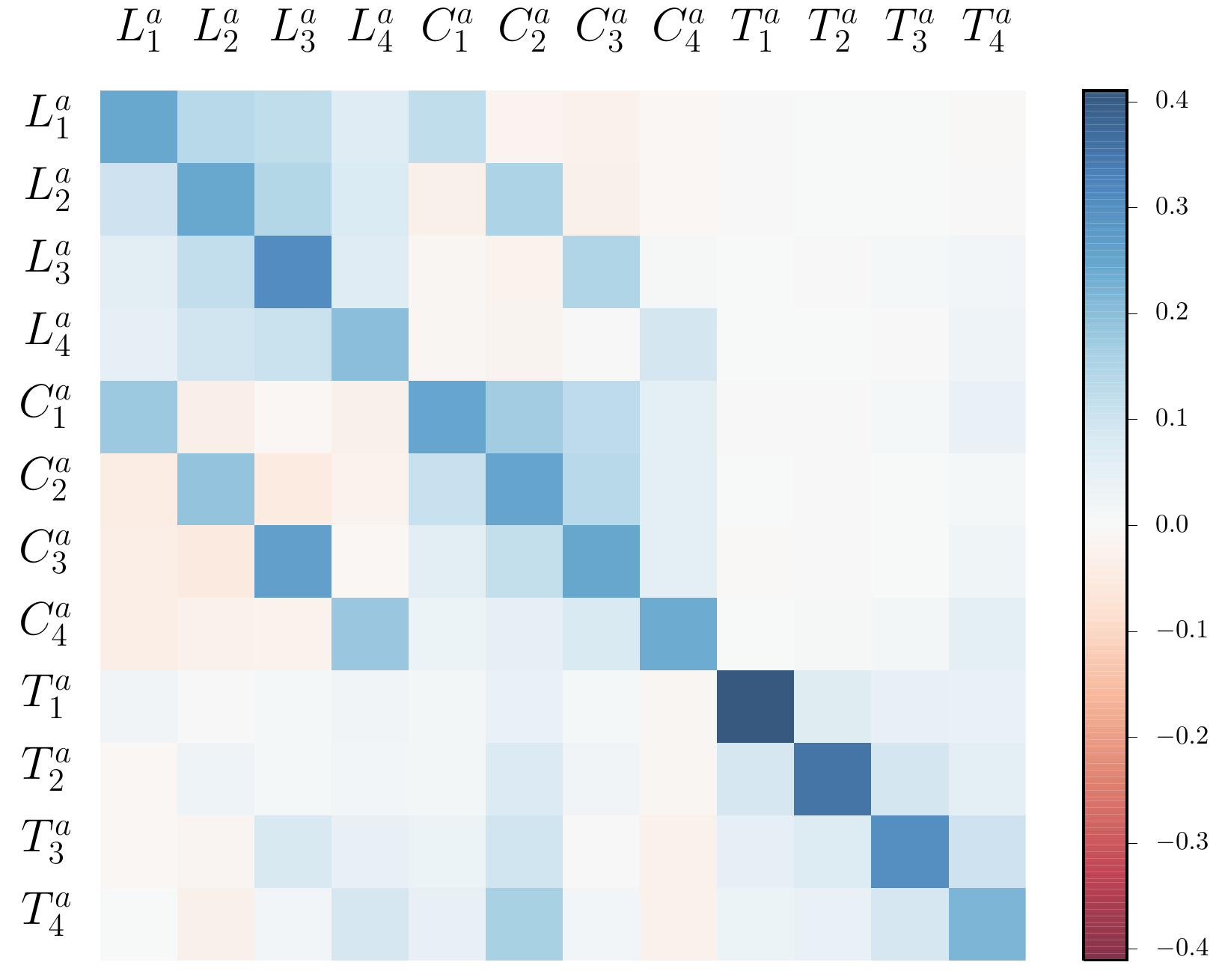}
\caption{Bund future: Ask-Ask quadrant matrix of the norm for the full order book model (left). Normalized norms (right)}
\label{fig:fullob_bund_xx}
\end{figure}
\begin{figure}[tbh]
\centering
\includegraphics[width=.48\textwidth]{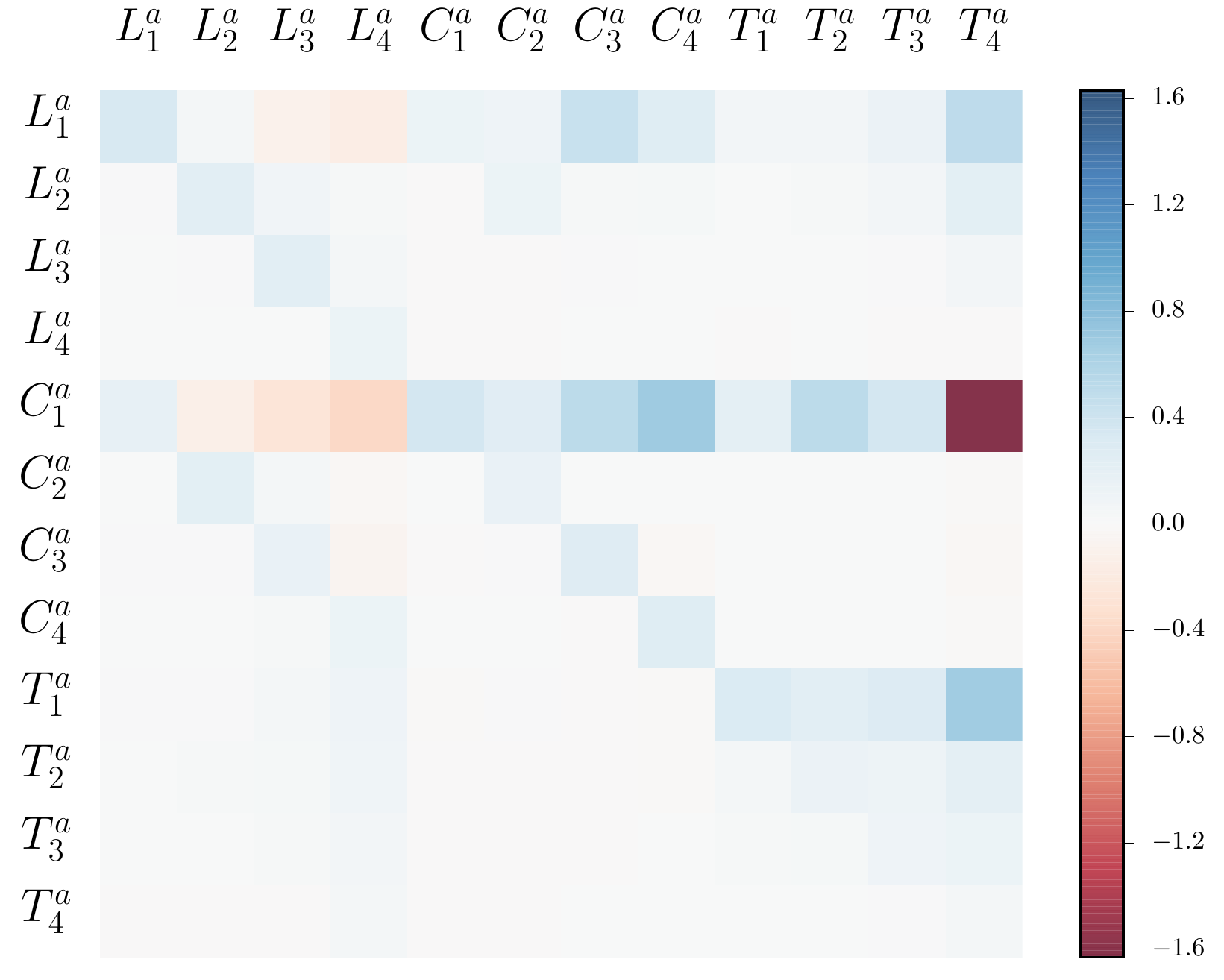}
\includegraphics[width=.48\textwidth]{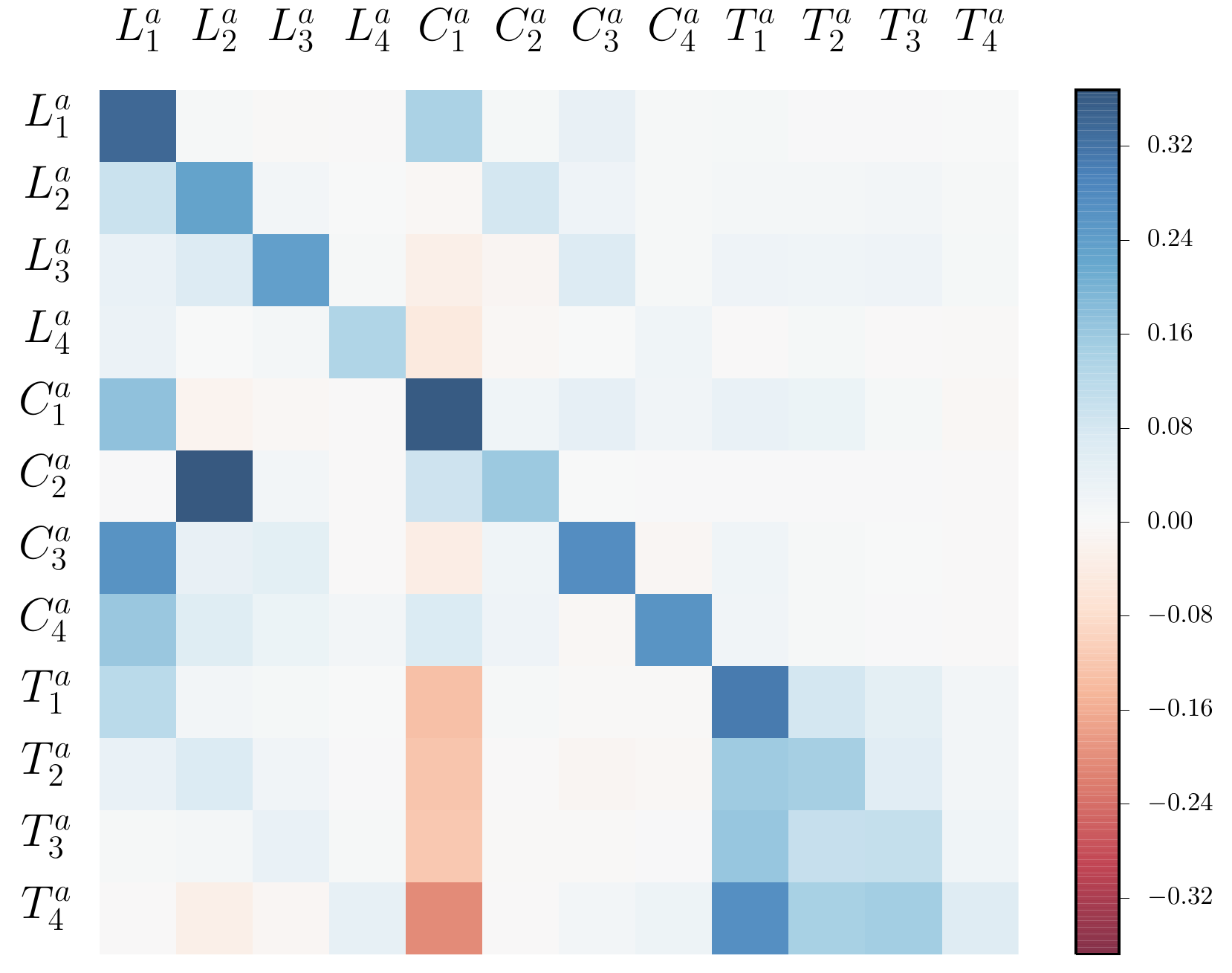}
\caption{DAX future: Ask-Ask quadrant matrix of the norm for the full order book model (top). Normalized norms (bottom)}
\label{fig:fullob_dax_xx}
\end{figure}

We start our analysis by examining the matrix of the norms $\Vert\phi^{ij} \Vert_1$, that give us a summary of the main mutual influences detected in the order book with our procedure. Since we find the bid/ask symmetry to be fairly well respected, we plot only the quadrants Ask$\rightarrow $Ask and Bid$\rightarrow $Ask in order to improve readability.

We first examine the detected relationship between same side orders. In Figures \ref{fig:fullob_bund_xx} and \ref{fig:fullob_dax_xx} we plot the ask/ask quadrant, along with the rescaled version, for the Bund and Dax future, respectively. We stress once again that in the case of the DAX the four volume components have very different average intensities. As a result, the normalized norms matrix looks very different from the original norm matrix. 
We note that diagonal terms, corresponding to self excitation, play an important role in the dynamics. Also in \citep{bacry2014estimation} diagonal terms (limit/limit, cancel/cancel, trade/trade) were found to be important. By adding order size to the analysis, we determine that this interaction is stronger for same size orders.
%We remark that in this analysis, self-excitation is found to be mostly specific to same size orders. This may be thus the result of order splitting \cite{lillo2005theory}. 
 
Another feature, more visible for the Bund, concerns the blocks $L^x_i \rightarrow C^x_j$ and $C^x_i \rightarrow L^x_j$. They describe the influence of cancellations on limits on the same side of the book (and vice-versa), and present positive diagonal terms and slightly negative and mostly negligible off diagonal terms. This means that excitation is found only between orders of the same (or similar) size. The presence of strong excitation between orders of the same size appears very reasonable. Indeed, one can only cancel orders that he placed before. The fact that the same structure is found in both the blocks $L^x \rightarrow C^x$ and $C^x \rightarrow L^x$ suggests that these blocks reflect repositioning of limit orders, for example because the price has moved.

We already discussed in the previous section the structure of the trade-trade interaction. Here we note that trades appear to be mostly influenced by other trades with very little excitation from other types of orders. This is somewhat expected, since market orders have been found in different studies \citep{MuniToke:2015eh} to trigger limit orders. Trades, especially large ones appear to have also a significant effect on cancellations. We will discuss in more detail the effects of trades in Section \ref{sub:ob_trades}.

\begin{figure}[tbh]
\centering
\includegraphics[width=.48\textwidth]{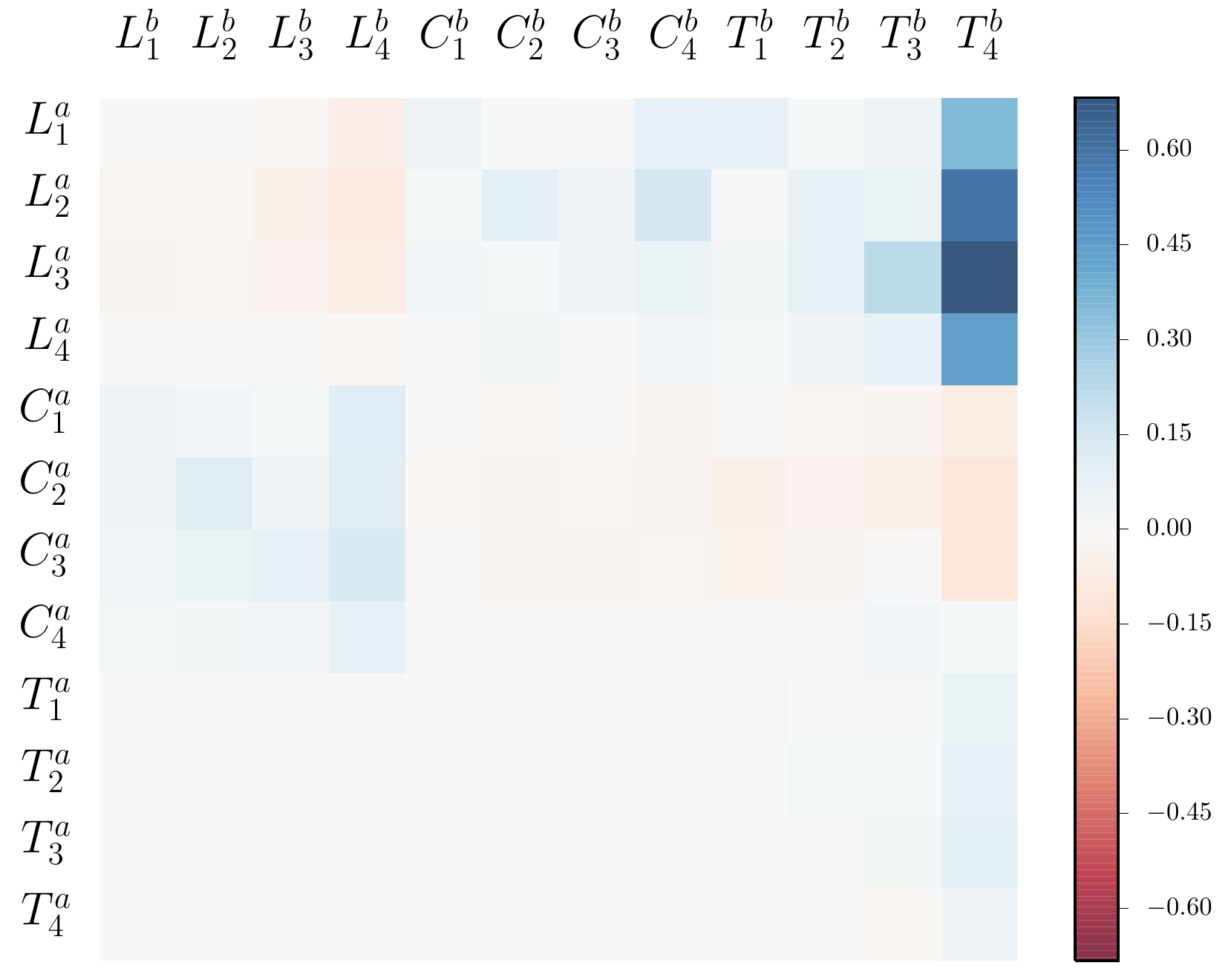}
\includegraphics[width=.48\textwidth]{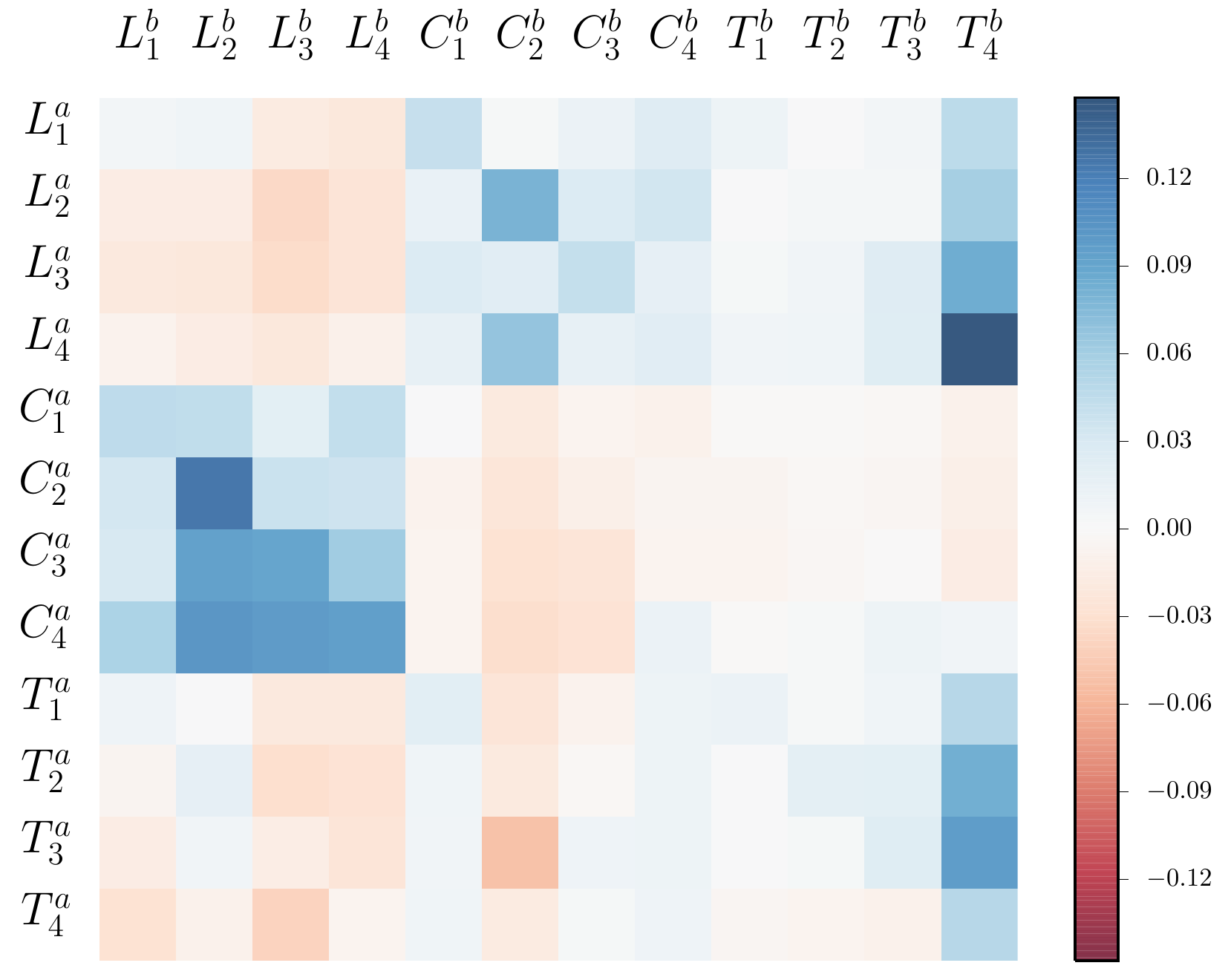}
\caption{Bund future: Ask-Bid quadrant matrix of the norm for the full order book model (left). Normalized norms (right)}
\label{fig:fullob_bund_xy}
\end{figure}
\begin{figure}[tbh]
\centering
\includegraphics[width=.48\textwidth]{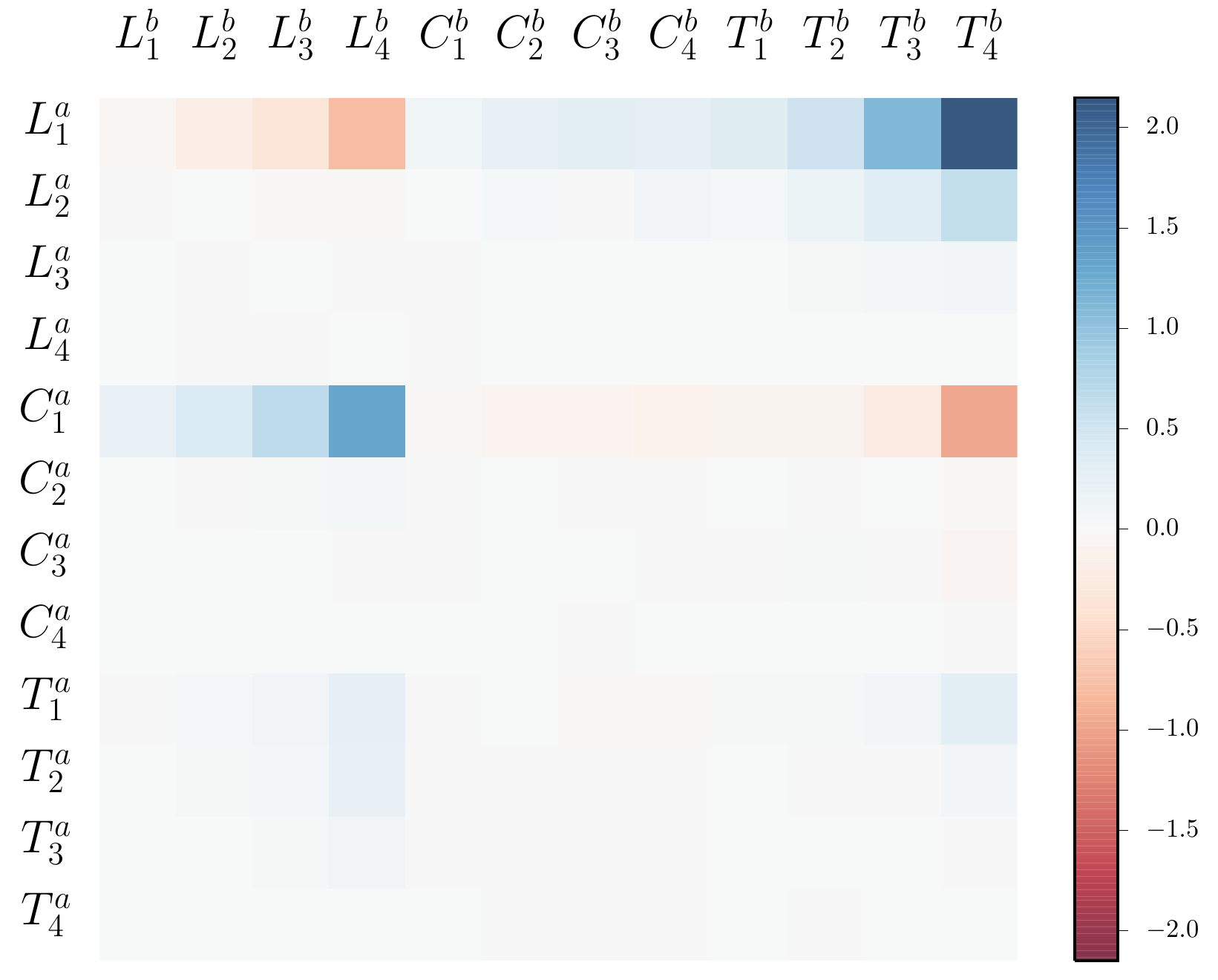}
\includegraphics[width=.48\textwidth]{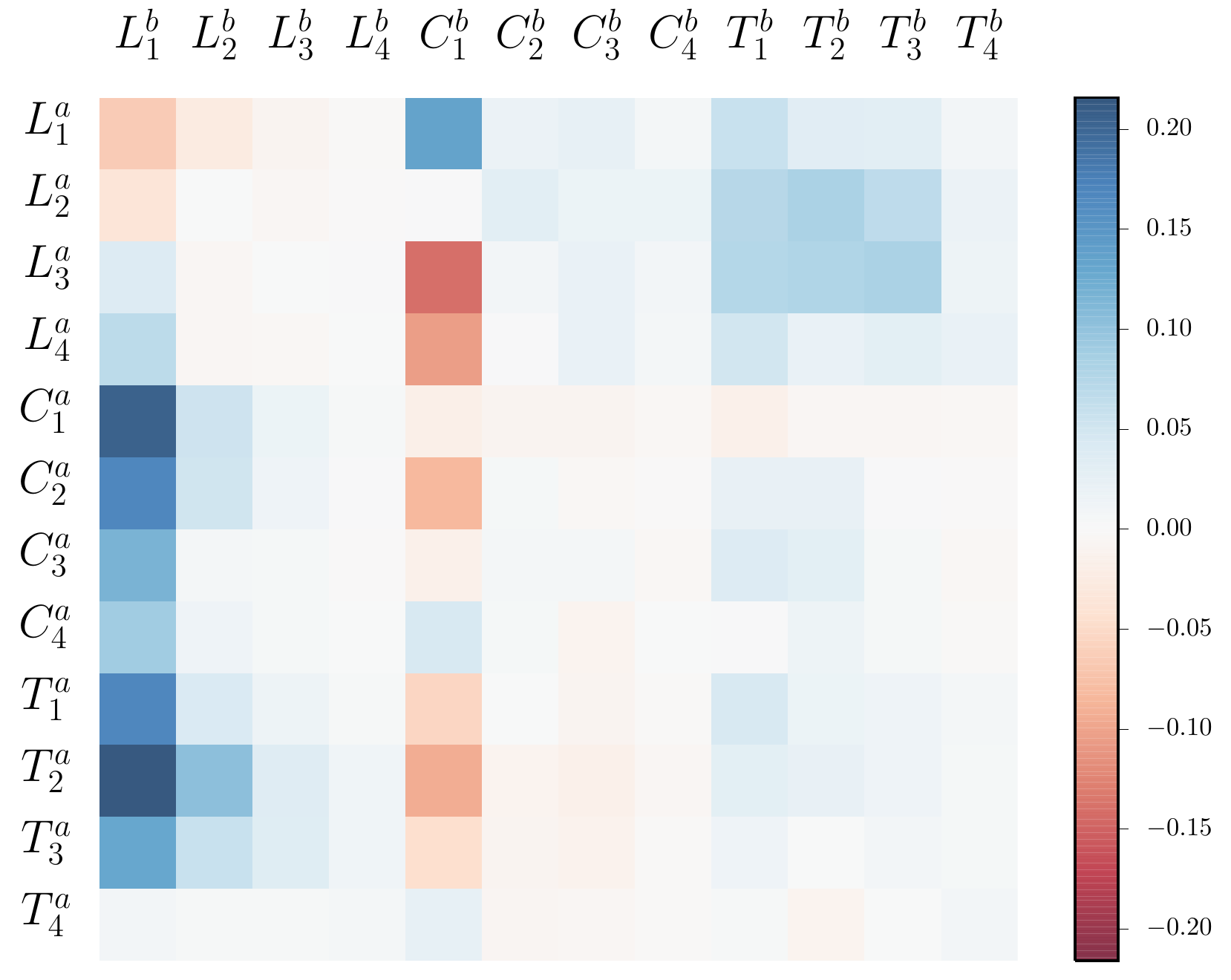}
\caption{DAX future: Ask-Bid quadrant matrix of the norm for the full order book model (top). Normalized norms (bottom)}
\label{fig:fullob_dax_xy}
\end{figure}

Figures \ref{fig:fullob_bund_xy} and \ref{fig:fullob_dax_xy}, show instead the quadrant ask/bid of the estimated kernel matrix. First, we observe that trades, particularly large ones, influence significantly limit orders on the opposite side of the book. This is more evident in the case of the DAX. The DAX is a case of a small-tick asset, so large trades are more likely to move the midprice. As a consequence, new bid limit orders are submitted in order to follow the price move upwards and vice-versa.

Finally, we note the mutual influence between limit and cancel orders on opposite side of the book and the inhibitory effect $L^x_i\rightarrow L^y_j$ and $C^x_i \rightarrow C^y_j$ ($x\ne y$), i.e. between the same order type on opposite side of the order book. This can be linked to a fair price view, in the sense that when the flow of limit orders is more intense on one side it means that the "fair price" is closer to the other side, and thus the flow of limit orders decreases. 

The norm matrix provides a summary of the mutual influence structure, however, to complete the analysis it is important to look at the precise shape of the kernels which add noteworthy information. In the following, we will examine the shape of the kernels focusing in particular on the effect of the different order sizes.

\subsection{Effects of limit orders and cancellations}

To identify the main effect of the arrival of an event of type $Z^x_j$ (e.g. limit, cancel, trade with size $j$), we plot all the kernels $\phi(Z^x_j\to \cdot)$ i.e. we plot along one column of the kernel matrix. To improve readability we separate the effect on the same side of the book from those on the opposite side.

In Figure \ref{fig:dax_limit}  we plot the kernels that describe the influence of a small and respectively a large limit order on the DAX future. Let us look first at the left column of the figure, that shows the influence on the same side of the book. We notice immediately that both a small and a large limit order trigger mainly small limit orders. However, while the small/small kernel attains high values already at very short time scales (below 300$\mu$s), the large/small one is negligible until that scale. This is a feature we already noticed in Section \ref{sec:unsigned_vol} and suggests that when a large order arrives, what follows is mainly the reaction of the market to the new information, and it takes some time to react. While the reaction part is present also for small orders, here a significant part seems to be the result of the same trader executing a series of orders at very short distances.  

Small limit orders also trigger cancellations of the same size on both sides of the book. This may be the result of traders repositioning their limit orders because e.g. the midprice has shifted. Note that the $L_1^a\to C_1^b$ kernel features high values also before the 300$\mu$s reaction time, unlike the $L_1^a\to C_1^a$ kernel. This is probably due to the fact that if a trader decides to change side on the book then he will almost simultaneously place a new limit order on one side and cancel his previous one on the other.

A large limit order also triggers many cancellations and, to a lesser extent, trades on the opposite side. A possible reading is the following. A large limit order on ask side conveys the information that the "fair price" is actually closer to the bid side. Then agents rush to cancel their outstanding limit order at the bid, but some get caught by fast 
traders that place market orders at the bid to take advantage of the situation.

Effect of cancellations, reported in Figure \ref{fig:dax_cancel} follow the same lines, albeit the effect on trades is negligible.
Similar considerations hold also in the Bund case, the main difference being that in the DAX case the occurrence of a large order (in our binning scheme) is a much rare event and thus has more dramatic consequences. 

\begin{figure}[tbh]
\centering
\includegraphics[width=.48\textwidth]{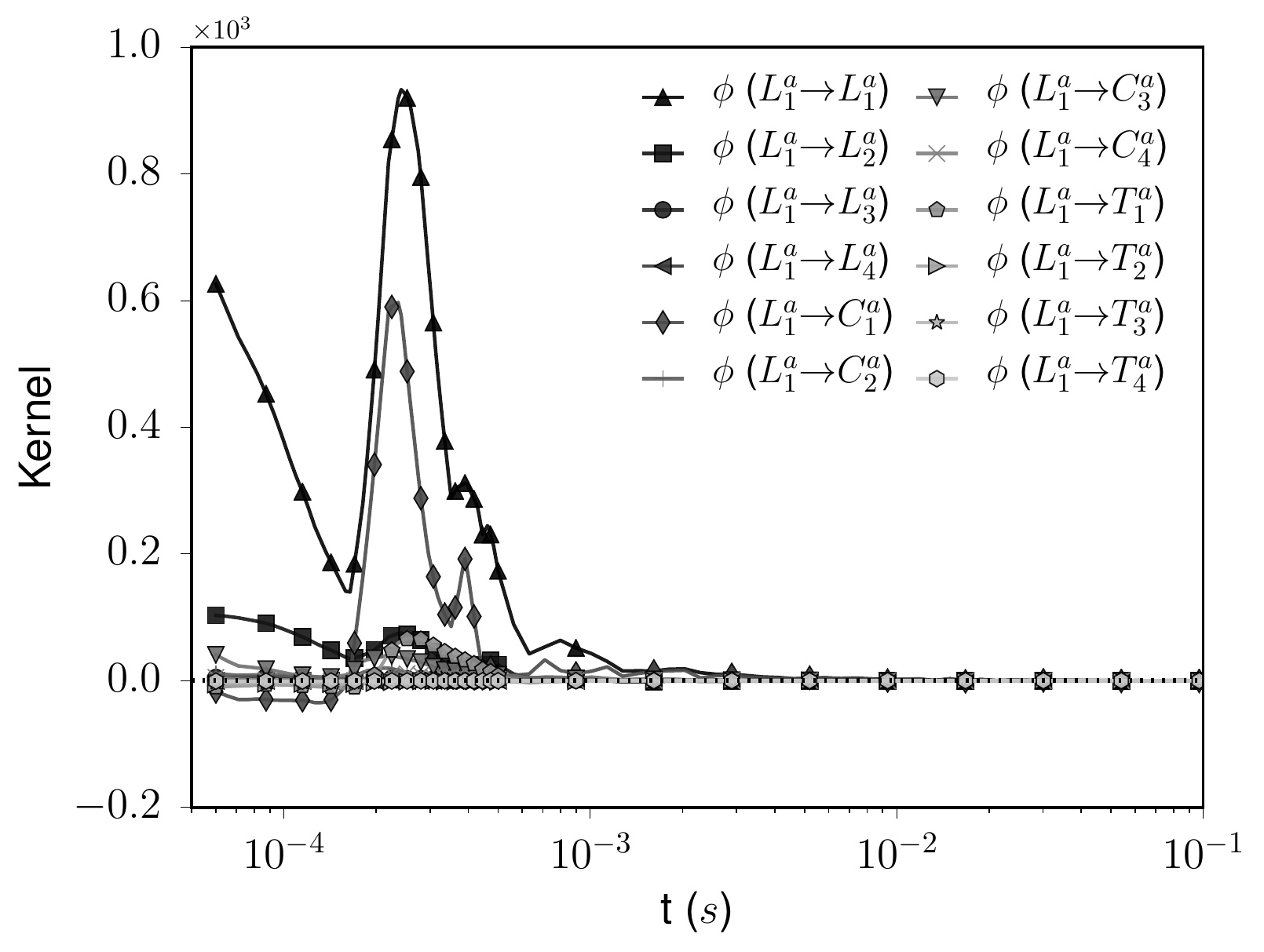}
\includegraphics[width=.48\textwidth]{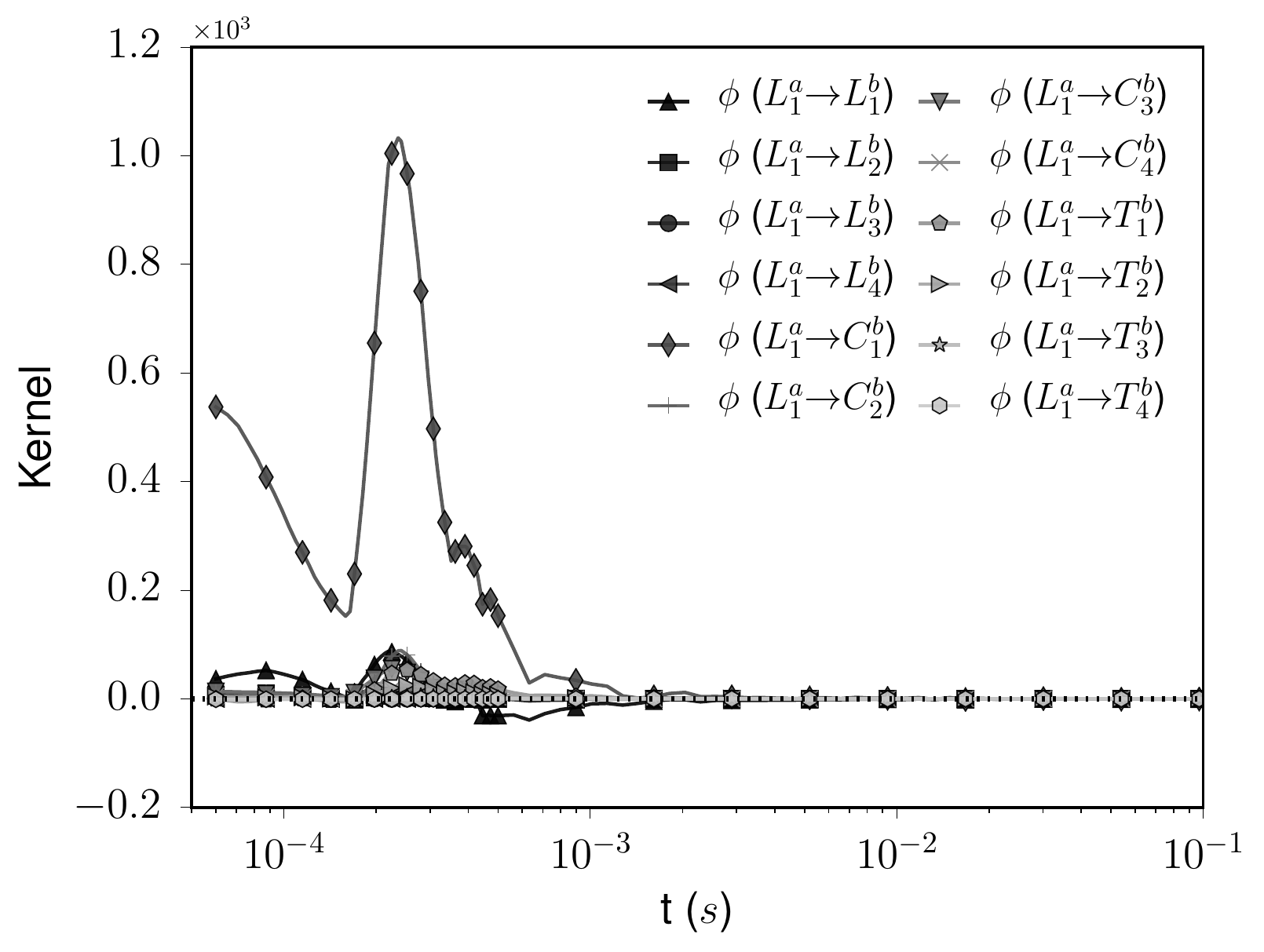}
\includegraphics[width=.48\textwidth]{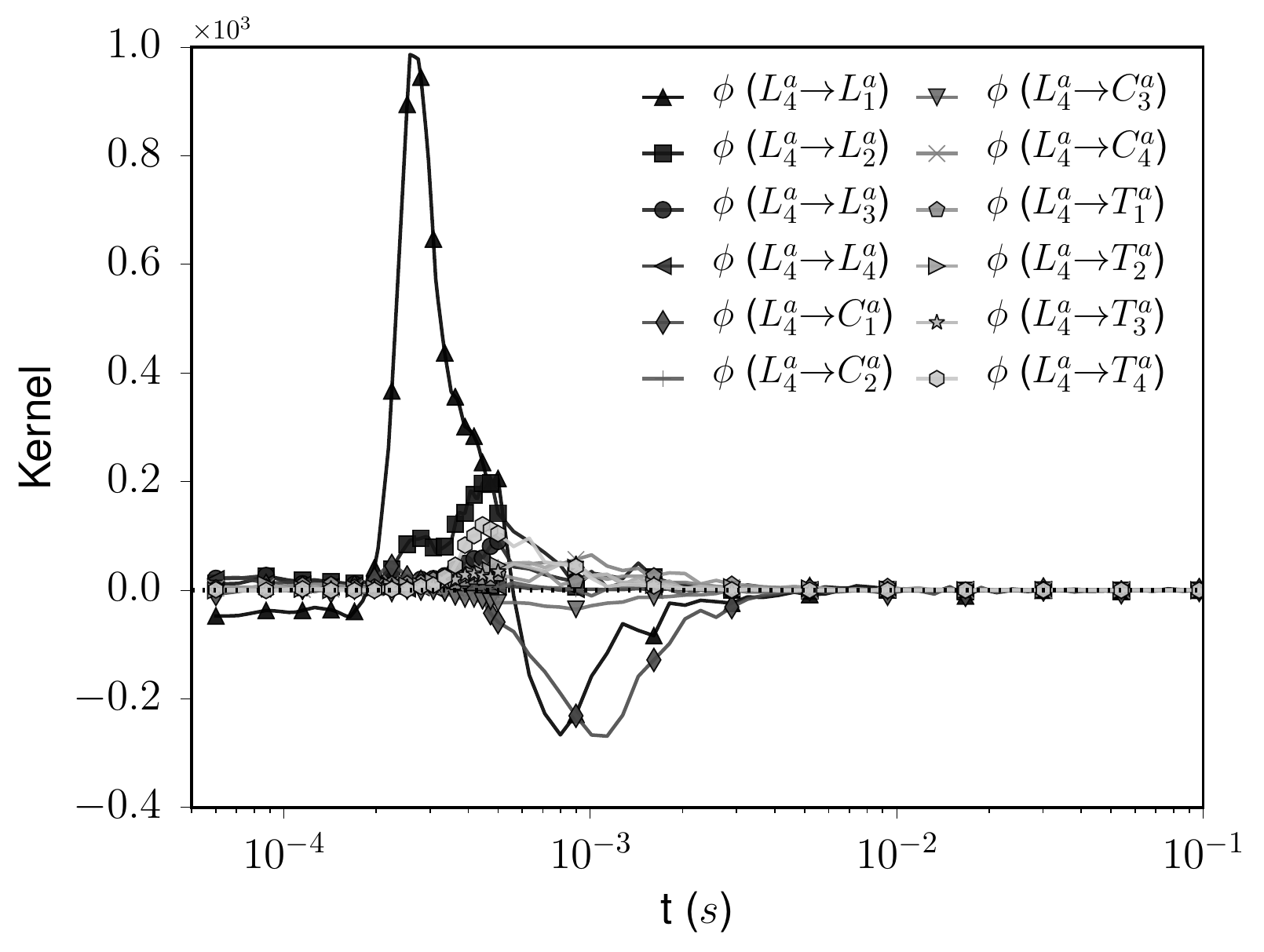}
\includegraphics[width=.48\textwidth]{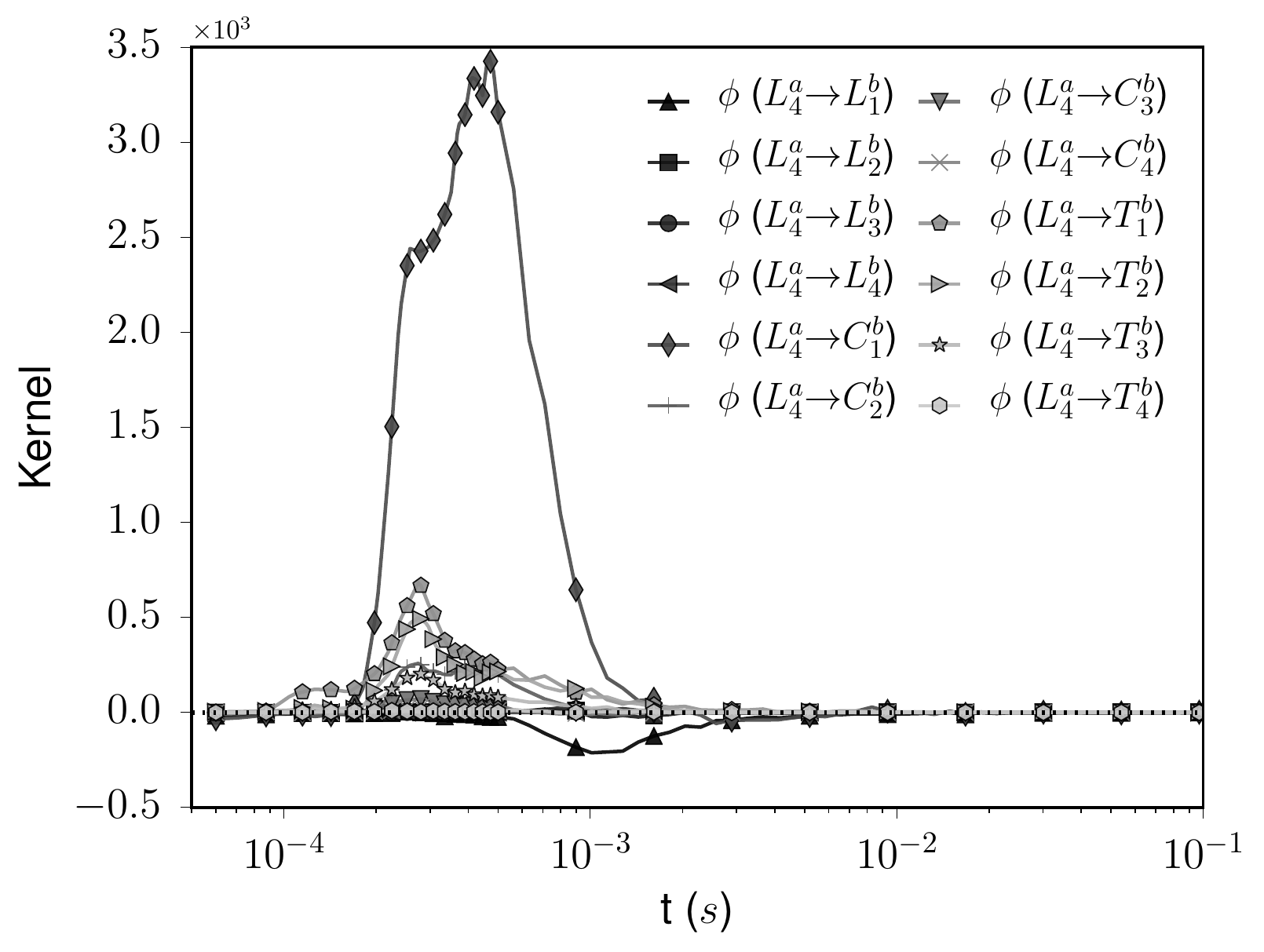}
\caption{DAX future: kernels describing the influence of a small (top) and large (bottom) limit order on other events. The left column refers to same side of the book, the right side to the opposite one.}
\label{fig:dax_limit}
\end{figure}

\begin{figure}[tbh]
\centering
\includegraphics[width=.48\textwidth]{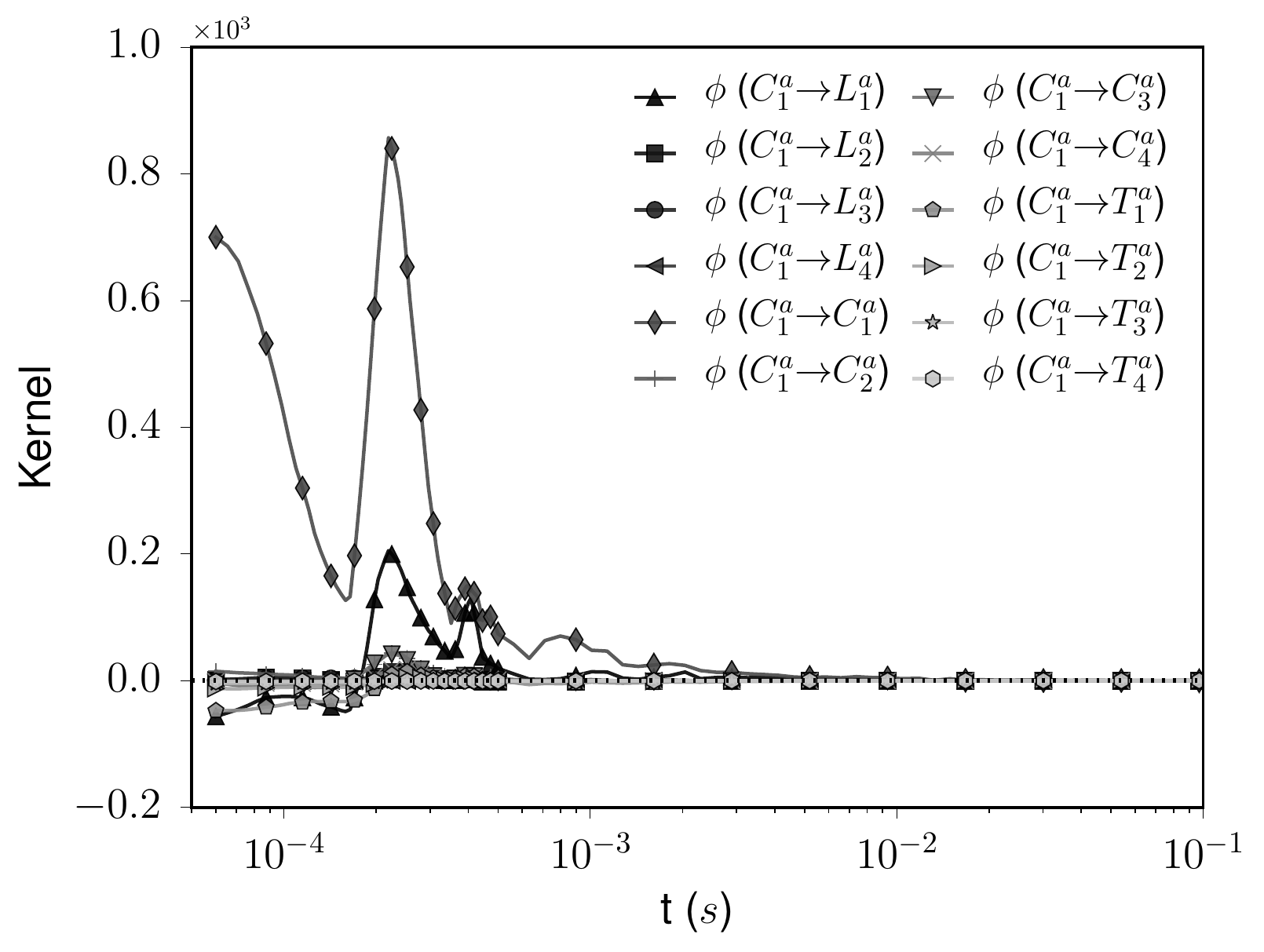}
\includegraphics[width=.48\textwidth]{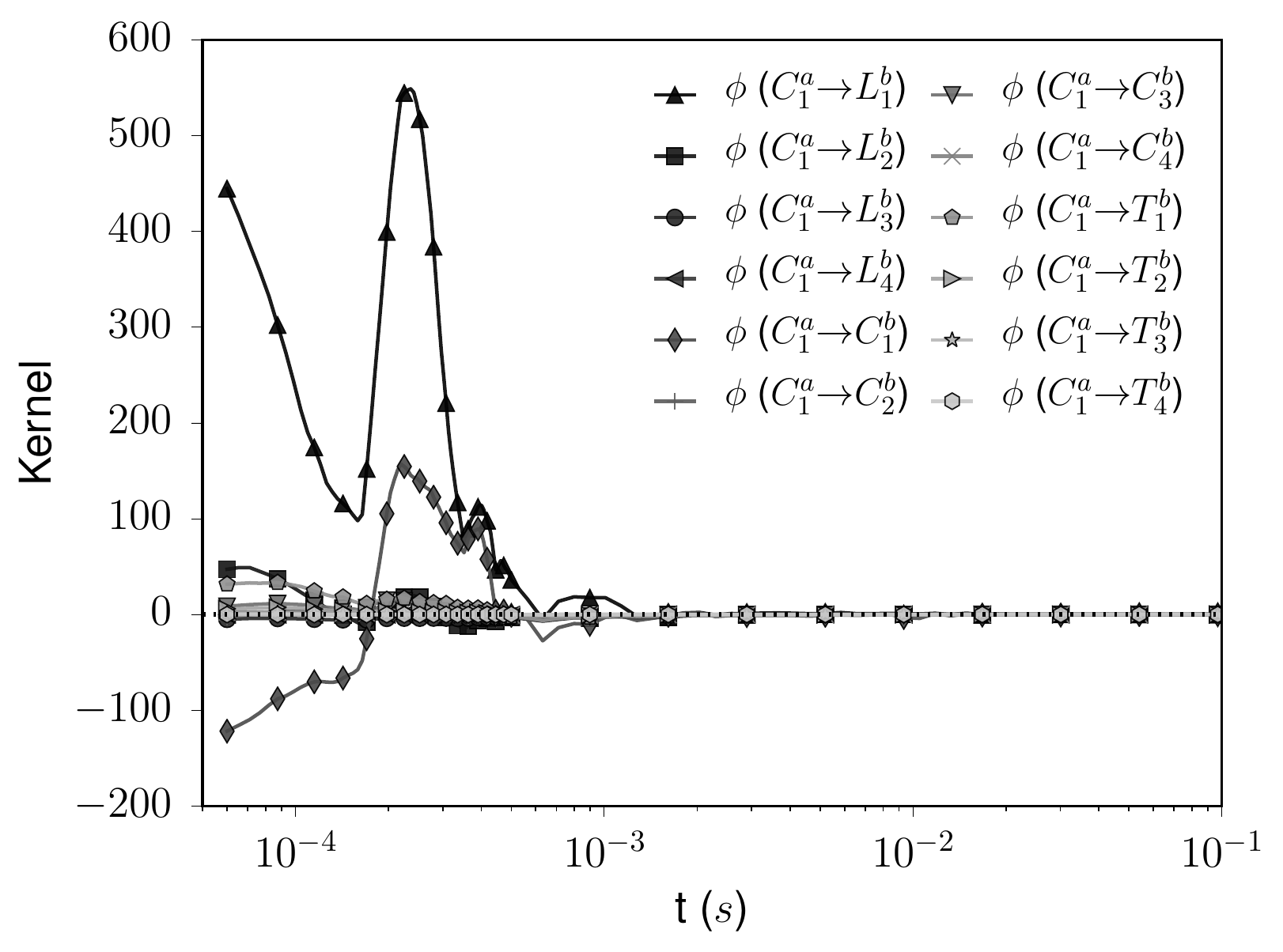}
\includegraphics[width=.48\textwidth]{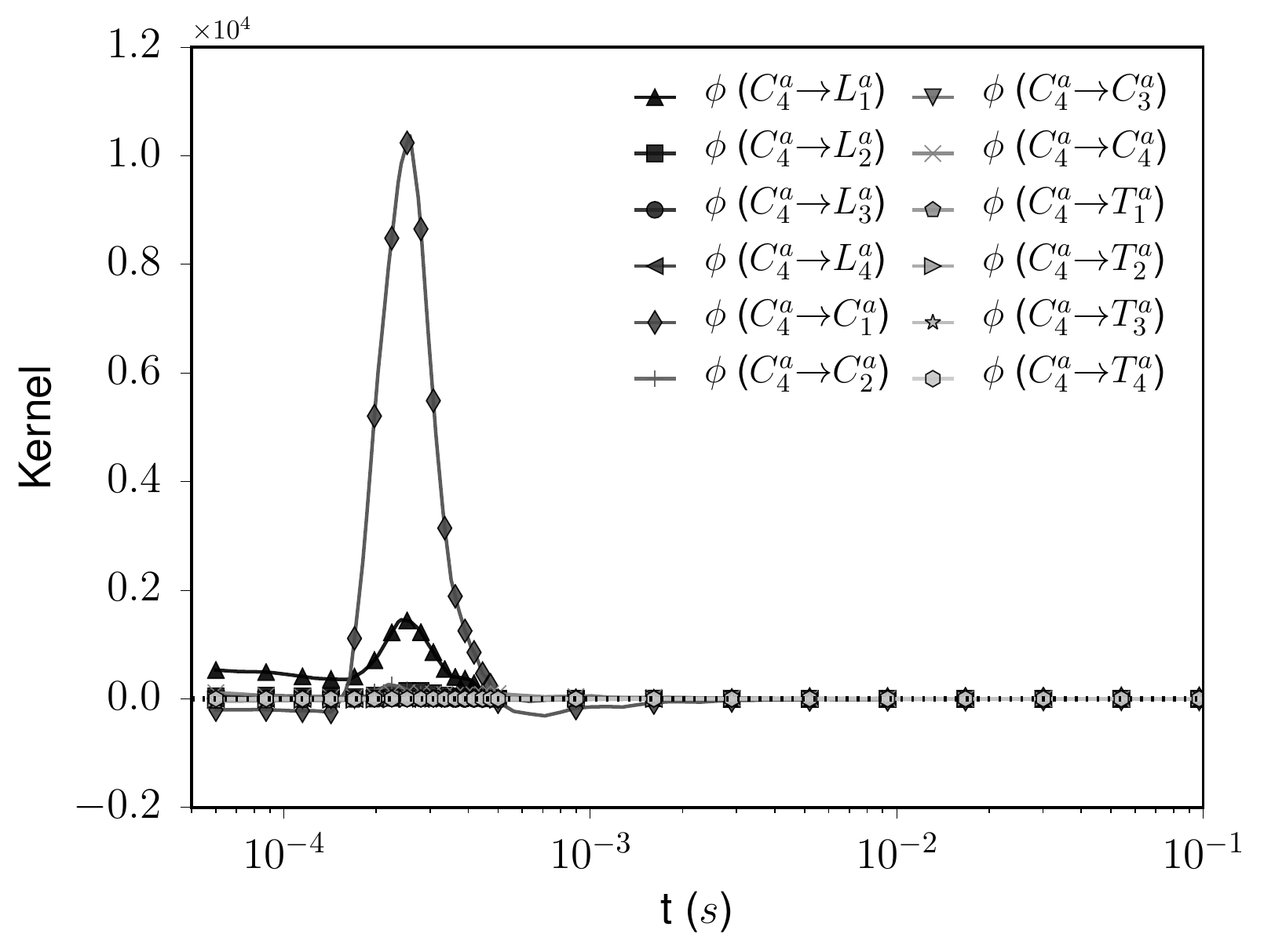}
\includegraphics[width=.48\textwidth]{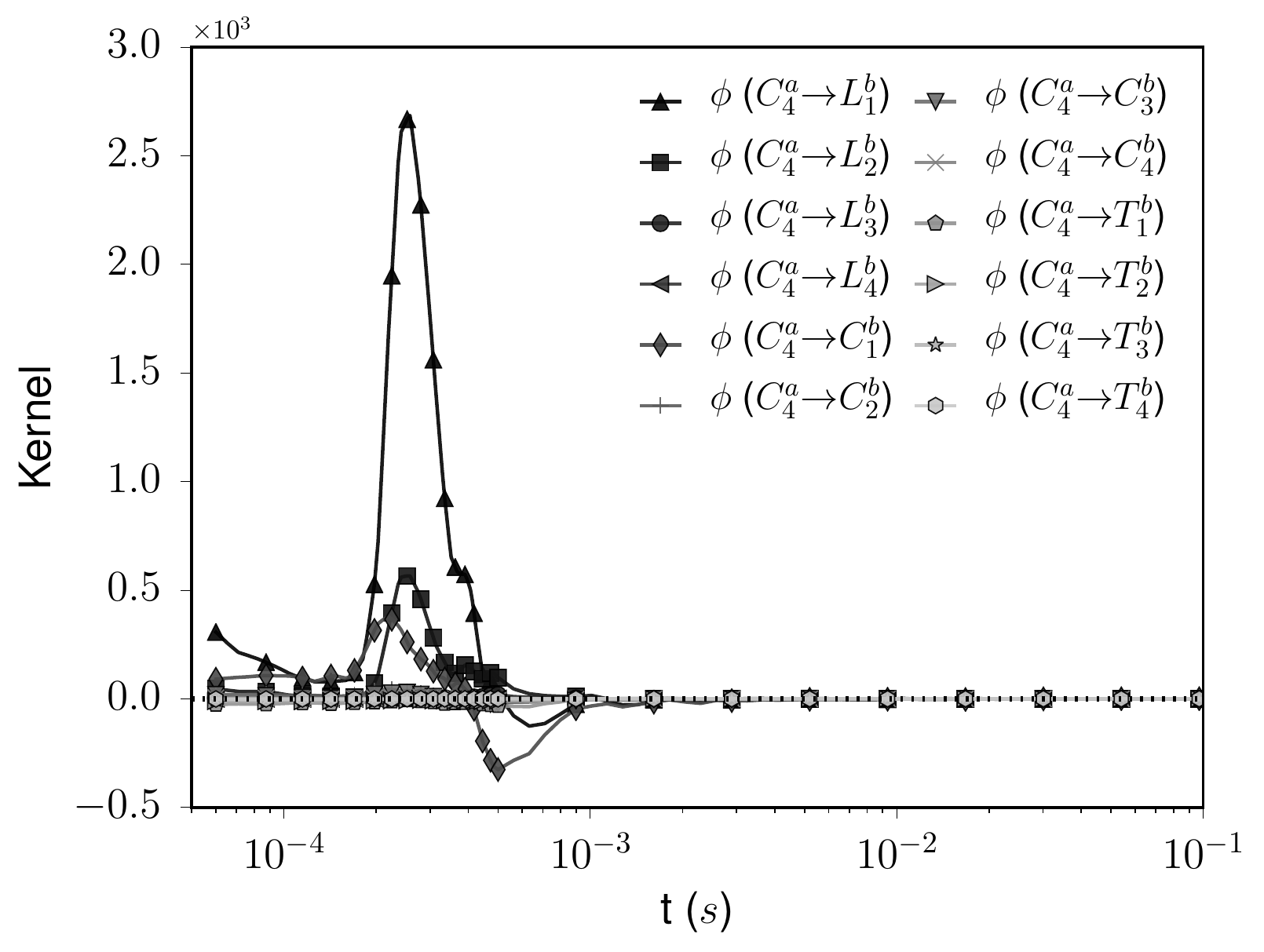}
\caption{DAX future: kernels describing the influence of a small (top) and a large (bottom) cancel order on other events. The left column refers to same side of the book, the right side to the opposite one.}
\label{fig:dax_cancel}
\end{figure}

\subsection{Effects of Trades}
\label{sub:ob_trades}

We now look at the effect of trades on other events. In Figures \ref{fig:dax_trade} and \ref{fig:bund_trade} the relative kernels are plotted for the DAX and the Bund futures, respectively. As we already pointed out, below the market reaction time we note the effect of order splitting, that manifests itself in the large positive values of the small-trade/small-trade kernel. In the case of the DAX this effect is absent for large trades, since there is no point in executing two or more large orders one immediately after the other. In the Bund case, where a 10 contracts order is still not so large we note some excitation from large trades to small ones at the shortest time scale. 

Another feature that emerges concerns the kernels $T^a\to C^a_1$. For both assets, these kernels are negative before the market reaction time, then they assume very large positive values around 300$\mu$s and then around $1m$s they revert to zero or to negative values in the case of large trades. We can give the following reading of this effect. When a large trade arrives, it will likely eat up a significant part of the outstanding liquidity. Of course, those limit orders that have been matched cannot be canceled anymore. Thus, the negative values of this kernel before the market reaction time are explained by this mechanical effect. The reaction of market participants with a market making strategy to a large trade is then to cancel their outstanding limit orders on the side hit by the trade, and to place limit orders on the opposite side. This is a consequence of the fact that a trade at the ask signals that the "true price" is closer to that side, and thus liquidity must adapt to the new situation.
We remark again that the effect of large trades on liquidity is both more intense (compare the different y-scales of the figures) and much more persistent in time. 

\begin{figure}[tbh]
\centering
\includegraphics[width=.48\textwidth]{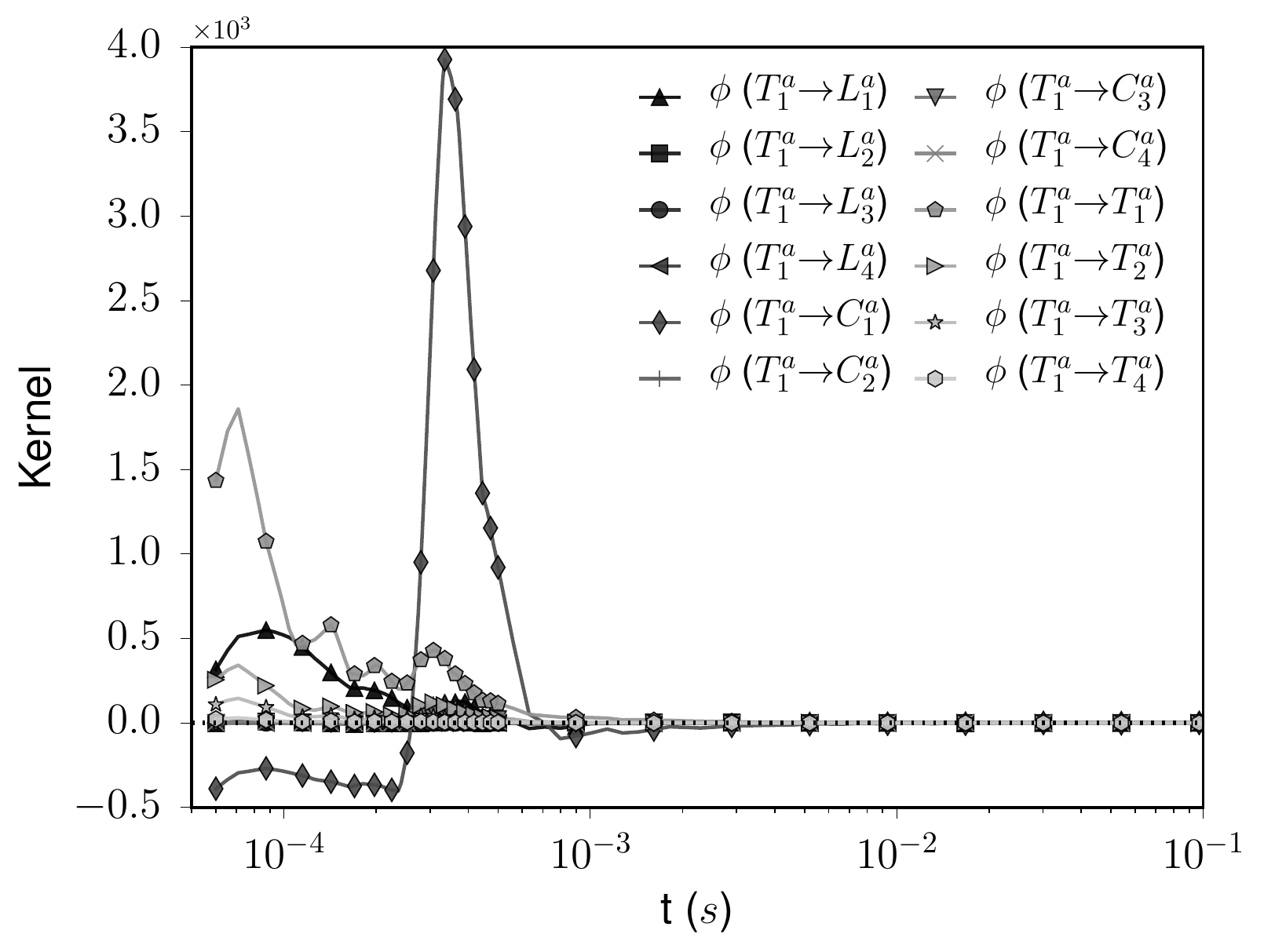}
\includegraphics[width=.48\textwidth]{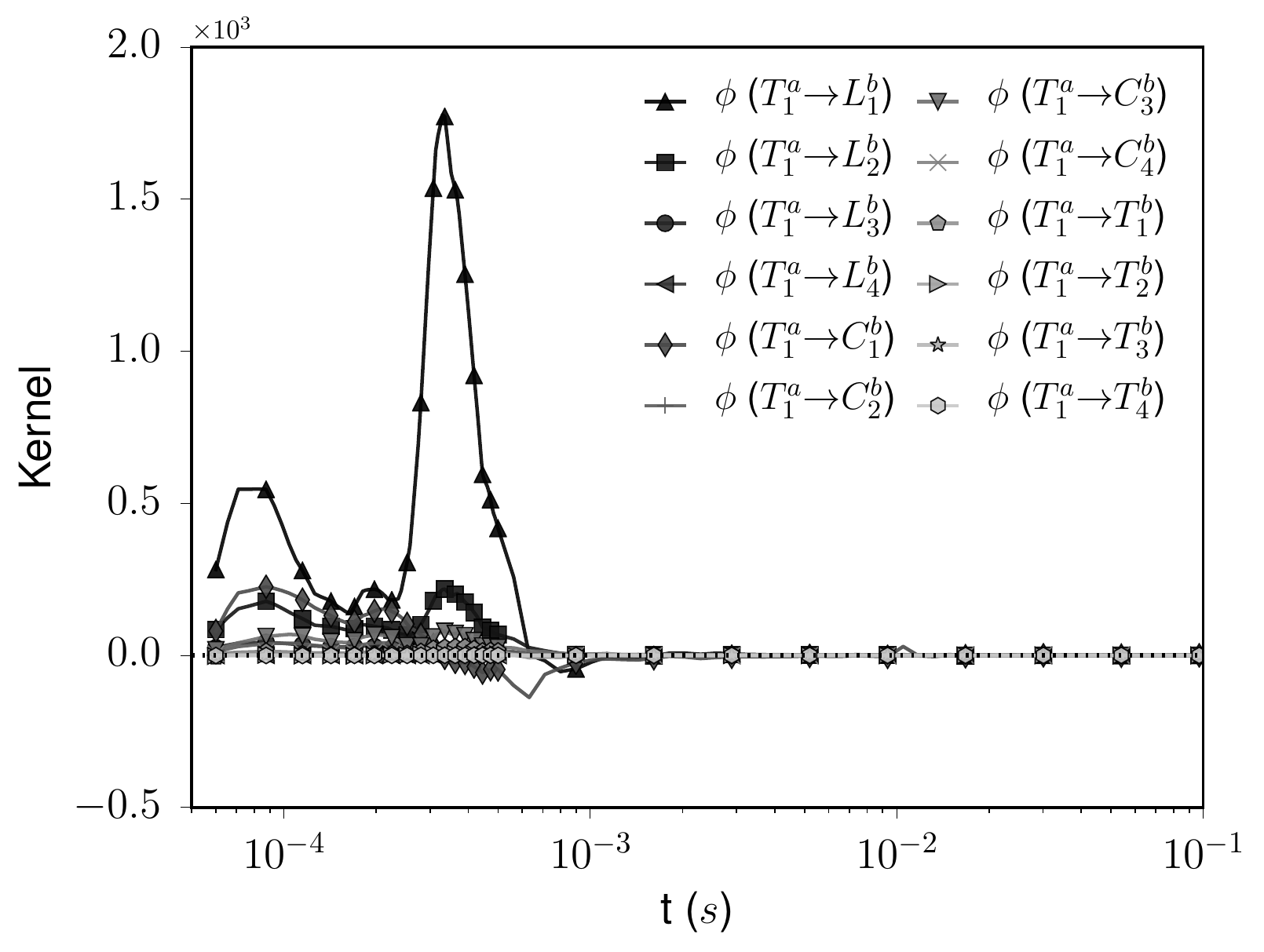}
\includegraphics[width=.48\textwidth]{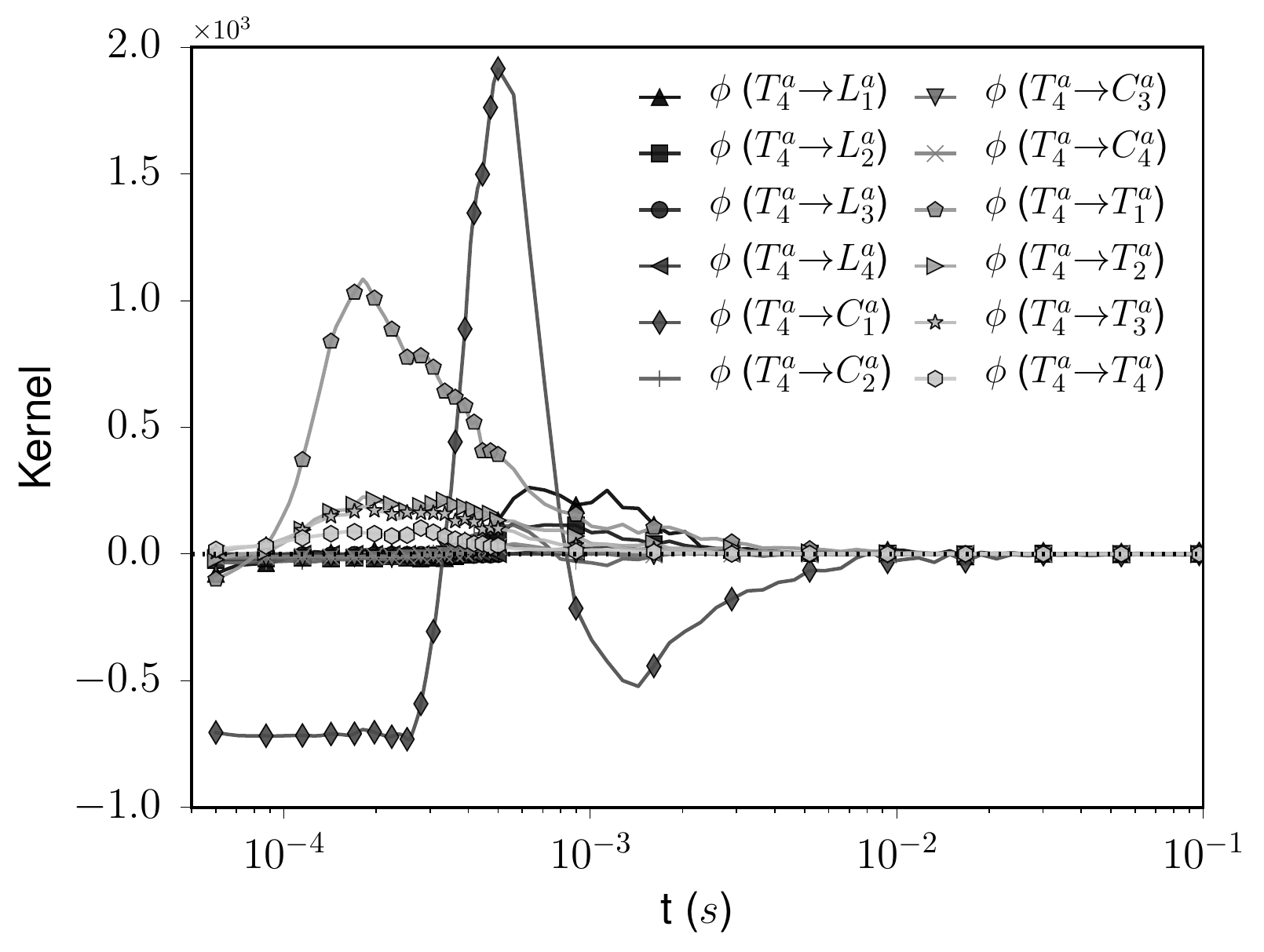}
\includegraphics[width=.48\textwidth]{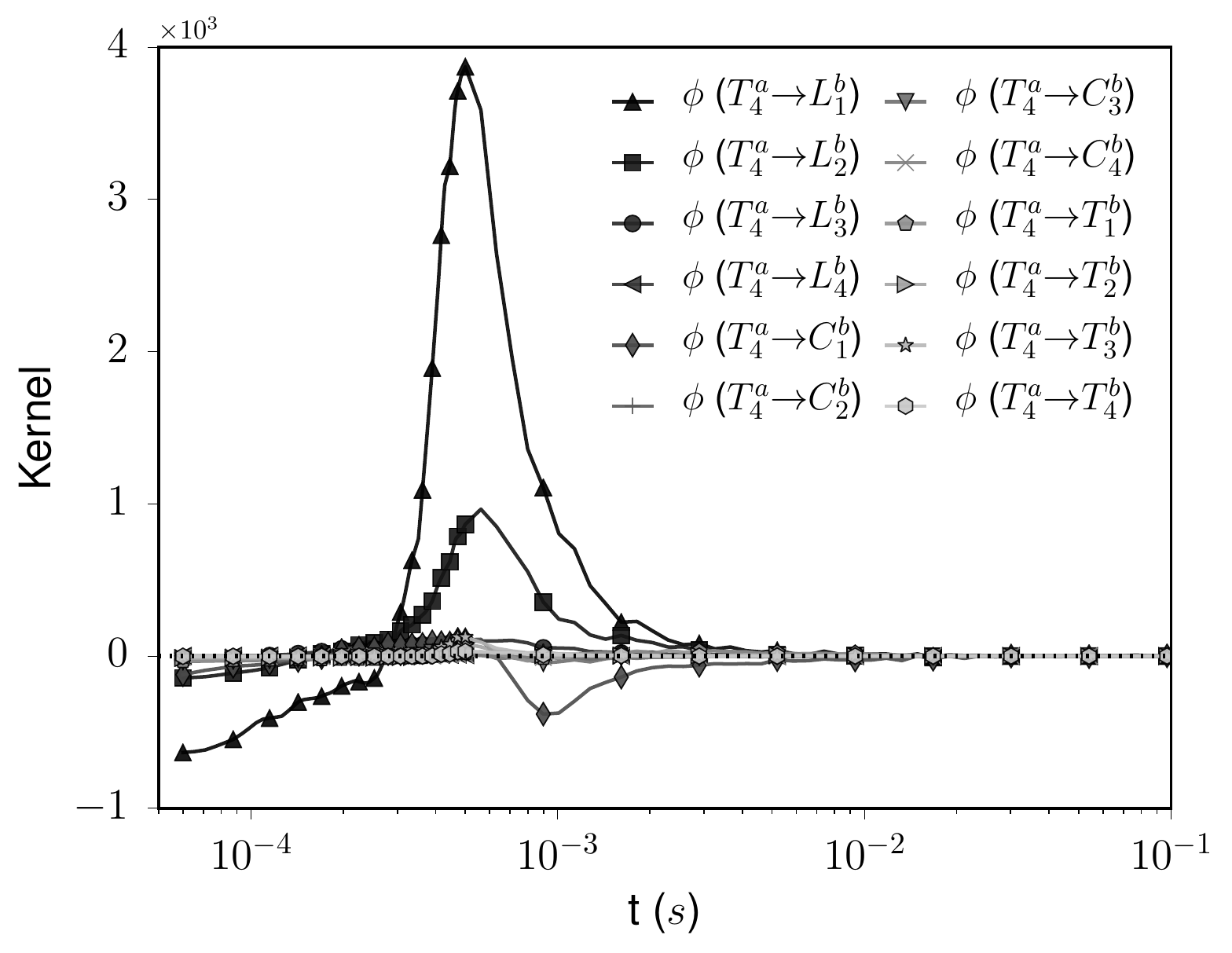}
\caption{DAX future: kernels describing the influence of a small (top) and a large (bottom) trade on other events. The left column refers to same side of the book, the right side to the opposite one.}
\label{fig:dax_trade}
\end{figure}

\begin{figure}[tbh]
\centering
\includegraphics[width=.48\textwidth]{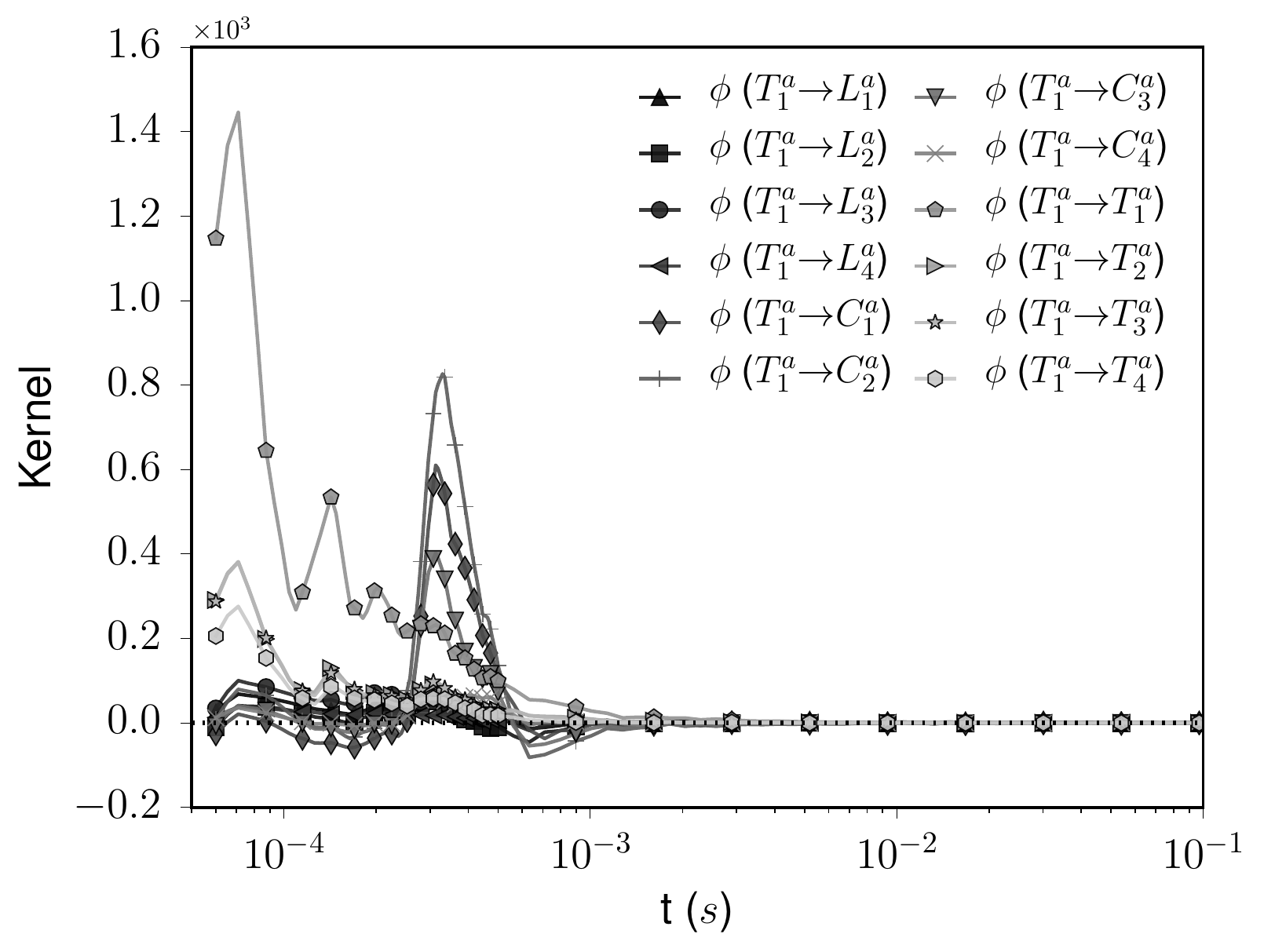}
\includegraphics[width=.48\textwidth]{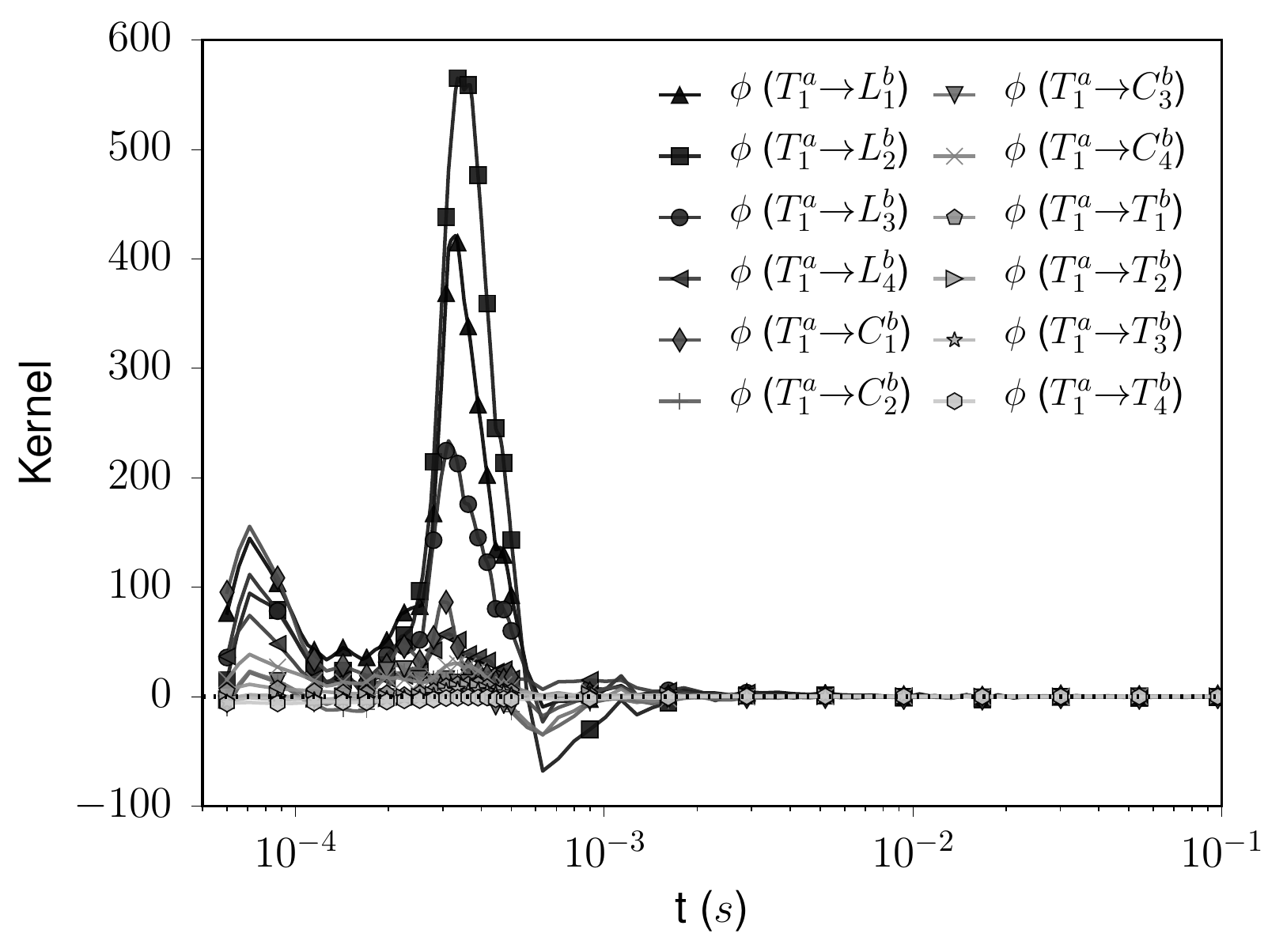}
\includegraphics[width=.48\textwidth]{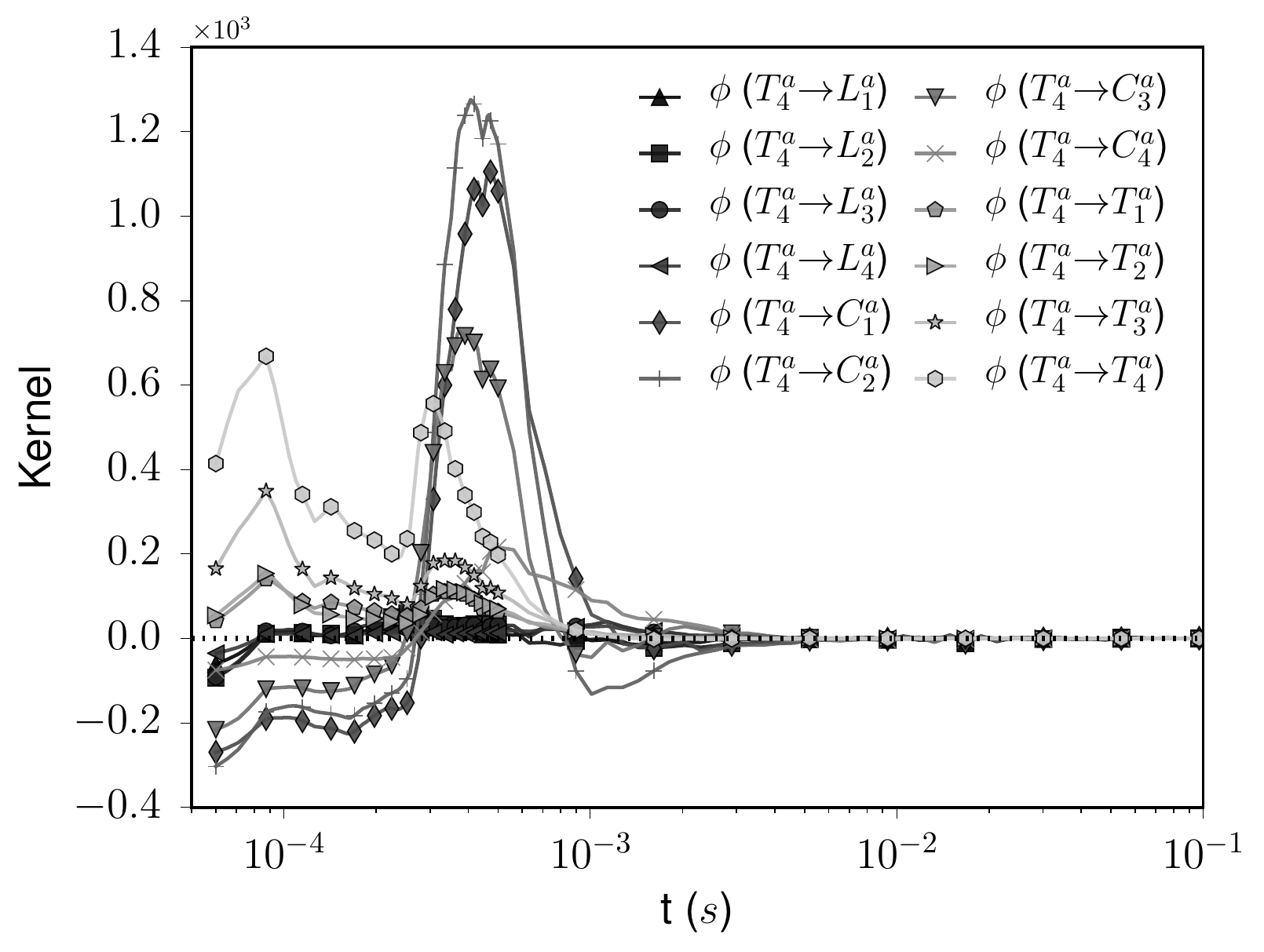}
\includegraphics[width=.48\textwidth]{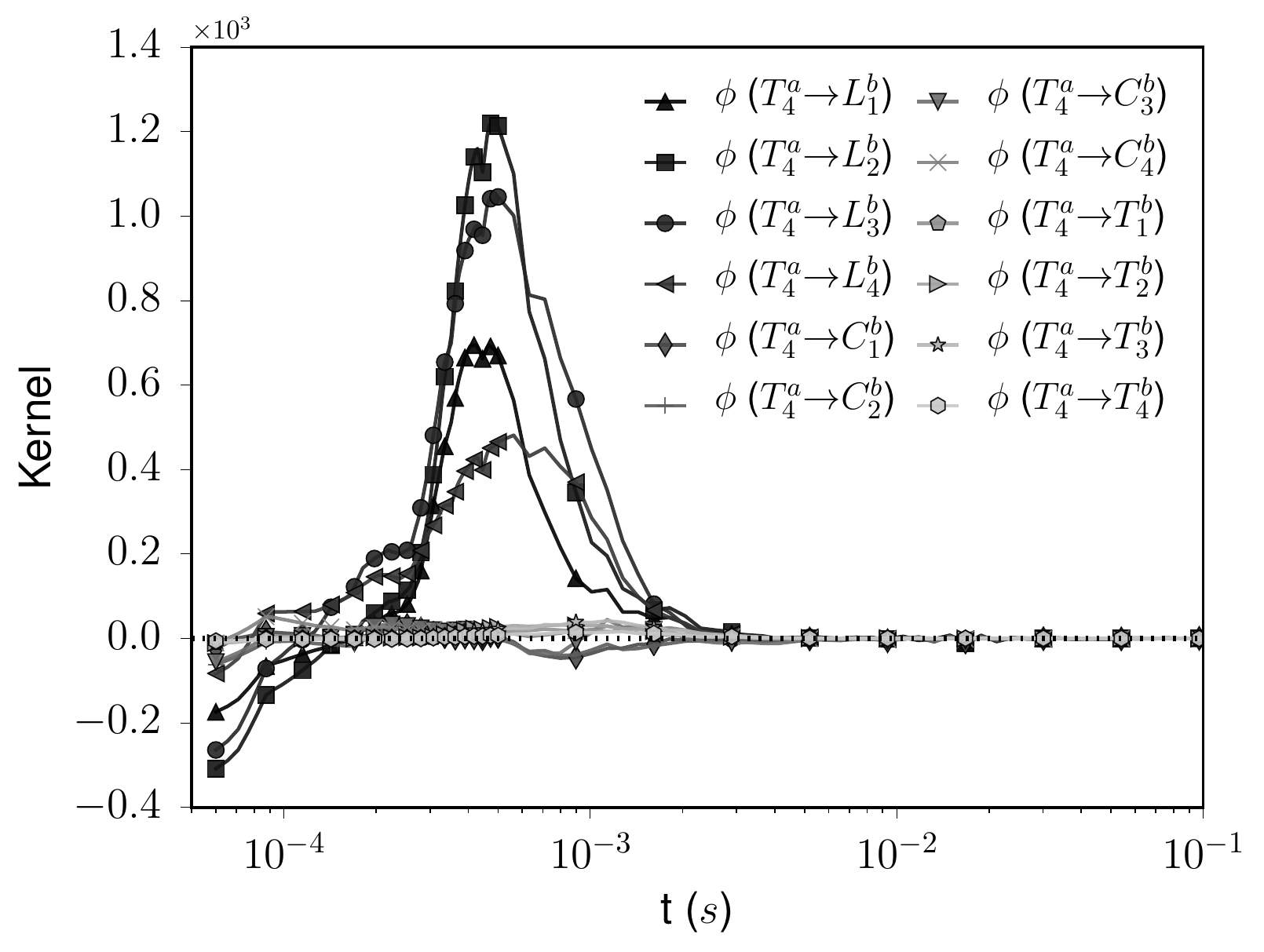}
\caption{Bund future: kernels describing the influence of a small (top) and a large (bottom) trade on other events. The left column refers to same side of the book, the right side to the opposite one.}
\label{fig:bund_trade}
\end{figure}

\subsection{Baseline intensities}

We complete the analysis of the order book section by examining the exogeneity of the different types and sizes of orders, measured by the baseline intensities $\mu_i$, we present the ratios $\mu_i/\Lambda_i$ in Table \ref{tab:mu_lam_ratio}. 

We can make three kind of considerations. The first concerns the degree of exogeneity of the different types of orders. Trades are the most exogenous, followed by limit orders. It is no surprise that the cancellations activity can be in large part explained by cross and self excitation, since cancel orders are mostly the reaction to some event. 
It is interesting to compare the ratios $\mu_i/\Lambda_i$ we obtained when examining trades alone with those of the full order book model. From Table \ref{tab:mu_comp} we note that, when including all the other orders into the picture the fraction of intensity explained by the baseline intensity decreases, albeit not dramatically. In fact, market orders have been shown to lead limit and cancels rather than be driven. Moreover, we have seen that most of the cross excitation comes from other trades, thus no big difference is expected between the full order book case and the trade-only case. Nevertheless, some trade activity is explained by the other events notably by large limit orders as we pointed out in the previous section.

The second consideration regards the effect of the order size. In line with the analysis conduced so far, we find large orders to be more exogenous than small ones. It is reasonable to assume that when a trader places a large order is because she has an "independent" reason to do so, either a clear view on the underling asset or a necessity to execute a large quantity immediately, rather than because she is following the market or engaging in market making. 

Finally, we observe that the ratios $\mu_i/\Lambda_i$ for large orders on the DAX are much higher than those on the Bund (see Tables \ref{tab:mu_lam_ratio} and \ref{tab:mu_comp}). This is mostly a consequence of the different volume distributions in the two assets. On the DAX future, orders larger than 10 are much more unusual, and perceived as big, than they are on the Bund.  
\begin{table}
\centering
\begin{tabular}{cccccccccccc}
\toprule[1pt]
\multicolumn{12}{c}{Bund}\\
\midrule
$L^a_1$ & $L^a_2$ & $L^a_3$ & $L^a_4$ & $C^a_1$ & $C^a_2$ & $C^a_3$ & $C^a_4$ & $T^a_1$ & $T^a_2$ & $T^a_3$ & $T^a_4$\\
20.6 &16.8 &13.0 &16.5 &9.0 &16.5 &13.3 &7.8 &27.3 &16.8 &17.4 &36.8 \\
\toprule
$L^b_1$ & $L^b_2$ & $L^b_3$ & $L^b_4$ & $C^b_1$ & $C^b_2$ & $C^b_3$ & $C^b_4$ & $T^b_1$ & $T^b_2$ & $T^b_3$ & $T^b_4$ \\
19.9 &16.5 &13.1 &17.1 &8.0 &16.5 &13.3 &7.6 &29.2 &17.0 &18.0 &36.4  \\
\toprule[1pt]
\multicolumn{12}{c}{DAX}\\
\midrule
$L^a_1$ & $L^a_2$ & $L^a_3$ & $L^a_4$ & $C^a_1$ & $C^a_2$ & $C^a_3$ & $C^a_4$ & $T^a_1$ & $T^a_2$ & $T^a_3$ & $T^a_4$\\
19.9 &19.1 &33.7 &72.3 &10.7 &10.8 &16.4 &18.5 &27.2 &30.5 &44.8 &50.6\\
\toprule
$L^b_1$ & $L^b_2$ & $L^b_3$ & $L^b_4$ & $C^b_1$ & $C^b_2$ & $C^b_3$ & $C^b_4$ & $T^b_1$ & $T^b_2$ & $T^b_3$ & $T^b_4$ \\
20.1 &20.0 &35.0 &71.2 &11.0 &11.5 &16.4 &18.4 &26.7 &30.5 &44.8 &49.6\\
\bottomrule[1pt]
\end{tabular}
\caption{Ratios $\frac{\mu_i}{\Lambda_i}$ expressed in percent.}
\label{tab:mu_lam_ratio}
\end{table}
\begin{table}
\centering
\begin{tabular}{cccccccc}
\toprule[1pt]
\multicolumn{8}{c}{Bund}\\
\midrule
$T^a_1$ & $T^a_2$ & $T^a_3$ & $T^a_4$ & $T^b_1$ & $T^b_2$ & $T^b_3$ & $T^b_4$\\
4.3&  4.3&  6.6 &10.7 &  2.4 & 3.6 & 5.6 & 10.4\\
\toprule[1pt]
\multicolumn{8}{c}{DAX}\\
\midrule
$T^a_1$ & $T^a_2$ & $T^a_3$ & $T^a_4$ & $T^b_1$ & $T^b_2$ & $T^b_3$ & $T^b_4$\\
 5.5&  7.3& 4.0 & 7.5 & 5.7 &  6.8 &  2.8 & 6.6\\
\bottomrule[1pt]
\end{tabular}
\caption{Difference in the ratios $\frac{\mu_i}{\Lambda_i}$ (expressed in percent) between the signed trade only case examined in Section \ref{sec:unsigned_vol} and the full order book case.}
\label{tab:mu_comp}
\end{table}

\section{Conclusions}
\label{sec:fin_rem}

In this paper we showed how multivariate Hawkes processes can be successfully applied to the modeling of the order book dynamics when order size is taken into account. We have shown that a simple multiplicative model, where time and size are factored, is not able to capture the complex interplay between these variables. On the contrary, our approach has proved to be capable of highlighting several features of the dynamics, from typical reaction time to complex interaction between different order type.

There are two main strengths of our approach as compared to the existing ones. First, Hawkes processes are point processes easily applicable in multiple dimensions and in this paper we have exploited this property by considering different types of orders (limit orders, cancellations, and trades) and different volume size. Second, we used a non parametric estimation method that limits the constraint imposed by traditional parametric models.  

When volume is not considered we recovered many of the results of \cite{bacry2014estimation} for limit order book dynamics. Moreover by using the multivariate approach we were able to separate contribution from different order sizes. The different impact of large orders clearly emerge from our analysis, as well as the longer persistence of their effect. Our work thus shows that the role of order size needs to be taken into account for a complete understanding of the order book dynamics. Our study still leave out some relevant information such as the size of the queue and the dynamics of the price and this will be the objective of future works.

We conclude by noting that our methodology is well suited to the application to different systems where the interaction of different event types as well as their marks are relevant.

\section*{Acknowledgments}

This research benefited from the support of the Chair Markets in Transition, under the aegis of Louis Bachelier Finance and Sustainable Growth laboratory, a joint initiative of Ecole Polytechnique, Université d’Evry Val d’Essonne and Fédération Bancaire Française. 
and from the chair of the Risk Foundation: Quantitative Management Initiative. We thank Thibault Jaisson, Iacopo Mastromatteo and Jean-François Muzy for useful discussions.

\appendix

\section{Proofs}
\label{sec:proof}
Here we report the proof of Proposition \ref{prop:2}. 
\begin{proof} 

Let us suppose that $\mathbf{\phi}(t) \geq 0$ for all $t \geq 0$, meaning that all the elements are positive functions. Define, for every $t\geq 0$, $S_{kj}^- = \lbrace s \in \mathbb{R^+} : g^{kj}(t-s) <0 \rbrace \neq \emptyset$  and $S_{kj}^+ = \lbrace s \in \mathbb{R^+} : g^{kj}(t-s) \geq 0 \rbrace$. 

Then we rewrite the integral equation \eqref{eq:wiener-hopf} as:
\begin{equation}
\begin{split}
g^{ij}(t)&=\phi^{ij}(t)+\sum_k\int_0^\infty g^{kj}(t-s) \phi^{ik}(s) \diff s\\
& = \phi^{ij}(t) +\sum_k\int_{S_{kj}^+} g^{kj}(t-s) \phi^{ik}(s) \diff s +\sum_k \int_{S_{kj}^-} g^{kj}(t-s) \phi^{ik}(s) \diff s
\end{split}
\end{equation}
Rearranging the terms we get:
\begin{equation}
g^{ij}(t)-\sum_k \int_{S_{kj}^-} g^{kj}(t-s) \phi^{ik}(s) \diff s = \phi^{ij}(t) +\sum_k\int_{S_{kj}^+} g^{kj}(t-s) \phi^{ik}(s) \diff s 
\end{equation}
The right hand side is always positive, so for the equation to hold we need, for all $t$ such that $g^{ij}(t)<0$:
\begin{equation}
g^{ij}(t)\geq \sum_k \int_{S_{kj}^-} g^{kj}(t-s) \phi^{ik}(s) \diff s
\end{equation}
Let us define 
\begin{equation}
m_{kj}= \inf_{\mathbb{R}^+} g^{kj}(t)
\end{equation}
and 
\begin{equation}
m_j = \min_{k} m_{kj}
\end{equation}
Now we can write
\begin{equation}
g^{ij}(t)\geq \sum_k \int_{S_{kj}^-} g^{kj}(t-s) \phi^{ik}(s) \diff s \geq \sum_k m_{kj} \int_{S_{kj}^-} \phi^{ik}(s) \diff s \geq m_j  \sum_k \int_{S_{kj}^-} \phi^{ik}(s) \diff s
\end{equation}
if we take this equation for the component $(i^*,j)$ such that $m_j= m_{i^*j}$, that is the $g$ that reaches the lowest value on the column $j$, we find that the above conditions requires
\begin{equation}
\sum_k \int_{S_{kj}^-} \phi^{i^*k}(s) \geq 1 \Rightarrow \sum_k \int_0^{\infty} \phi^{i^*k}(s) = \sum_k \| \phi^{i^*k} \|_1 \geq 1 
\end{equation} 
since the $\phi$ are positive and causal functions. 
If the matrix $\mathbf{g}(t)$ has at least one entry on every column (or row) wich takes negative values, then the minimum row (column) sum of $\mathbf{\hat{\phi}}_0$, $r$ ($c$) must be greater than one. This implies that also the spectral radius, $\rho$, of $\mathbf{\hat{\phi}}_0$ is greater than one (see \citet{minc1988nonnegative} pp. 24-25).

\end{proof}

\section{Robustness checks}
\label{app:b}
%\subsection{Timestamps randomization}
In this appendix we briefly discuss two possible issues that can arise when dealing with real data, namely the effects of finite time resolution and those of the intraday activity pattern. We find that none of these issues appears to significantly influence the results of our analysis. 

\paragraph{Timestamps randomization}
As noted in the main text, in our data we never observe two trades closer in time than 50 $\mu$s. This could be the result of some technical time required to process the order and write the data in the database. For changes in the best quote the minimum time difference observed is 10 $\mu$s, despite the timestamps being recorded with microsecond precision. So there is the possibility that actual events are sometimes closer in time to each other than they result from the database. 

We check the outcome of the non-parametric estimation also with randomized data. To randomize the data, we first round all the trades timestamps to the nearest 10$\mu s$.
Then, we subtract to them a random number uniformly distributed in $[0, 50)\, \mu s $. The effect of randomization on short durations are shown in Figure \ref{fig:dur_pre_after}.
\begin{center}
\begin{figure}[tb]
\centering
\includegraphics[width=0.49\textwidth]{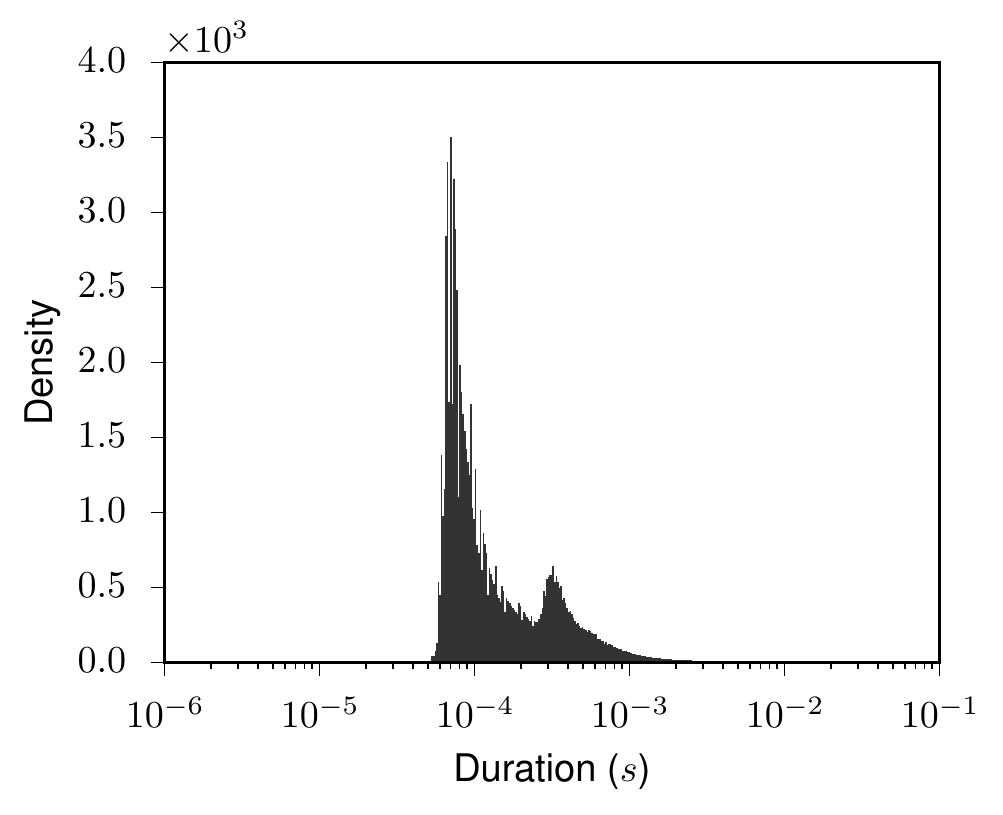}
\includegraphics[width=0.49\textwidth]{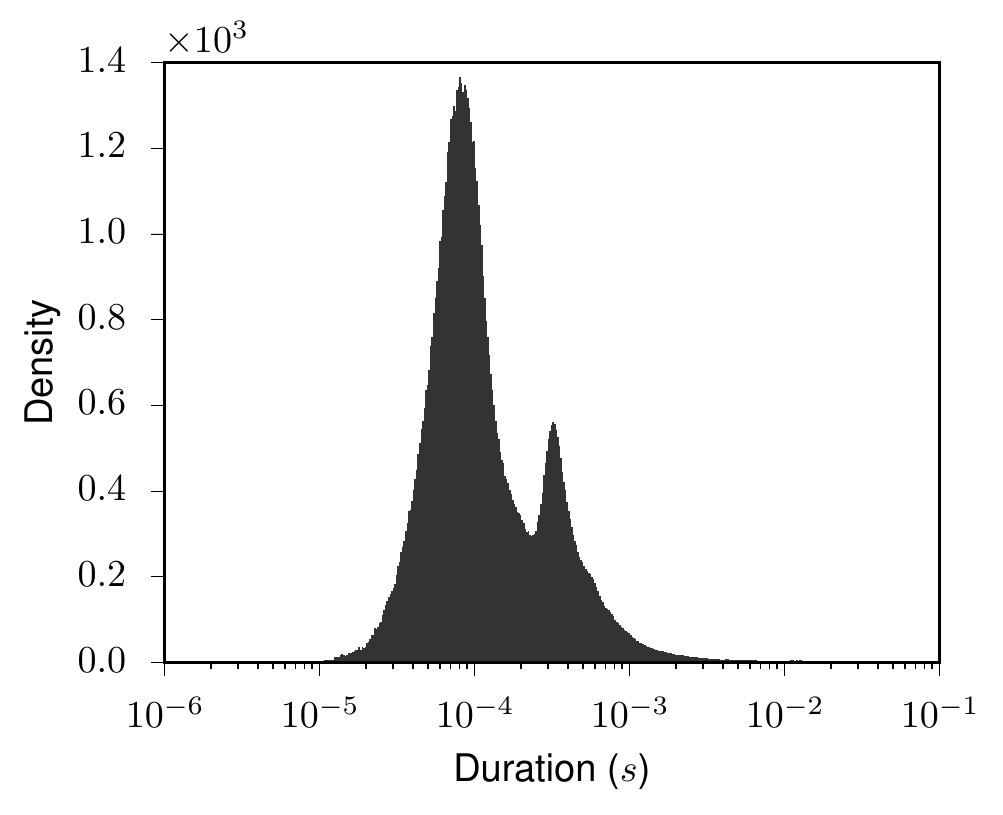}
\caption{Bund future: Inter events distribution before (left) and after (right) randomization.}
\label{fig:dur_pre_after}
\end{figure}
\end{center}
The estimation parameters are the same as those used in Section \ref{sec:unsigned_vol}. In Figure \ref{fig:6D_rand_norm} (left) the kernels norms are presented, we note that the outcome is almost identical to the non randomized case (see Fig. \ref{fig:dax_norm}). We find also that the shape of the obtained kernels is essentially the same,  except that the peaks at 100 and 300$\mu s$ are broader as a result of the broader distribution of the inter event times due to randomization. We conclude that the time resolution of our data is fine enough not to pose noticeable problem to the estimation procedure and we therefore prefer to use the original timestamps.
\begin{center}
\begin{figure}[tb]
\centering
\includegraphics[width=0.49\textwidth]{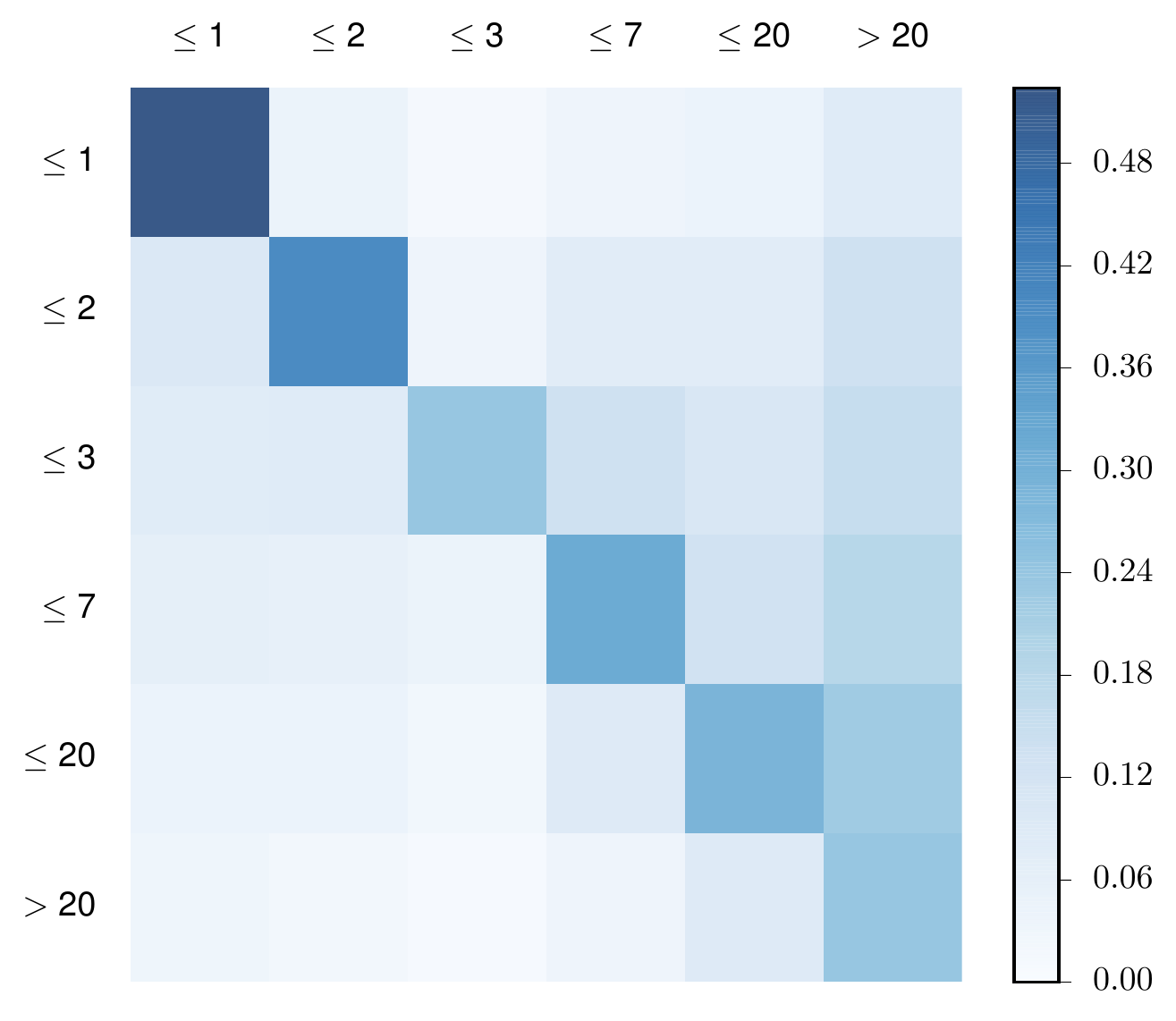}
\includegraphics[width=0.49\textwidth]{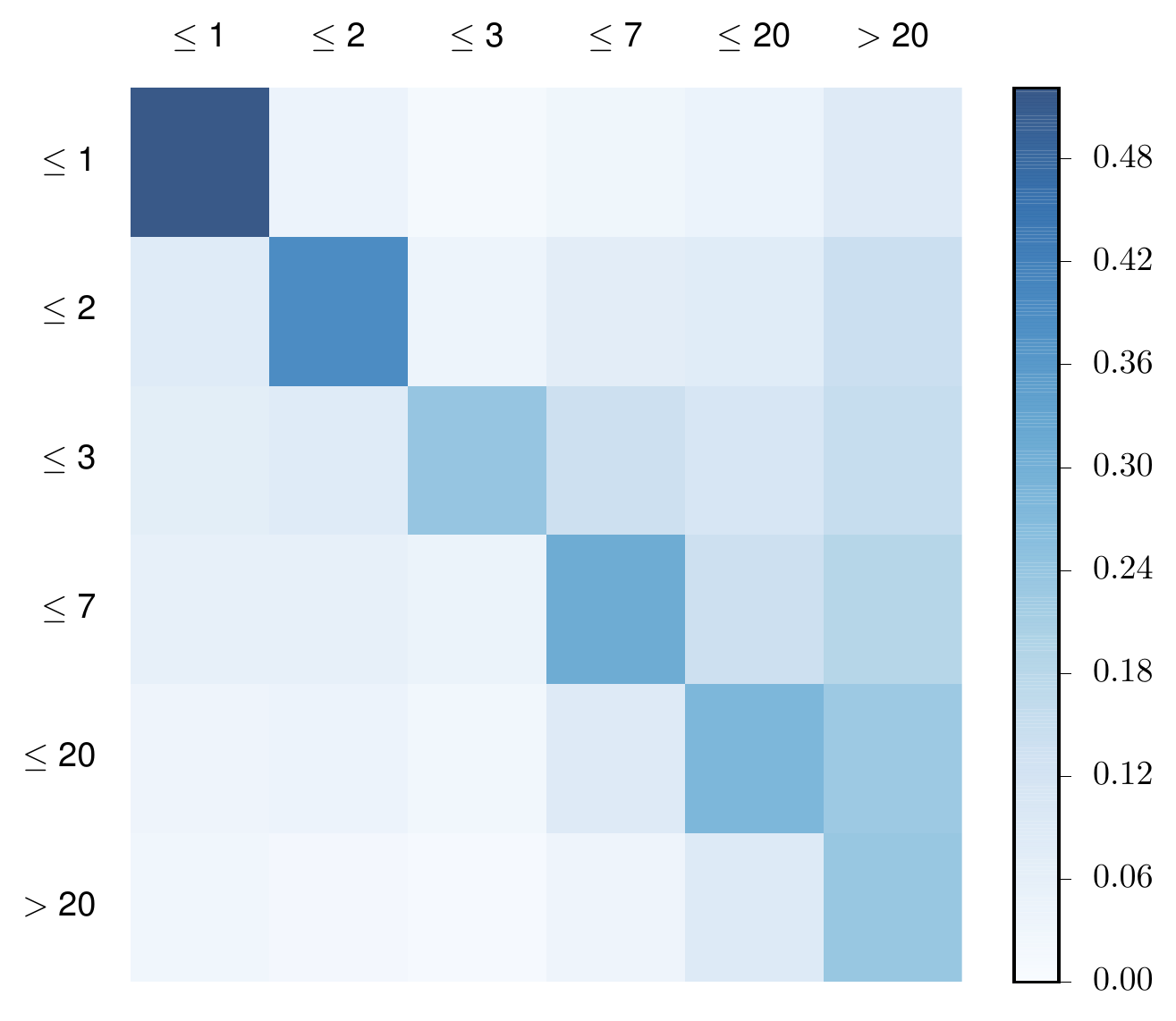}
\caption{Bund future: rescaled kernel norms for the six dimensional case. Results from randomized data (left), and results form trades occurred only between 11:00 and 17:00 local time (right).}
\label{fig:6D_rand_norm}
\end{figure}
\end{center}

\paragraph{Intraday Seasonality}
In principle intraday non stationarities could pose a problem for estimation. To check for daily-pattern effects on the estimation procedure we repeat the estimation on trade events comprised between 11:00 and 17:00 Frankfurt time. We consider the same six bins of volume as specified in Table \ref{tab:bin6D}. The matrix of the norms is reproduced in Figure \ref{fig:6D_rand_norm} right. We also computed the relative differences
$$ \frac{n_{ij}-n'_{ij}}{n_{ij}}$$
Where, $n_{ij}$ is the rescaled norm of $\phi_{ij}$ estimated on the whole period (08:00-22:00), while $n'$ is the analogous quantity estimated on the restricted period. We find that the differences are of the order of a few percent, with the highest differences (about 10\%) appearing on the smallest norms where we also have the highest estimation error. 
We therefore conclude that our main results are robust also with respect to seasonality effects. This can be linked to the fact that our study focuses only on the very short time scale dynamics (we estimate the kernels only up to 0.5$s$).

\bibliography{bibRambaldiBacryLillo}
\bibliographystyle{chicago}

\end{document}